\documentclass[preprint,12pt,review,authoryear]{elsarticle}
\usepackage[T1]{fontenc}
\usepackage[latin9]{inputenc}
\usepackage[margin=1in]{geometry}
\usepackage{verbatim}
\usepackage{amsmath}
\usepackage{amsthm}
\usepackage{amssymb}
\usepackage[authoryear]{natbib}
\usepackage{hyperref}
\usepackage[capitalise]{cleveref} 
\crefname{figure}{Figure}{Figures}
\usepackage{graphicx}
\usepackage{multirow}
\usepackage{adjustbox}

\usepackage[colorinlistoftodos,textsize=tiny
			]{todonotes}

\makeatletter
\theoremstyle{plain}
\newtheorem{thm}{\protect\theoremname}
\theoremstyle{plain}

\theoremstyle{remark}

\theoremstyle{remark}
\newtheorem*{rem*}{\protect\remarkname}
\theoremstyle{definition}

\theoremstyle{definition}

\newtheorem{lem}{Lemma}
\newtheorem{cor}{Corollary}

\newcommand{\zq}[1]{}
\makeatother

\usepackage{babel}
\providecommand{\examplename}{Example}
\providecommand{\problemname}{Problem}
\providecommand{\propositionname}{Proposition}
\providecommand{\remarkname}{Remark}
\providecommand{\theoremname}{Theorem}
\usepackage{setspace}

\journal{Journal of Financial Economics}

\begin{document}
\begin{frontmatter}
\title{Handling Sparse Non-negative Data in Finance\tnoteref{label1}}
\tnotetext[label1]{We are grateful for stimulating conversations with Charles-Albert Lehalle, Garud Iyengar, Han Hong, Yanguang Liu, and Johan Ugander, and to Ron Bekkerman and Davide Proserpio for sharing the dataset on residential real estate investment, as well as for technical support from Moody's Corporation for access to their Default and Recovery Database (DRD).}


\author{Agostino Capponi} 
\ead{ac3827@columbia.edu}
\affiliation{organization={Department of Industrial Engineering and Operations Research, Columbia University},
            city={New York},
            state={NY},
            country={USA}
            }
\author{Zhaonan Qu} 
\ead{zq2236@columbia.edu}
\affiliation{organization={Data Science Institute, Columbia University},
            city={New York},
            state={NY},
            country={USA}
            }
\author{{\href{https://drive.google.com/file/d/1qJS-sRlTLggwkjiYLQr0ihGnYKch6sUU/view?usp=sharing}{\textcolor{red}{\large Latest Version Here}}}}
\begin{abstract}
\normalfont
We show that Poisson regression, though often recommended over log-linear regression for modeling count and other non-negative variables in finance and economics, can be far from optimal when heteroskedasticity and sparsity---two common features of such data---are both present. We propose a general class of moment estimators, encompassing Poisson regression, that balances the bias-variance trade-off under these conditions. A simple cross-validation procedure selects the optimal estimator. Numerical simulations and applications to corporate finance data reveal that the best choice varies substantially across settings and often departs from Poisson regression, underscoring the need for a more flexible estimation framework.

\end{abstract}

\end{frontmatter}

\section{Introduction}
\label{sec:intro}
Data with non-negative outcome (dependent) variables are ubiquitous in applications across a diverse range
of disciplines, including finance \citep{coles2006managerial,hirshleifer2012overconfident,akey2021limits}, economics \citep{hausman1984econometric,card2001estimating}, health services \citep{dow2003choosing}, epidemiology \citep{holfold1980rates,bohning1999zero}, political science \citep{king1989event}, 
and sociology \citep{boulton2018analyzing}.

Specific examples include discrete count data, such as the number of corporate defaults and patents, and continuous count-like data, such as health expenditures and sales revenue. A popular approach to modeling non-negative data, particularly in finance and economics, is regressing the log transformation of the outcome variable (often shifted by a positive constant such as 1) on explanatory variables. 

Despite the popularity of this log-linear regression approach, there is increasing evidence that it may not be the best approach to modeling non-negative data. For example, \citet{cohn2022count,chen2024logs} point out the lack of meaningful interpretation of log-linear estimates, while \citet{king1988statistical,silva2006log,silva2011further,mullahy2022transform} demonstrate the fragility of log-linear regression in the presence of \emph{heteroskedasticity} or \emph{sparsity}, two prominent features of count and non-negative data. As a solution, recent works have advocated for the use of \emph{pseudo maximum likelihood} (PML) estimation \citep{gourieroux1984pseudoa} based on count models, most notably the Poisson regression \citep{wooldridge1999quasi}.\footnote{In the econometrics literature, some works use the term ``quasi-maximum likelihood''  \citep{white1982maximum,wooldridge1999quasi}. In this paper, we follow the terminology of \citet{gourieroux1984pseudoa} and \citet{silva2006log} and use ``pseudo maximum likelihood''.}  Compared to log-linear regression, Poisson regression possesses several desirable properties that facilitate its application in practice, such as interpretability, robustness, computational efficiency, and the ability to accommodate separable group fixed effects and instrumental variables \citep{windmeijer1997endogeneity,mullahy1997instrumental,correia2020fast,cohn2022count,chang2024inferring}. These considerations have led to a shift towards Poisson regression in economics and finance \citep{sautner2023firm,addoum2023temperature, hollingsworth2024gift}. 

Importantly, estimates of model parameters can change significantly when switching from a log-linear to a Poisson specification. 
For example, \citet{cohn2022count} find 
\begin{quote}
  ``... through replication of
existing papers that economic conclusions can be highly sensitive to the regression model
employed. ... the choice between Poisson and log-linear regression typically has a larger effect on estimates than omitting the
most important control variable in real-world applications.''  
\end{quote}
Their findings highlight the importance of selecting an appropriate regression model in empirical research, including corporate finance studies on topics such as innovation \citep{he2013dark,fang2014does} and industrial pollution \citep{akey2021limits,xu2022financial}. 
More broadly, a shift from log-linear to Poisson regression may carry significant implications for empirical work across finance and economics. This naturally raises a fundamental question: should Poisson regression be the {\it default choice} when analyzing non-negative, count-like data?

Existing works have only partially addressed this question by focusing on {\it either} the heteroskedasticity {\it or} the sparsity of the data. For example, the efficiency properties of Poisson and other PML estimators under heteroskedasticity are well-known and depend on the \emph{dispersion} of data \citep{gourieroux1984pseudob,cameron2013regression}. On the other hand, when data exhibits high levels of sparsity, \citet{silva2006log,silva2011further} provide simulation evidence on the robustness of Poisson PML compared to log-linear regression and other PML estimators. Alternatives such as zero-inflated \citep{lambert1992zero} and hurdle models \citep{mullahy1986specification} for sparse non-negative data have also been studied. 

However, as is well-known to empirical researchers and also illustrated by our examples in \cref{sec:empirical}, non-negative data in many finance and economics applications frequently exhibit \emph{both} sparsity \emph{and} heteroskedasticity. Yet, few studies address these characteristics simultaneously, and it remains unclear whether Poisson regression should be used as a rule of thumb when handling sparse non-negative data. This major gap motivates the central question of our paper: How can we effectively model non-negative data when \emph{both} sparsity and heteroskedasticity are prominent?

\begin{table}
\begin{centering}
\begin{tabular}{c|c|c|c|c|c|c}
 & \multicolumn{3}{c|}{first coordinate of $\theta_0$} & \multicolumn{3}{c}{second coordinate of $\theta_0$}\tabularnewline
\hline 
 & $\alpha=0$ & $\alpha=1$ & $\alpha=2$ & $\alpha=0$ & $\alpha=1$ & $\alpha=2$\tabularnewline
\hline 
 optimal $\kappa$ & 1 & 0.78 & 0.11 & 1 & 0.9 & 0.42\tabularnewline
\hline 
\hline 
RMSE of Poisson estimator & 0.045 & 0.046 & 0.064 & 0.12 & 0.12 & 0.131\tabularnewline
\hline 
RMSE of optimal estimator & 0.017 & 0.021 & 0.058 & 0.03 & 0.047 & 0.105\tabularnewline
\hline 
\% of improvement & 62\% & 54\% & 9\% & 75\% & 61\% & 20\%
\end{tabular}
\par\end{centering}
\caption{Root mean squared errors (RMSE) of optimal vs. Poisson ($\kappa=0$) estimators under different levels of heteroskedasticity
(indexed by $\alpha$) and a fixed level of sparsity (see \cref{sec:experiments} for details).}
\label{tab:RMSE-simulation}
\end{table}

We address this question by developing a novel unified framework that integrates classical pseudo maximum likelihood methods, including Poisson regression, which so far have often been applied in an ad hoc manner.
As a main contribution of this paper, we show that the optimal modeling choice depends on the relative prominence of each feature in the data, and could diverge considerably from the Poisson estimator. Furthermore, we provide a simple data-driven procedure that empirical researchers can employ to determine the optimal choice. 

The key component of our framework is a novel class of moment estimators that includes classical PML estimators for non-negative data, such as Poisson and gamma, as special cases. These estimators place heterogeneous weights on observations in the estimating equations. On one hand, heteroskedasticity encourages larger weights on samples with smaller conditional means to improve \emph{efficiency}. On the other hand, the presence of many zeros discourages such weighting to reduce \emph{bias}. As a result, there exists an estimator that optimally balances the fundamental bias-variance trade-off created by these two aspects of the data. 
As illustrated in 
\cref{fig:phase-transition-simple} and \cref{tab:RMSE-simulation}, the optimal estimator is sensitive to the level of sparsity and heteroskedasticity, and can improve significantly upon the Poisson PML estimator (see \cref{sec:experiments} for full details). Notably, although non-linear least squares (NLS) has largely been avoided due to its inefficiency under heteroskedasticity, our simulation results suggest that it could be preferable when heteroskedasticity is mild compared to sparsity. Moreover, this trade-off is a {\it unique} feature of data with sparse non-negative outcome variables, setting it apart from the conventional bias-variance trade-off in statistics and machine learning, which is typically driven by model complexity.

\begin{figure}[t]
\begin{centering}
\includegraphics[scale=0.52]{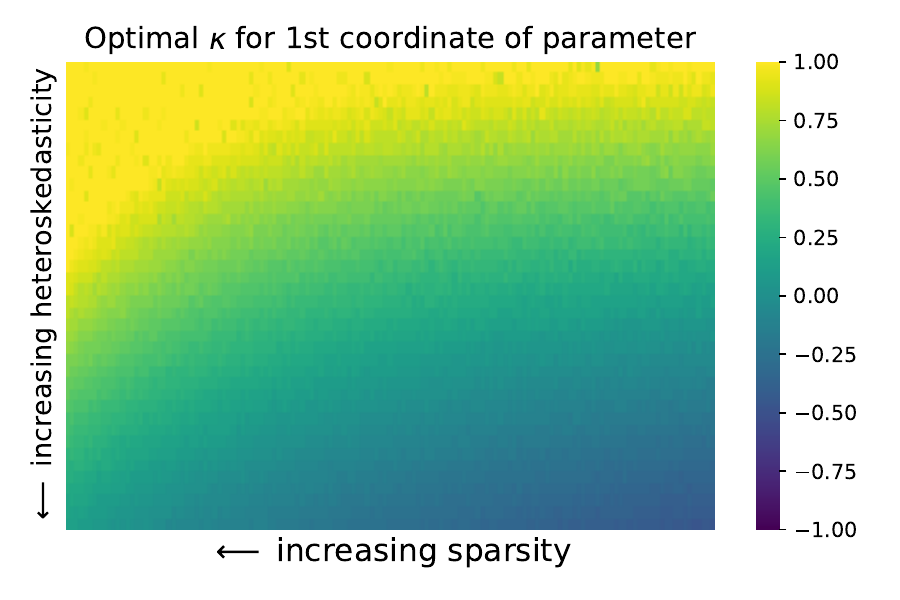}
\includegraphics[scale=0.52]{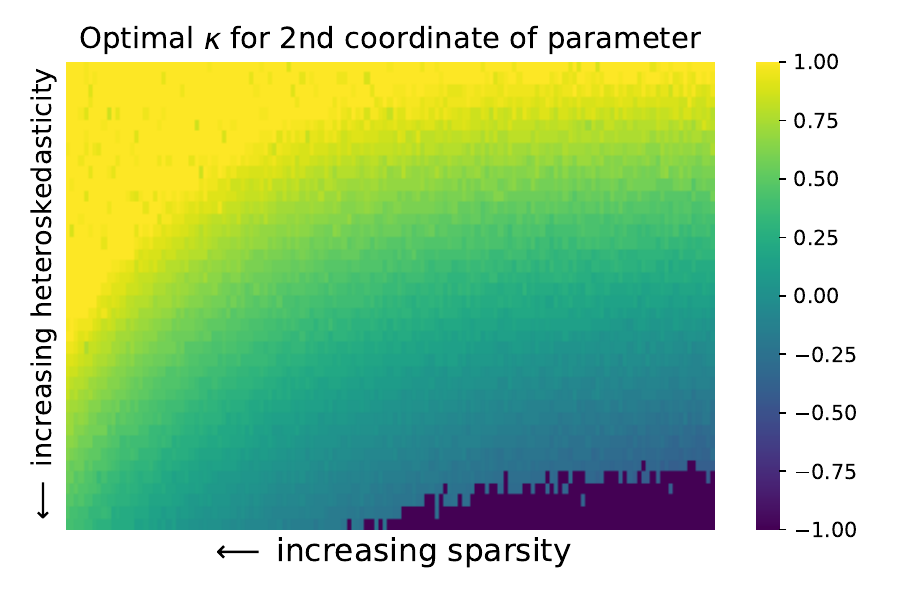}
\par\end{centering}
\caption{Behavior of the optimal estimator under different levels of heteroskedasticity and sparsity. Estimators are indexed by a parameter $\kappa$ taking values between $-1$ and $1$, which controls the relative weight the moment estimator assigns to samples with larger conditional means vs. smaller conditional means. Poisson PML corresponds to a value of 0, non-linear least squares corresponds to 1, while gamma PML corresponds to $-1$. When heteroskedasticity is low relative to sparsity, a larger $\kappa$ is optimal. Poisson is only optimal for a small region in the figures, while non-linear least squares and gamma can also be optimal.}
\label{fig:phase-transition-simple}
\end{figure}

We apply our framework to simulated and real finance datasets, including those related to corporate defaults, credit ratings, corporate patents \citep{hirshleifer2012overconfident}, and real estate investments \citep{bekkerman2023effect}. We find that the estimator recommended by our approach varies significantly with the dataset, and can improve upon the Poisson regression estimator. For instance, an estimator close to the Poisson regression estimator yields the best fit when modeling industry-level defaults. In contrast, for credit ratings and corporate patents data, the optimal estimator recommended by our framework differs from and outperforms Poisson significantly, while for real estate investment data our framework recommends the gamma PML estimator.  Across these datasets, the selected estimator can reduce the out-of-sample root mean squared error (RMSE) by 50\% to 90\% compared to the Poisson PML. It can also provide a better description of the sparsity level in the data compared to Poisson regression, as illustrated in \cref{fig: hist-intro}. Moreover, the choice of the estimator can result in substantially different parameter estimates and economic conclusions (\cref{tab:estimates} in \cref{sec:empirical}). These findings challenge the practice of applying a single regression model across all contexts involving non-negative dependent variables, and underscore the importance of evaluating model suitability for the specific dataset at hand. 

\begin{figure}[t]
\begin{centering}
\includegraphics[scale=0.5]{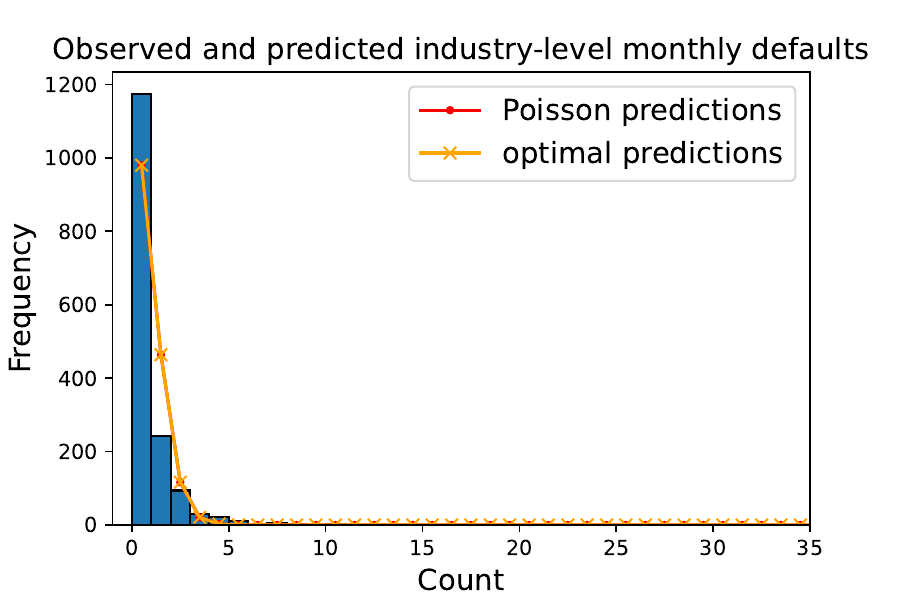}
\includegraphics[scale=0.5]{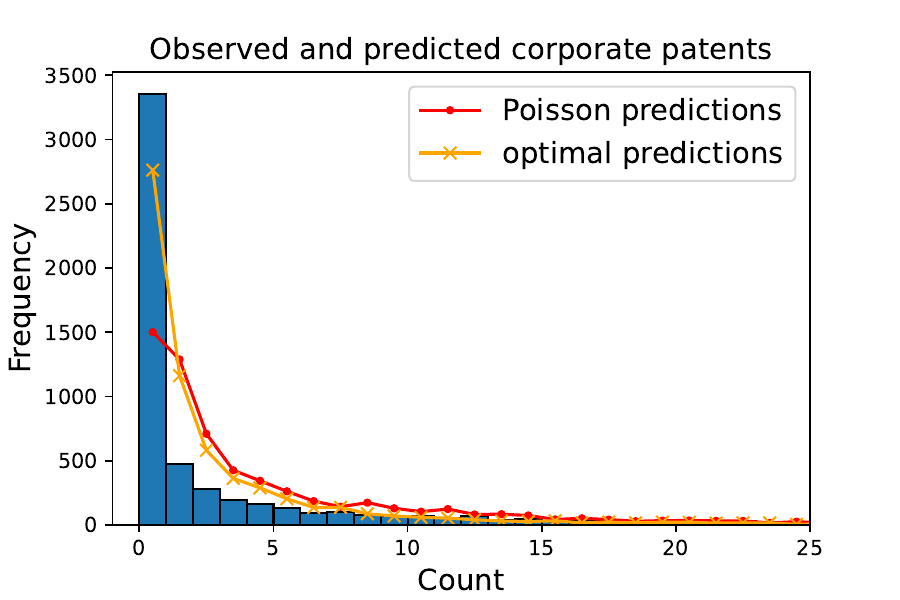}
\par\end{centering}
\caption{Histograms of monthly defaults in Moody's 35 Industrial Categories (left) and corporate patents data of \citet{hirshleifer2012overconfident} (right), overlayed with predictions by the Poisson regression estimator and optimal estimator in our framework. For the defaults data, the optimal model is close to Poisson, and both provide a good description of the outcome variable's distribution. For the patents data, the optimal model is very different from Poisson, and provides a much better description of the outcome variable's distribution, particularly near zero. See \cref{sec:empirical} for details.}
\label{fig: hist-intro}
\end{figure}

The rest of the paper is organized as follows. In \cref{sec:background}, we provide the econometric background on non-negative data modeling and introduce the formal setup and key concepts of the paper. In \cref{sec:trade-off}, we propose a class of moment estimators that we call generalized pseudo maximum likelihood estimators. We characterize the bias and variance of these estimators in the presence of heteroskedasticity and excess sparsity, and propose a cross-validation procedure to select the optimal generalized PML estimator. In \cref{sec:experiments}, we conduct extensive numerical experiments to study the bias-variance trade-off behavior of generalized PML estimators under excess sparsity and heteroskedasticity. In \cref{sec:empirical}, we apply our framework to major datasets in finance. 

\section{Econometrics of Non-negative Data} 
\label{sec:background}
Data with count or non-negative dependent variables often exhibits two prominent features:
heteroskedasticity and sparsity. We begin by providing econometric background on these characteristics. To To contextualize the discussion, we present several relevant regression models that help motivate our framework and main results.

\subsection{Pseudo Maximum Likelihood Estimation and Heteroskedasticity}
\label{subsec:PML-heteroskedasticity}
One of the most popular approaches to modeling data with non-negative dependent variables is log-linear regression, i.e., assuming that the logarithm of the response variable $Y \in \mathbb{R}^+$ is linear in the covariates $X\in \mathbb{R}^d$. To deal with zero values of $Y$, this approach typically uses the transformation $\log(1+Y)$, resulting in the specification
\begin{align*}
\tilde{Y}:=\log(1+Y) & =\beta^{T}X+\epsilon.
\end{align*}
Prominent applications of log-linear regression include the modeling of trade data and panel data with non-negative outcomes, such as earnings \citep{card2001estimating}, and count and count-like data in finance, such as number of corporate patents granted \citep{hirshleifer2012overconfident} and firms' toxic waste release volumes \citep{akey2021limits}. However, previous works have raised some important drawbacks of the log-linear regression approach. For example, \citet{cohn2022count} demonstrate that log-linear models do not always produce estimates with economically meaningful interpretations, and can even result in the wrong sign in expectation. \citet{chen2024logs} similarly show that treatment effect estimates based on log-linear transformations are not invariant under change of units and therefore should not be interpreted as percentages. More importantly, as discussed in previous works, log-linear regression is not robust against heteroskedastic \citep{silva2006log,cohn2022count} and sparse data \citep{silva2010existence,silva2011further}. Similar issues with other non-linear transformations of non-negative outcome variables, such as the inverse hyperbolic sine (IHS), 
have been studied by \cite{mullahy1998much}, \citet{manning1998logged}, \citet{ai2000standard}, and \citet{mullahy2022transform}.

As a more robust alternative, previous works have recommended using 
\emph{pseudo maximum likelihood} (PML) estimation \citep{white1982maximum,gourieroux1984pseudoa,gourieroux1984pseudob,wooldridge1999quasi}. PML has been widely applied to model data when the true likelihood is difficult to specify, e.g., for generalized linear models \citep{mccullagh2019generalized}. The Poisson PML (PPML) is a popular choice in practice for modeling non-negative outcomes, especially count data \citep{cameron2013regression}. When $Y\in\mathbb{Z}^+$
is integer valued, a natural modeling approach is to
assume that $Y$ conditional on $X$ has a Poisson distribution with rate parameter $\lambda(X,\theta)$ that is log-linear in
$X$ with parameter $\theta$, i.e., $\lambda(X,\theta)=\exp(\theta^{T}X)$. The likelihood is given by
\begin{align*}
\mathbb{P}(Y=y\mid X) & =\frac{(\lambda(X,\theta))^{y}\cdot\exp(-\lambda(X,\theta))}{y!}=\frac{e^{y\cdot\theta^{T}X}\cdot\exp(-e^{\theta^{T}X})}{y!}.
\end{align*}
 The maximum likelihood estimator (MLE) $\hat{\theta}$ of $\theta$, given $n$ independent and identically distributed (i.i.d.) samples $\{y_{i},x_{i}\}_{i=1}^{n}$, is then obtained as the maximizer of the
log-likelihood 
\begin{align}
\label{eq:poisson-likelihood}
\max_\theta \sum_{i}y_{i}\cdot\theta^{T}x_{i}-\exp(\theta^{T}x_{i})
\end{align}
modulo constants. When the true distribution of $Y$ given $X$ is not Poisson, the PML approach continues to maximize \eqref{eq:poisson-likelihood}. 
A seminal result of \citet{gourieroux1984pseudob} implies that, under standard assumptions, even when $Y \in \mathbb{R}^+$ is not Poisson distributed or even integer valued, the PML
estimator which maximizes \eqref{eq:poisson-likelihood} is still consistent as long as the conditional mean of $Y$ given $X$ is log-linear,
i.e., $\mathbb{E}(Y \mid X) =\exp({\theta^{T}X})$. 
Importantly, the consistency of Poisson PML holds even if the data is \emph{heteroskedastic}, i.e., $\text{Var}(Y\mid X)$ depends on $X$. In contrast, log-linear regression
is only consistent when the data is homoskedastic.\footnote{The crucial distinction between PPML and log-linear regression lies in the ordering of expectation and logarithm operations. Whereas log-linear regression assumes that the conditional expectation of the log of $Y$ (or $Y+1$) is
linear in $X$, i.e., $\mathbb{E}(\log Y\mid X) =\beta^{T}X$, PPML assumes that the log of the conditional expectation of $Y$ is linear
in $X$, i.e., $\log(\mathbb{E}(Y\mid X)) =\theta^{T}X$. The implications of this distinction, including the robustness of PPML to heteroskedasticity, are discussed extensively in \citet{silva2006log}.} 

{Heteroskedasticity} arises naturally when the outcome variable is constrained to be non-negative. 
Intuitively, when the conditional mean $\mathbb{E}\left(Y\mid X\right) \geq 0$ is large, the residual $Y-\mathbb{E}\left(Y\mid X\right)$ can have larger variance without violating the non-negativity requirement $Y\geq0$. However, when $\mathbb{E}\left(Y\mid X\right)$ is close to zero, the residual needs to have smaller variance in order to satisfy $Y\geq0$. The intrinsic nature of heteroskedasticity in non-negative data can also be justified when economic theory prescribes a model of $\mathbb{E}(Y\mid X)$ with a multiplicative form, e.g., 
$Y=\exp(\theta^{T}X)\cdot\eta$ with $\mathbb{E}(\eta\mid X)=1$. Examples include the gravity models of international trade \citep{anderson1979theoretical,anderson2003gravity}. Expressing $Y$ in the standard additive form $Y = \exp(\theta^{T}X) + \varepsilon$ with $\mathbb{E}(\varepsilon\mid X)=0$ yields $\varepsilon = (\eta - 1)\cdot \exp(\theta^T X)$. The conditional variance $\text{Var}(\varepsilon\mid X)$ will in general depend on $X$, even if the conditional variance $\text{Var}(\eta\mid X)$ does not, resulting in heteroskedasticity. 

Although Poisson pseudo maximum likelihood is robust to heteroskedasticity, 
it is only \emph{efficient} when the conditional variance of $Y$ is \emph{equal} to the conditional mean of $Y$, i.e., ${\rm Var}(Y\mid X) =\mathbb{E}(Y\mid X)$, often called equi-dispersion. When the conditional variance of $Y$ does not scale linearly with the conditional mean, giving rise to the so-called over-dispersion and under-dispersion phenomena, there exist other consistent PML estimators that are more efficient. For example, the gamma PML (GPML) is efficient when ${\rm Var}(Y\mid X)=\mathbb{E}^2(Y\mid X)$, while the closely related negative binomial PML is efficient when ${\rm Var}(Y\mid X) = \mathbb{E}(Y\mid X) + a \mathbb{E}^ 2(Y\mid X)$ for $a>0$.
The non-linear least squares (NLS) estimator 
  can also be viewed as a PML estimator, and is efficient when ${\rm Var}(Y\mid X)$ is constant. 
In this paper, we use $\alpha\geq 0$ to index heteroskedasticity in the specification 
\begin{align}
\label{eq:heteroskedasticity-general}
{\rm Var}(Y\mid X) & =\mathbb{E}^\alpha(Y\mid X).
\end{align} 

It may appear that under heteroskedasticity, one should always use the estimator that best captures the conditional variance as a function of the conditional mean, by determining the value of $\alpha$. Previous works have proposed these types of methods to determine the dispersion using regression-based tests \citep{park1966estimation,cameron1990regression,manning2001estimating}. However, that is only half of the story. When the data contains many observations whose outcomes are zero, efficiency alone may not be sufficient to justify the use of a particular estimator. 
For example, simulation studies in previous works and our paper find that with sparse outcomes resulting from asymmetric censoring, Poisson PML can outperform the gamma PML in terms of mean squared error, even when gamma PML is most efficient. This phenomenon reveals the importance of sparsity when modeling non-negative data, and is a main motivation for the current paper. We therefore discuss sparsity next.

\subsection{Excess Sparsity, Zero-inflation, and Two-part Models}
\label{subsec:censoring}

Non-negative count-like data often have distributions that are right-skewed and exhibit \emph{sparsity}, i.e., having many outcomes with value zero. Such sparsity may arise from features of the data-generating process, including the rarity of the underlying events or the structural properties of networks. For instance, financial networks that capture contractual relationships or interactions between entities are typically sparse and display a core-periphery architecture \citep{craig2014interbank}. Sparse outcomes also arise naturally when the data is censored near zero. For example, in international trade, measurements of large trade volumes are usually more accurate than those of small trade volumes \citep{frankel1993trade}, which are often rounded down to zero. Previous studies have proposed various methods to explain and model sparsity in non-negative data. We now focus on two approaches that are particularly relevant to our framework, beginning with a modification of the Poisson process. 

Suppose that conditional on explanatory variables $X$, the non-negative outcome variable $Y$ follows a Poisson distribution with mean parameter $\exp(\theta_0^T X)$. Then $\mathbb{P}(Y=0 \mid X) =  \exp(-\exp(\theta_0^T X))$. However, in practice, the fraction of zeros observed in data is often not consistent with the predictions of this model, which in a large sample is close to $\int_X  \mathbb{P}(Y=0 \mid X) dX$. In this case, one needs to modify the model to account for the discrepancy. A popular proposal is the two-part or \emph{hurdle model} \citep{mullahy1986specification,king1989event,mullahy2022transform}. It assumes a two-part generating process of $Y$ given $X$. An independent binomial random event determines whether $Y$ is set to zero. If not, the realization of $Y$ is determined according to a \emph{truncated} distribution that takes on strictly positive values. In the simplest case, instead of assuming $ \mathbb{P}(Y=0 \mid X) =  \exp(-\exp(\theta_0^T X))$ as in a Poisson model, we assume 
\[  \mathbb{P}(Y=0 \mid X) =  \exp(-\beta\exp(\theta_0^T X)), \]
which can be understood as the probability of a Poisson variable with mean parameter $\beta\exp(\theta_0^T X)$ being equal to 0, while the distribution of $Y>0$ follows a truncated Poisson with mean parameter $\exp(\theta_0^T X)$, i.e., 
\begin{align*}
    \mathbb{P}(y\mid x)=	\begin{cases}
\exp(-\beta\exp(\theta_{0}^{T}x)) & y=0,\\
\frac{1-\exp(-\beta\exp(\theta_{0}^{T}x))}{1-\exp(-\exp(\theta_{0}^{T}x))}\cdot\frac{e^{-\exp(\theta_{0}^{T}x)}\exp^{y}(\theta_{0}^{T}x)}{y!} & y>0.
\end{cases}
\end{align*}
When $\beta=1$, we recover the original Poisson model. When $\beta<1$, this modified data generating process will result in more sparsity than the standard Poisson model, and vice versa for $\beta>1$. More generally, we can use non-Poisson based distributions for both the binomial zero event variable and the positive outcome variable. This hurdle model has been widely used in practice for non-negative data with high levels of sparsity. 

We can also understand the data generating process in the hurdle model equivalently as the following two-step \emph{censoring} model. For example, when $\beta<1$, in the first step, a variable $Y \geq 0$ is drawn according to a Poisson distribution with mean parameter $\exp(\theta_0^T X)$. If $Y>0$, then with probability 
\begin{align*}
    \frac{\exp(-\beta\exp(\theta_{0}^{T}X))-\exp(-\exp(\theta_{0}^{T}X))}{1-\exp(-\exp(\theta_{0}^{T}X))},
\end{align*}
$Y$ is censored to $0$. We can verify that 
the resulting variable is zero with probability $\exp(-\beta\exp(\theta_0^T X))$, and has the same distribution as the truncated Poisson when it is positive. This censoring process is closely related to the \emph{zero-inflation} model of \citet{lambert1992zero}, which is another common approach to handle excess zeros. In  \citet{lambert1992zero}, a standard Poisson model is generated with mean parameter $\exp(\theta_0^T X)$. Then with probability 
\[\frac{1}{1+\exp(\beta\theta_{0}^{T}X)},\]
the observation is censored to zero. The main difference with the hurdle model is in the censoring probabilities, which is modeled by a logistic function for the zero-inflated Poisson. See also \citet{zorn1998analytic} for a discussion of the hurdle model and zero-inflation model.  

Importantly, in both the hurdle and the zero-inflation models, the probability of observing zero depends on $\exp(\theta_{0}^{T}X)$ (hence on $X$). For example, in the zero-inflation model with $\beta>0$, the censoring probability $\frac{1}{1+\exp(\beta\theta_{0}^{T}X)}$ \emph{decreases} in $\exp(\theta_{0}^{T}X)$, the conditional mean. Similarly, in the hurdle model, the probability $\exp(-\beta\exp(\theta_0^T X))$ of observing a zero outcome variable decreases rapidly in $\exp(\theta_0^T X)$. This feature aligns with how excess zeros appear in applications. In survey and census data, a small outcome or response variable is often more likely to be rounded down to zero than a large outcome, which is less prone to measurement errors \citep{frankel1997regional}. Similarly, count variables of rare events are much more likely to be zero when its conditional mean is close to zero. As we will demonstrate in this paper, this {\emph{asymmetric} sparsity} feature of non-negative data is important for its modeling, as it calls for the use of estimators that may otherwise be inefficient under heteroskedasticity. For example, with sparsity generated by rounding near zero, \citet{silva2006log,silva2011further} demonstrate in simulations that Poisson PML is generally more robust than the gamma PML estimator, even under over-dispersion with ${\rm Var}(Y\mid X) =\mathbb{E}^2(Y\mid X)$, when gamma PML is the most efficient PML. Similarly, our findings in \cref{sec:experiments,sec:empirical} suggest that  \emph{non-linear least squares} can be more robust than the Poisson PML when sparsity is strong, even under equi-dispersion when Poisson is most efficient. 

These observations have received relatively less attention than the efficiency properties of PML estimators, but they raise some important questions on PML estimators, including Poisson regression. Although their behaviors with respect to heteroskedasticity and sparsity have been studied separately before, they are less clear when both features are significant.  
When modeling data with non-negative outcomes, understanding the \emph{interactions} between these features can help empirical researchers decide which regression approach to use  that best handles the sparse non-negative data at hand. 

In this paper, we provide a systematic study of the interplays between these two prominent features of non-negative data. We use the model~\eqref{eq:heteroskedasticity-general} to capture  heteroskedasticity indexed by $\alpha$. To capture excess levels of sparsity, we use the following censoring model:
\begin{align}
\label{eq:censoring-model}
\begin{split}
Y & =\begin{cases}
\exp(\theta_0^T X) + \varepsilon & \text{with probability } 1-P(X,\theta_0)\\
0 & \text{with probability } P(X,\theta_0)
\end{cases},
\end{split}
\end{align}
where the probability $P(X,\theta_0)$ takes one of the following asymmetric forms:
\begin{align}
\label{eq:censoring-probability}
P(X,\theta_0) =  \exp(- (\tau\exp(\theta_{0}^{T}X))^\beta) \quad\text{ or }\quad \frac{1}{1+ (\tau\exp(\theta_{0}^{T}X))^\beta}.
\end{align} 
The probabilities specified in \eqref{eq:censoring-probability} generalize those used in the hurdle and zero-inflation models, with $\beta$ controlling the imbalance in the likelihood of zero between samples with small and large conditional means $\exp(\theta_{0}^{T}X)$. The additional parameter $\tau$ controls the range of $\exp(\theta_{0}^{T}X)$ near zero for which the probability of zero outcomes is significant. A larger $\tau$ leads to a lower level of overall sparsity. As we will see, this probability model induces biases for a family of estimators indexed by $\kappa$, with biases generally increasing and variances decreasing in $\kappa$. The optimal estimator therefore balances this bias-variance trade-off in the presence of excess sparsity and heteroskedasticity.
  
\section{Generalized Pseudo Maximum Likelihood Estimators}
\label{sec:trade-off}
In this section, we build on the preceding discussions and propose a novel family of estimators, which encompasses existing PML estimators such as Poisson and gamma PML. This family of single parameter estimators place heterogeneous weights on samples with varying values of conditional means, which enable them to balance the relative magnitudes of heteroskedasticity and sparsity. We show that these two salient features of non-negative data create a bias-variance trade-off under the models discussed in the previous section. In particular, the optimal choice of estimator varies significantly depending on the data. To enable empirical researchers to leverage these estimators, we also propose a simple cross-validation method to select the estimator that best fits their data, instead of resorting to a default method, such as the Poisson PML.
\subsection{A Unified 
Perspective on Pseudo Maximum Likelihood Estimators}
\label{subsec:z-estimators}
PML estimators for non-negative data are generally constructed based on likelihoods of qualitatively distinct distributions such as the Poisson and gamma distributions. In this paper, we focus on another perspective that reveals their closer connections through the \emph{moment} conditions. \citet{gourieroux1984pseudob} show that the following PML estimators are all consistent for $\theta_0$ when i.i.d. data is generated from the model $Y=\exp(\theta_0^{T}X)+\varepsilon$: 

 \begin{enumerate}
     \item Non-linear least squares, which solves 
\begin{align*}
\min_\theta\sum_{i}(y_{i}-\exp(\theta^{T}x_{i}))^{2} & \Rightarrow\sum_{i}(y_{i}-\exp(\theta^{T}x_{i}))\exp(\theta^{T}x_{i})x_{i}=0,
\end{align*}

\item Poisson PML, which solves
\begin{align*}
\max_\theta\sum_{i}y_{i}\theta^{T}x_{i}-\exp(\theta^{T}x_{i}) & \Leftrightarrow\sum_{i}(y_{i}-\exp(\theta^{T}x_{i}))x_{i}=0,
\end{align*}

\item Gamma PML, which solves 
\begin{align*}
\max_\theta\sum_{i}-\theta^{T}x_{i}-y_{i}\exp(-\theta^{T}x_{i}) & \Leftrightarrow\sum_{i}(y_{i}-\exp(\theta^{T}x_{i}))\exp(-\theta^{T}x_{i})x_{i}=0,
\end{align*}

\item Negative binomial PML, which solves
\begin{align*}
\max_{\theta}\sum_i y_{i}\theta^{T}x_{i}-(\frac{1}{b}+y_{i})\log(1+b\exp(\theta^{T}x_{i})) & \Leftrightarrow\sum_{i}({y_{i}-\exp(\theta^{T}x_{i})})(1+b\exp(\theta^{T}x_{i}))^{-1}x_{i}=0.
\end{align*}
 \end{enumerate}
 Note that the first order conditions of the above PML estimators are all of the form 
\begin{align}
\label{eq:weighted-FOC}
    \sum_{i}(y_{i}-\exp(\theta^{T}x_{i}))(c+\exp(\theta^{T}x_{i}))^{\kappa}x_{i}=0.  
\end{align}
 NLS, Poisson, and gamma PML correspond to $c=0$ and $\kappa=1,0,-1$, respectively, while negative binomial corresponds to $c=1/b$ and $\kappa=-1$. Therefore, these PML estimators can be understood as solving a system of equations that balance the \emph{observed} outcomes $y_i$ with the \emph{expected} outcomes $\exp(\theta^{T}x_{i})$, weighted by the covariates $x_i$ and an estimator-specific term $(c+\exp(\theta^{T}x_{i}))^{\kappa}$ for different values of $c$ and $\kappa$. 

The moment equations \eqref{eq:weighted-FOC} 
provide some intuition on the behavior of PML estimators with respect to heteroskedasticity and sparsity. For example, the well-known inefficiency of NLS under strong heteroskedasticity can be attributed to the weights $(c+\exp(\theta^{T}x_{i}))^{\kappa}=\exp(\theta^{T}x_{i})$ it places on samples: samples with larger (estimated) conditional means $\exp(\theta^{T}x_{i})$ receive larger weights in \eqref{eq:weighted-FOC}. However, such samples also tend to be \emph{noisier} under heteroskedasticity, when ${\rm Var}(Y\mid X)$ increases with $\mathbb{E}(Y\mid X)$. In contrast, the gamma PML estimator, which uses the weights $\exp(-\theta^{T}x_i)$, places higher weight on samples with smaller conditional means, which are less noisy, leading to efficiency gains under strong heteroskedasticity. The Poisson PML has uniform weights $(c+\exp(\theta^{T}x_{i}))^{\kappa}\equiv 1$ on samples, and can therefore be seen as a mid-point between gamma and NLS. In short, under stronger heteroskedasticity in \eqref{eq:heteroskedasticity-general}, i.e., larger  $\alpha$ in ${\rm Var}(Y\mid X)=\mathbb{E}^\alpha(Y\mid X)$, PMLs which place smaller weights on samples with larger conditional means, i.e., smaller $\kappa$ in \eqref{eq:weighted-FOC}, tend to be more efficient and thus preferable. 

However, the opposite is true for bias. When data exhibits exceptional levels of sparsity than predicted by a standard PML model (as in \eqref{eq:censoring-model}), the corresponding PML estimator may suffer from bias. Moreover, when samples with smaller conditional means are more likely to be zero, the bias is more {severe} for estimators that place higher weights on such samples. Take for example the model \eqref{eq:censoring-model} with censoring probability 
\[P(X,\theta_0)=\frac{1}{1+ (\tau\exp(\theta_{0}^{T}X))^\beta}\]
in \eqref{eq:censoring-probability}. Under such a model, samples with smaller $\exp(\theta_0^{T}x_i)$ become \emph{less} reliable as they are more likely to be censored. Consequently, a PML estimator that places larger weights on such samples, i.e., having a smaller $\kappa$ in \eqref{eq:weighted-FOC}, suffers from more severe bias. 
 This explains the observation by some existing works that when sparsity is strong, Poisson ($\kappa=0$) tends to have smaller bias than gamma PML ($\kappa=-1$).

The above discussions suggest that there is an intrinsic trade-off when applying PML methods to data with non-negative outcomes. On one hand, heteroskedasticity necessitates smaller weights on \emph{larger} (noisier) samples to improve efficiency. On the other hand, sparsity encourages smaller weights on \emph{smaller} (more bias-inducing) samples to reduce bias. In the former case, PMLs with small $\kappa$ such as gamma PML and negative binomial are generally preferred, while in the latter case PMLs with larger $\kappa$ such as NLS are generally preferred. \citet{silva2006log} argue that the Poisson PML is a ``reasonable compromise'' between NLS and gamma PML when both heteroskedasticity and excess sparsity are present. However, it is clear that depending on the relative magnitude of the two features, there may be other estimators that can better balance the bias and variance. Our work is precisely motivated by this intuition. 

We propose the following general class of $Z$-estimators \citep{amemiya1985advanced} with estimating equations 
\begin{align*}
\frac{1}{n}\sum_i\psi(y_i,x_i,\theta)\equiv \frac{1}{n}\sum_{i}(y_{i}-\exp(\theta^{T}x_{i}))(c+\exp(\theta^{T}x_{i}))^{\kappa}x_{i}= 0, \quad \kappa\in \mathbb{R}.
\end{align*}
 This class of estimators includes as special cases the standard PML estimators. When $c=0$ and $\kappa=0$, we recover Poisson
PML. When $c=0, \kappa=1$, we recover NLS. When $c=0, \kappa=-1$ we recover gamma PML. When $c=1,\kappa=-1$, we recover negative binomial PML. For non-integer values of $\kappa$, the resulting estimator does not arise from pseudo maximum likelihood estimation based on a particular distribution, but their estimating equations generalize those of the standard PML estimators. For this reason, we refer to them as \textbf{generalized PML estimators}\footnote{When $\kappa>0$, the estimators can also be understood as M-estimators that maximize a non-concave sample average objective, especially when some covariates are strictly positive or negative. See \ref{subsec:optimization}.}. We primarily focus on the case with $c=0$, 
resulting in the estimating
equations 
\begin{align}
\label{eq:generalized-PML}
\frac{1}{n}\sum_{i}(y_{i}-\exp(\theta^{T}x_{i}))\cdot\exp^{\kappa}(\theta^{T}x_{i})\cdot x_{i} & =0, \quad\kappa\in\mathbb{R}.
\end{align} 
Our main message is that, depending on the magnitudes of heteroskedasticity and excess sparsity, 
generalized PML estimators with different (potentially non-integer) values of $\kappa$ is preferable. In particular, Poisson regression ($\kappa=0$) does not always optimally balance the bias-variance trade-off created by the two features.  
To illustrate the magnitude of improvements over Poisson regression, we provide the RMSEs of the optimal estimator vs. the Poisson regression estimator in our simulation studies in \cref{tab:RMSE-simulation}. 

In the rest of the section, we study the asymptotic properties of generalized PML estimators, and provide practical guidance on how to select $\kappa$ in practice.

\subsection{Asymptotic Properties of Generalized PML Estimators}
In this subsection, we provide formulae for the bias and asymptotic variance for the estimators defined by \eqref{eq:generalized-PML}. Under the censoring model \eqref{eq:censoring-model}, the bias does not have an exact formula in general. Instead, we provide an approximation based on Taylor expansions.
\begin{thm}
\label{thm:bias}
For each fixed $\kappa \in [-1,1]$, suppose the equation
\begin{align}
\label{eq:population-moments}
\mathbb{E}_{x}\left\{ (1-P(\theta_{0},x))\cdot\exp(\theta_{0}^{T}x)-\exp(\theta^{T}x)\right\} \cdot\exp(\kappa\theta^{T}x)\cdot x = 0
\end{align}
has a unique solution at $\theta = \tilde \theta_0$ for some $\tilde \theta_0$ in the interior of $\Theta$, a compact subset of $\mathbb{R}^d$. Moreover, assume there exists a function $\alpha(x,y)$ with $\mathbb{E} \alpha(x,y) < \infty$ and almost surely
\[\sup_{\theta \in \Theta}\|(y-\exp(\theta^{T}x))\cdot\exp^{\kappa}(\theta^{T}x)\cdot x\|_\infty \leq \alpha(x,y).\]
The bias of generalized PML estimators $\hat \theta$ that solve \eqref{eq:generalized-PML} can be approximated by 
\begin{align}
\label{eq:bias-formula}
       \hat{\theta}-\theta_{0}	=(A^{T}A)^{-1}Ab + o_p(1),
   \end{align}
   where, with $P(\theta_0,x)$ the censoring probability in \eqref{eq:censoring-model}, $A$ and $b$ are given by \begin{align}
   \label{eq:approximation-bias-integral}
       \begin{split}
           A=\mathbb{E}_x\left\{(P(\theta_{0},x)\kappa-1)\cdot xx^{T}\exp((\kappa+1)\theta_{0}^{T}x)\right\},\\
b= \mathbb{E}_x\left\{P(\theta_{0},x)\cdot x^{T}\exp((\kappa+1)\theta_{0}^{T}x)\right \}.
       \end{split}   \end{align}
\end{thm}
In particular, in the absence of censoring, i.e., $ P(\theta_{0},x)\equiv 0$, 
so that $b\equiv 0$ and 
\cref{thm:bias} implies that all generalized PML estimators are \emph{consistent} estimators of $\theta_0$. Moreover, the bias does not depend on $\alpha$. The quality of the approximation formula in \cref{thm:bias} is an important consideration in practice. Intuitively, since the Taylor expansions used in the derivation of \cref{thm:bias} rely on the small magnitude of $\exp(\theta_{0}^{T}x)$, the formula \eqref{eq:bias-formula} provides a good approximation of the bias when the censoring probability and $\kappa$ result in integrals with small densities on $x$ with large $\exp(\theta_{0}^{T}x)$. In \cref{sec:experiments}, we provide simulation evidence that the approximation formula we derive is close to the true bias in a wide range of settings. 

Next, we compute the asymptotic variance of the generalized PML estimators, taking into account the bias that results from the censoring model \eqref{eq:censoring-model}.
\begin{thm}
\label{thm:asymptotic-variance}
Suppose, in addition to the assumptions in \cref{thm:bias}, that there exists (a possibly different) $\alpha(x,y)$ with $\mathbb{E} \alpha(x,y) < \infty$ and almost surely
\[\sup_{\theta \in \Theta}\|(y-\exp(\theta^{T}x))^2\cdot\exp^{2\kappa}(\theta^{T}x)\cdot xx^T\|_\infty \leq \alpha(x,y), \sup_{\theta \in \Theta}\|\nabla_\theta \left\{(y-\exp(\theta^{T}x))\cdot\exp^{\kappa}(\theta^{T}x)\cdot x\right\}\|_\infty \leq \alpha(x,y). \]    
     Then we have asymptotic normality 
$\sqrt{n}(\theta-\tilde{\theta}_0)\rightarrow_d\mathcal{N}(0,J^{-1}IJ^{-1})$, where
\begin{align}
\label{eq:variance-component}
\begin{split}
   I & =\mathbb{E}_{x}xx^{T}\exp(2\kappa\tilde{\theta}_0^{T}x)\cdot\left[P(\theta_{0},x)(\exp(\tilde{\theta}_0^{T}x))^{2}+(1-P(\theta_{0},x))((\exp(\theta_{0}^{T}x)-\exp(\tilde{\theta}_0^{T}x))^{2}+\exp(\alpha\theta_{0}^{T}x))\right]\\
J & =\mathbb{E}_{x}\left[-\left((1-P(\theta_{0},x))\cdot\exp(\theta_{0}^{T}x)-\exp(\tilde{\theta}_0^{T}x)\right)\cdot\exp(\kappa\tilde{\theta}_0^{T}x)\cdot\kappa xx^{T}+\exp(\kappa\tilde{\theta}_0^{T}x)\exp(\tilde{\theta}_0^{T}x)xx^{T}\right]
\end{split}
\end{align}
where $\tilde{\theta}_{0}$ is the solution of \eqref{eq:population-moments}.
\end{thm}
 
\begin{figure}
\begin{centering}
\begin{adjustbox}{margin=-1cm 0cm 0cm 0cm}
\includegraphics[scale=0.4]{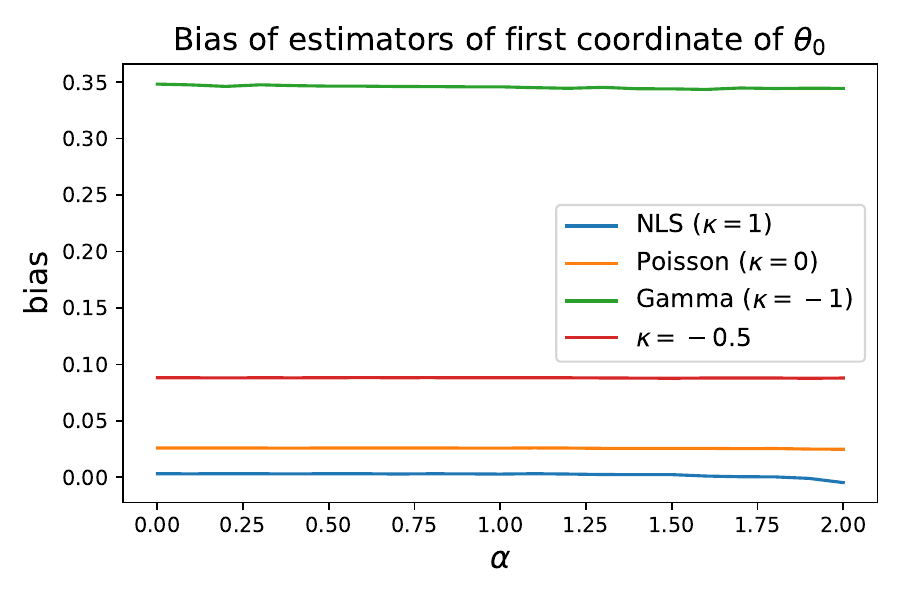}\includegraphics[scale=0.4]{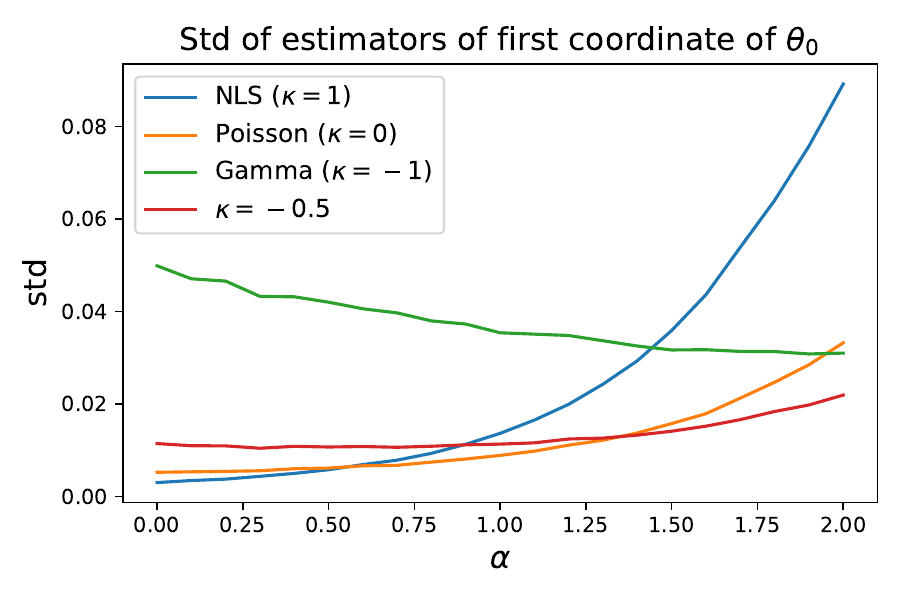}\includegraphics[scale=0.4]{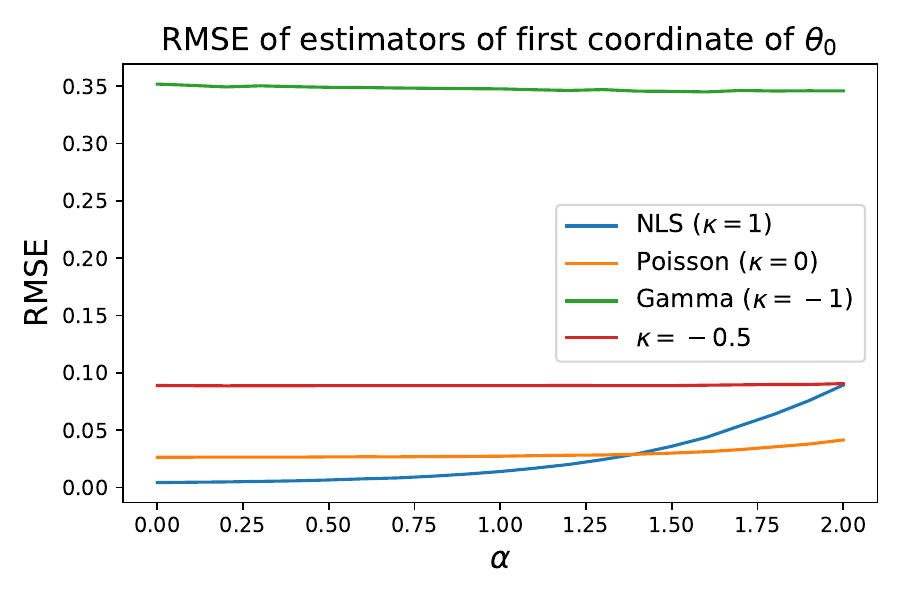}
\end{adjustbox}
\par\end{centering}
\caption{Bias, standard deviation, and RMSE of generalized PML estimators as a function of the heteroskedasticity parameter $\alpha$ with censoring probability $\frac{1}{1+(\tau\exp(\theta_{0}^{T}x))^\beta}$ with $\beta=2$ and $\tau=2$.}
\label{fig:bias-std-rmse-alpha-simple}
\end{figure}

The proof of \cref{thm:asymptotic-variance} relies on asymptotic theory of method of moment estimators \citep{newey1994large}. The uniform boundedness assumptions in \cref{thm:bias} and \cref{thm:asymptotic-variance} are standard and ensure bounded moments and the use of uniform law of large numbers. In the case of no censoring and consistent estimation of $\theta_0$, the asymptotic variance reduces to
\begin{align}
    \left \{ \mathbb{E}_x\left[\exp^{\kappa+1}(\theta_0^Tx) x x^T\right] \right \}^{-1}  \left \{ \mathbb{E}_x\left[\exp^{2\kappa+\alpha}(\theta_0^Tx) x x^T\right] \right \}  \left \{ \mathbb{E}_x\left[\exp^{\kappa+1}(\theta_0^Tx) x x^T\right] \right \}^{-1},
    \label{eq:asymptotic-variance}
\end{align}
which implies that for heteroskedasticity indexed by $\alpha$ in \eqref{eq:heteroskedasticity-general}, the most efficient generalized PML estimator is the one with $\kappa=1-\alpha$, a fact known in existing works for Poisson ($\kappa=0$), gamma ($\kappa=-1$), and NLS ($\kappa=1$). 

\subsection{Bias-variance Trade-off and Selecting the Optimal $\kappa$}
\label{subsec:cross-validation}
As discussed in \cref{subsec:z-estimators}, the generalized PML estimators defined by the moment equations \eqref{eq:generalized-PML} balance the bias-variance trade-off created by heteroskedasticity and excess sparsity. The asymptotic properties of generalized PML estimators established in this section can help provide quantitative descriptions of this trade-off. For example, when $\alpha=0$, $x$ is uniform on $[0,1]$, and $P(\theta_0,x)$ is a threshold censoring model, both the bias and variance of NLS are upper bounded by those of the Poisson PML. As $\alpha$ increases from 0 to 2, the biases do not change, while 
the variance of NLS increases dramatically relative to that of the Poisson PML (see \ref{sec:bias-variance-discussion}). These behaviors lead to a particular value of $\overline{\alpha}\in(0,2)$ where the RMSE of NLS starts to exceed that of the Poisson. Therefore, below heteroskedasticity level  $\overline{\alpha}$, NLS should in fact be preferred than Poisson, even though NLS may have larger variance. We illustrate this phenomenon in \cref{fig:bias-std-rmse-alpha-simple}, based on the same simulation design as that used in \cref{tab:RMSE-simulation} and \cref{sec:experiments}. We see that the bias of estimators decreases with $\kappa$, as the corresponding estimator places less weight on samples with small conditional means. Meanwhile, as $\alpha$ increases, estimators with smaller $\kappa$ becomes more efficient. The trade-off is manifested through the observation that below $\overline{\alpha}\approx 1.4$, the RMSE of NLS is the smallest, despite Poisson and gamma being more efficient. 

More generally, for a particular level of heteroskedasticity and sparsity, there exists a $\kappa$ that optimally balances the bias-variance trade-off. If we consider the two dimensional plane where the x axis is $\alpha$, which quantifies heteroskedasticity, and the y axis is $\tau$, which quantifies the sparsity level, we observe a \emph{phase transition}, where the optimal $\kappa$ changes depending on the particular region of the plane. In \cref{sec:experiments}, we demonstrate through numerical experiments that such a phase transition indeed occurs (\cref{fig:phase-transition,fig:phase-transition-simple}).

The bias-variance trade-off considered in this paper has important implications on how one should handle data with non-negative outcomes in practice. Although Poisson regression is now a popular choice, our work calls for a more principled approach to select the generalized PML estimator with optimal $\kappa$ given the magnitudes of heteroskedasticity and excess sparsity of the data at hand. These characteristics can in principle be determined from data, using for example statistical tests \citep{ park1966estimation,mullahy1986specification}. In order to then determine the optimal $\kappa$, a natural idea is to leverage the asymptotic results developed in this section to assess the bias-variance trade-off. However, there are several challenges to this approach. First, it requires determining the parameters $\alpha$ in the heteroskedasticity model \eqref{eq:heteroskedasticity-general} and $\beta,\tau$ in the censoring model. Second, the formulae for bias and asymptotic variance depend on $\theta_0$, whose estimator $\hat \theta$ based on \eqref{eq:generalized-PML} is biased under a censoring model, which may affect the accuracy of the bias and variance estimates. In this work, we propose an alternative, simpler method to select $\kappa$, which we describe next.

In practice, we can determine the optimal $\kappa$ using a standard $k$-fold cross-validation. More precisely, consider a dataset consisting of $n$ i.i.d. samples $\{y_{i},x_{i}\}_{i=1}^{n}$ where $y_i\geq0$. We split the dataset randomly into $k$ folds $F_1\cup \dots\cup F_k=\{y_{i},x_{i}\}_{i=1}^{n}$ with roughly equal sizes. For each fold $j=1,\dots,k$, we construct generalized PML estimators $\hat \theta^{(j)}_\kappa$ with $\kappa$ in a grid on $[-b,b]$ using data from the complement $\cup_{j'\neq j}F_{j'}$ of $F_j$, and compute the out-of-sample MSE 
\[e_{\kappa}^{(j)}=\frac{1}{|F_{j}|}\sum_{(x_{i},y_{i})\in F_{j}}(y_{i}-\exp(x_{i}^{T}\hat{\theta}_{\kappa}^{(j)}))^{2}.\]
We then average the MSEs across all folds, i.e., $e_\kappa = \frac{1}{k}\sum_{j=1}^k e_{\kappa}^{(j)}$, and select the $\kappa$ that minimizes $e_\kappa$. In \cref{sec:empirical}, we apply this procedure to four datasets from finance applications, and demonstrate that the optimal $\kappa$ can vary significantly across applications (\cref{fig:RMSE}). 

\section{Simulation Studies} 
\label{sec:experiments}
In this section, we verify our theoretical findings and demonstrate the prevalence of the bias-variance trade-off for data with non-negative outcomes.
\subsection{Simulation Design}
Our simulation design follows the setup in \citet{silva2006log}, combined with the censoring model introduced in this paper. 
First, we generate a 2-dimensional covariate $X$ with its first coordinate $X_1$ drawn from a standard normal distribution and second coordinate $X_2$ drawn from the uniform distribution on $[0,1]$. Given $X$, the outcome variable is generated according to the model \eqref{eq:censoring-model}:
\begin{align*}
Y_{i}=\begin{cases}
\exp(\theta_{0}^{T}X_{i})\cdot\eta_{i} & \text{with probability } 1-P(\theta_{0},X_{i})\\
0 & \text{with probability }P(\theta_{0},X_{i}),
\end{cases}
\end{align*}
 where the parameter of interest $\theta_0=[1,1]$ and the multiplicative error $\eta_i$ is log-normal, i.e., $\log \eta_i$ is normal with zero mean. The variance of $\eta_i$ is set to be $\exp^{\alpha-2}(\theta_{0}^{T}X_{i})$, so that the conditional variance of $\exp(\theta_{0}^{T}X_{i})\cdot\eta_{i}$ is equal to $\exp^{\alpha}(\theta_{0}^{T}X_{i})$. The censoring probability $P(\theta_{0},X_{i})$ is given by 
\[P(X_i,\theta_0) =\frac{1}{(1+ (\tau\exp(\theta_{0}^{T}X_{i}))^\beta)},\]
which can be viewed as continuous approximations to the step function. Compared to the double exponential specification in \eqref{eq:censoring-probability}, the power decay specification provides a heavier tail, but their qualitative behaviors with respect to $\beta$ and $\tau$ are similar. $\beta$ regulates the contrast in censoring probabilities between samples with small or large conditional means. As $\beta$ increases from 1, samples with $\tau \exp(\theta_{0}^{T}X_{i})<1$ become increasingly susceptible to censoring, while samples with $\tau \exp(\theta_{0}^{T}X_{i})\geq1$ become less affected. $\tau$ regulates the point of ``discontinuity''. As $\tau$ increases from 0, the range of $\exp( \theta_{0}^{T}X_{i})$ for which censoring is significant becomes smaller and more concentrated around 0, resulting in overall \emph{less} sparsity. For some constant $b\geq 1$, we consider the family of generalized PML estimators defined by \eqref{eq:generalized-PML}:
\begin{align*}
    \frac{1}{n}\sum_{i}(Y_{i}-\exp(\theta^{T}X_{i}))\cdot\exp^{\kappa}(\theta^{T}X_{i})\cdot X_{i} & =0,\quad\kappa\in[-b,b],
\end{align*}
and study their bias and variance. A practical consideration is how to obtain the generalized PML estimators for $\kappa \notin \{-1,0,1\}$. In \ref{subsec:optimization}, we propose an approach by solving optimization problems whose first order conditions recover the estimating equations \eqref{eq:generalized-PML}.

\subsection{Evidence of the Bias-Variance Trade-off}
In this part of the simulation studies, we investigate the various manifestations of the bias-variance trade-off, by holding subsets of the tuple $(\alpha,\tau,\beta,\kappa)$ in our framework fixed and varying the rest. 

\begin{figure}
\begin{centering}
 \begin{adjustbox}{margin=-1cm 0cm 0cm 0cm}
\includegraphics[scale=0.4]{figures/bias_1_alpha.pdf}\includegraphics[scale=0.4]{figures/std_1_alpha.pdf}\includegraphics[scale=0.4]{figures/RMSE_1_alpha.pdf}
\end{adjustbox}
\par\end{centering}
\begin{centering}
\begin{adjustbox}{margin=-1cm 0cm 0cm 0cm}
\includegraphics[scale=0.4]{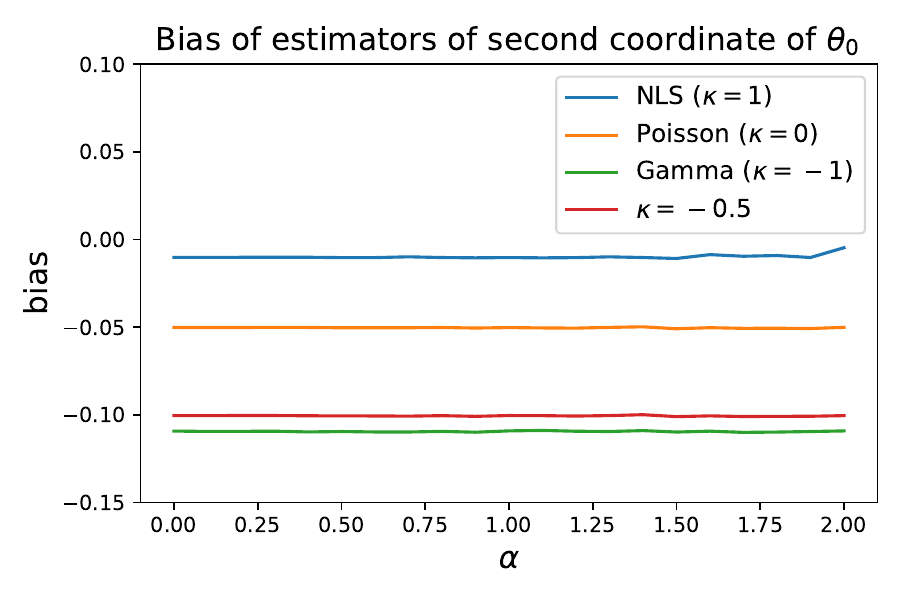}\includegraphics[scale=0.4]{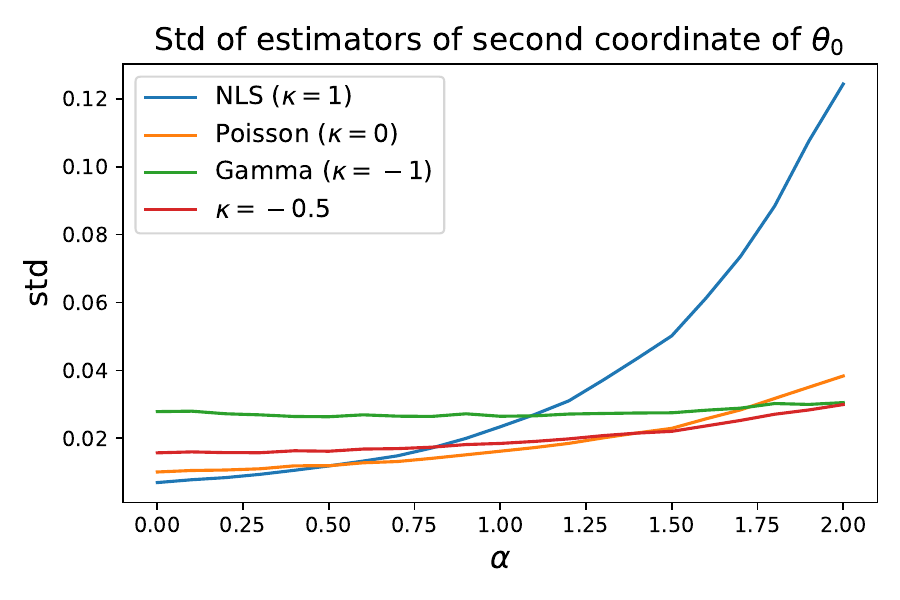}\includegraphics[scale=0.4]{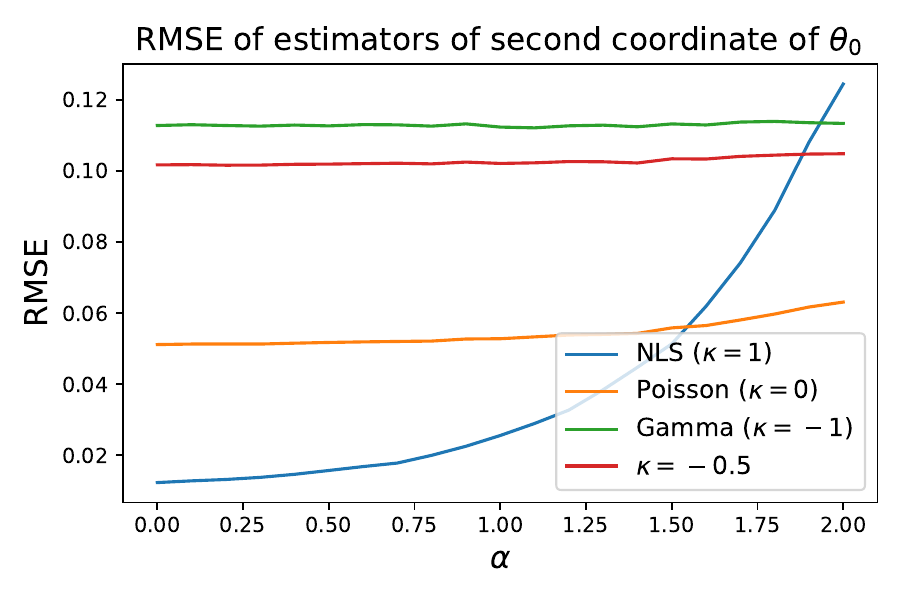}
\end{adjustbox}
\par\end{centering}
\caption{Bias, standard deviation, and RMSE of generalized PML estimators as a function of the heteroskedasticity parameter $\alpha$ and censoring probability $\frac{1}{1+(\tau\exp(\theta_{0}^{T}x))^\beta}$ with $\beta=2$ and $\tau=2$. Results are calculated based on 5000 replications of data with sample size 3000.\zq{Make the legend fonts larger}
} 
\label{fig:bias-std-rmse-alpha}
\end{figure}
First, in \cref{fig:bias-std-rmse-alpha} we plot the bias, standard deviation, and root mean squared error (RMSE) of generalized PML estimators as a function of the degree of heteroskedasticity $\alpha\in[0,2]$, for fixed $\beta=2, \tau=2$. We find that bias is decreasing in $\kappa$, and is essentially flat in $\alpha$, which is expected since the bias approximation formula does not depend on $\alpha$. Moreover, the standard deviation increases with $\alpha$ for most estimators, but those with larger $\kappa$ experience faster increase of the standard deviation, as they are less efficient under heteroskedasticity. As a result, when $\alpha$ grows, for the first coordinate of $\theta_0$, the standard deviation of NLS starts to be dominated by that of the Poisson at $\alpha\approx 0.6$, which in turn is lower bounded by the standard deviation of the generalized PML estimator with $\kappa=-0.5$ starting $\alpha\approx 1.35$ and lower bounded by the gamma PML starting $\alpha\approx 1.5$. Interestingly, the standard deviation of gamma PML ($\kappa=-1$) is \emph{larger} than that of $\kappa=-0.5$ even for $\alpha=2$, when in the absence of asymmetric censoring gamma PML is most efficient. This is due to censoring having an effect on the asymptotic variance.

If we ignore the bias from asymmetric censoring, we might conclude that Poisson PML is preferable for heteroskedasticity levels $\alpha \in [0.6, 1.35]$-a rule of thumb that is, in fact, often recommended in practice. However, the RMSE comparisons in \cref{fig:bias-std-rmse-alpha} suggest that the NLS in fact dominates the Poisson (and gamma) for all $\alpha \leq 1.35$, due to its smaller bias. Moreover, the efficiency gain of gamma PML and $\kappa=-0.5$ for large $\alpha$ is overwhelmed by their large biases, making the Poisson PML the best option for $\alpha>1.35$. Similar observations hold for the second coordinate of $\theta_0$. These results suggest that in the presence of sparsity, conventional consensus on the superiority of Poisson PML may be incomplete and misguided, and that the bias plays an important role in a more systematic analysis. 
Our simulation results in \cref{fig:bias-std-rmse-alpha} suggest that, under the common asymmetric censoring model that results in excess sparsity, the RMSE of NLS continues to dominate that of the Poisson PML beyond $\alpha=1$, with a transition threshold of $\alpha\approx 1.35$ to $1.5$.

\begin{figure}
\begin{centering}
\begin{adjustbox}{margin=-1cm 0cm 0cm 0cm}
\includegraphics[scale=0.4]{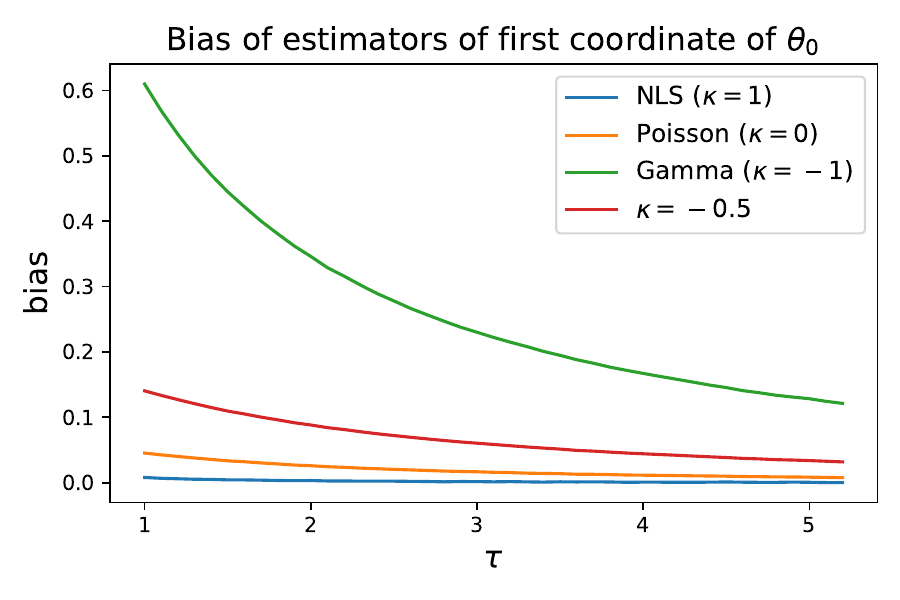}\includegraphics[scale=0.4]{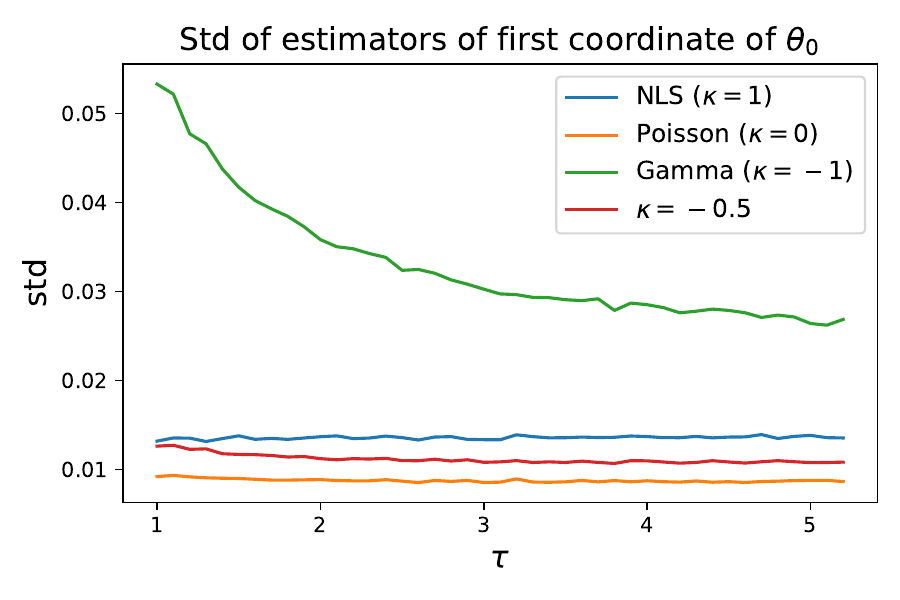}\includegraphics[scale=0.4]{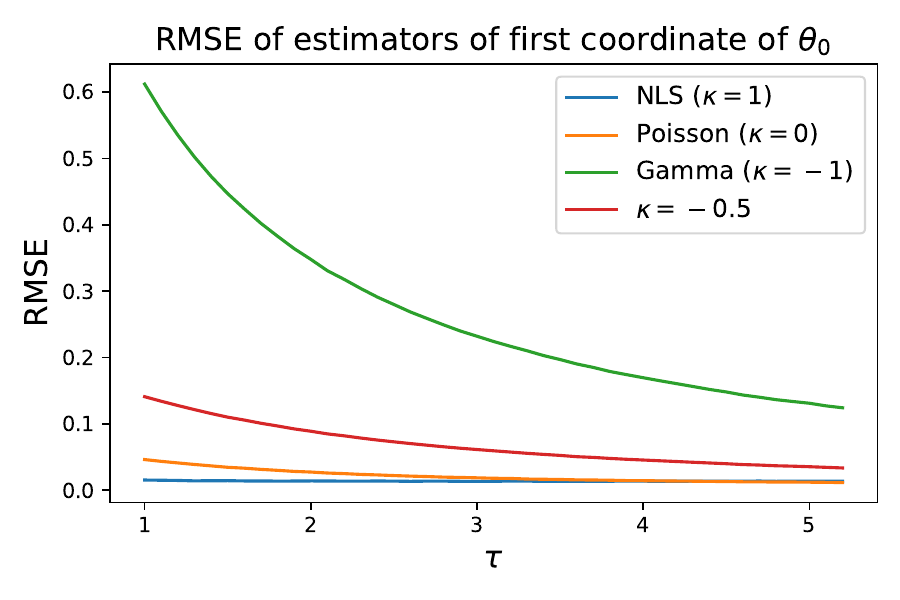}
\end{adjustbox}
\par\end{centering}
\begin{centering}
\begin{adjustbox}{margin=-1cm 0cm 0cm 0cm}
\includegraphics[scale=0.4]{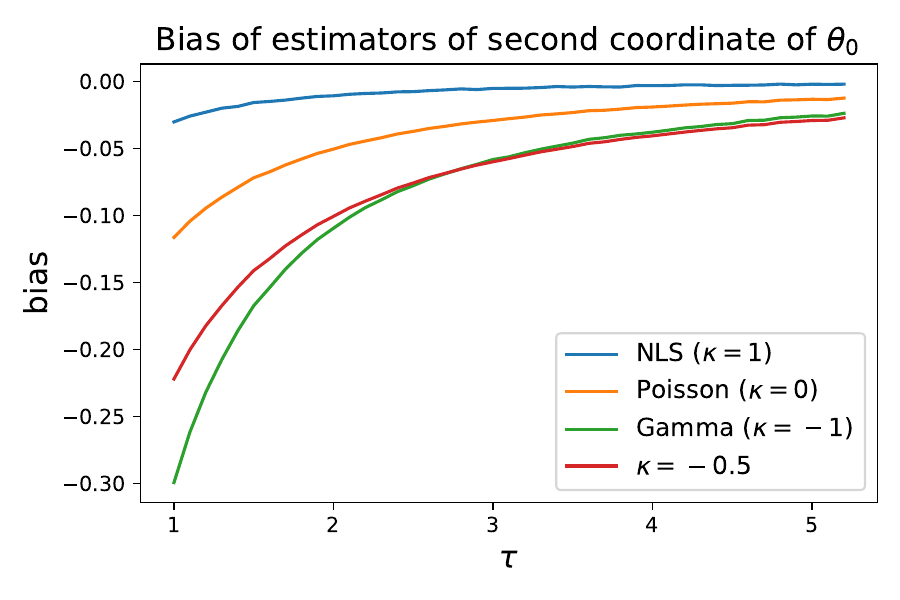}\includegraphics[scale=0.4]{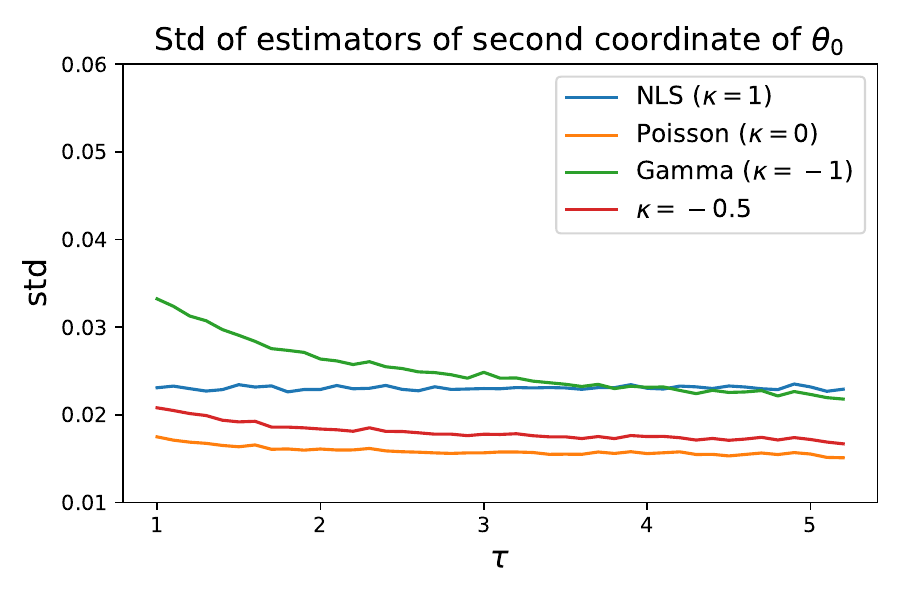}\includegraphics[scale=0.4]{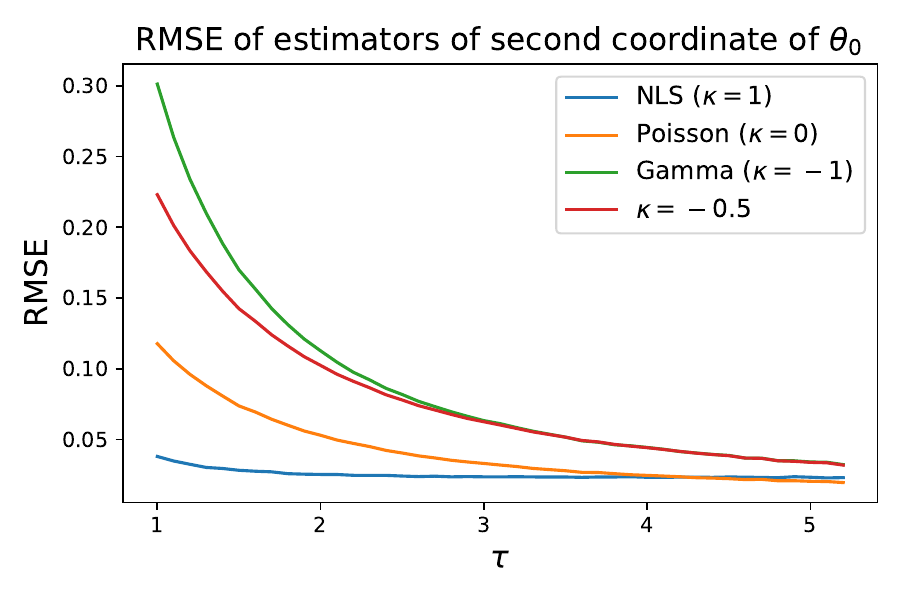}
\end{adjustbox}
\par\end{centering}
\caption{Bias, standard deviation, and RMSE of PML estimators under censoring probability $\frac{1}{1+(\tau\exp(\theta_{0}^{T}x))^\beta}$ and $\alpha=1,\beta=2$. Results are calculated based on 5000 replications of data with sample size 3000.}
\label{fig:bias-tau}
\end{figure}

Next, in \cref{fig:bias-tau}, we plot the bias, variance, and RMSE of generalized PML estimators as a function of $\tau$ under the censoring probability $\frac{1}{1+(\tau\exp(\theta_{0}^{T}x))^\beta}$ with $\beta=2$ and $\alpha=1$. Recall that increasing $\tau$ reduces overall sparsity of the data, so we should expect the bias of estimators to decrease, which is indeed the case. Moreover, the bias plots show that NLS is relatively insensitive to the level of sparsity in the data, whereas PML estimators with negative $\kappa$ are quite sensitive to small values of $\tau$, when censoring is strong and samples with small conditional means are less reliable. The variance of Poisson is the smallest among all generalized PML estimators, since $\alpha=1$. The RMSE plots again  demonstrate the bias-variance trade-off, with NLS having the smallest RMSE for $\tau\leq 4$, while Poisson starts to have smaller RMSE for $\tau>4$, when the bias from sparsity becomes less dominant. These results again suggest that in the presence of excess sparsity, contrary to popular belief, the Poisson PML is not necessarily the best choice, even when it is expected to be the most efficient. The optimal choice depends on the degree of censoring.

\begin{figure}[t]
\begin{centering}
\begin{adjustbox}{margin=-1cm 0cm 0cm 0cm}
\includegraphics[scale=0.4]{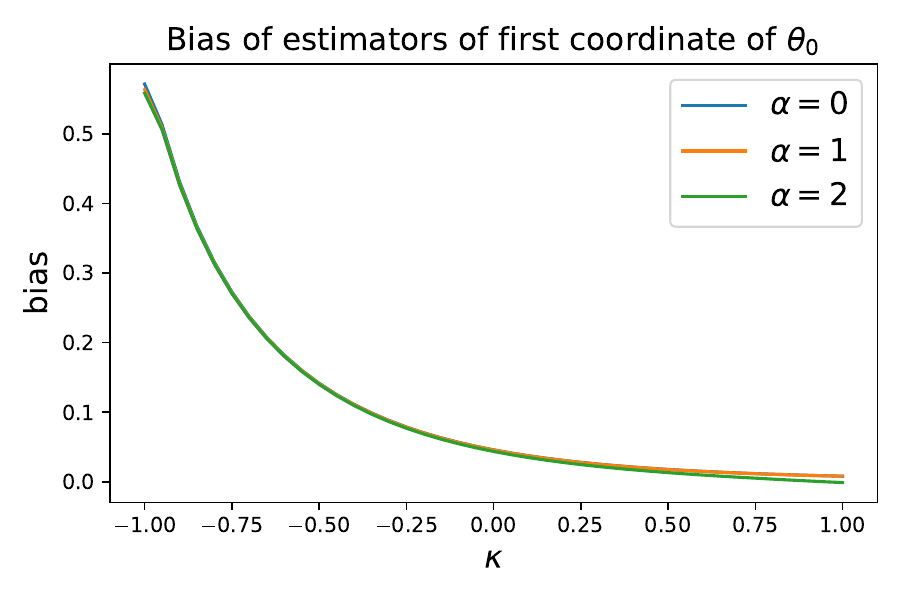}\includegraphics[scale=0.4]{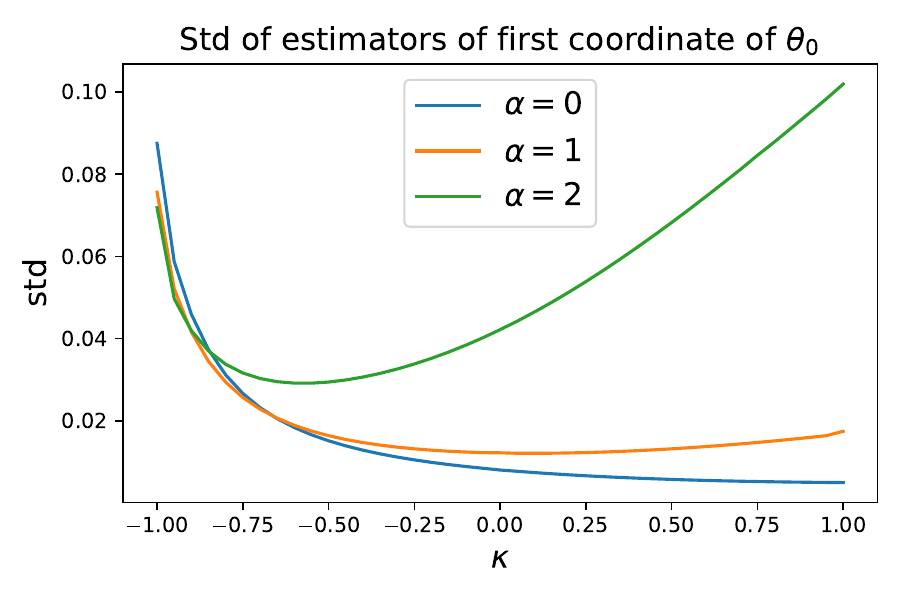}\includegraphics[scale=0.4]{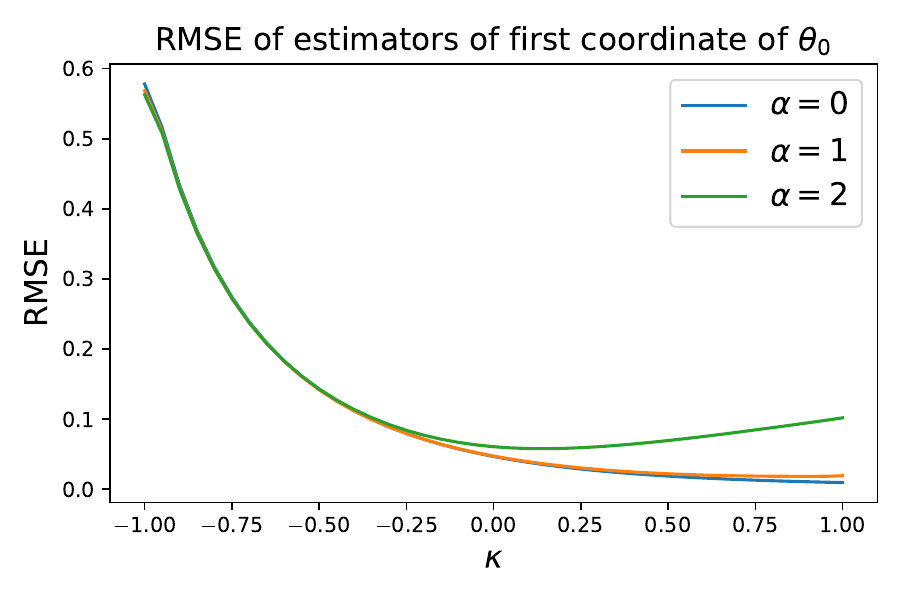}
\end{adjustbox}
\par\end{centering}
\begin{centering}
\begin{adjustbox}{margin=-1.2cm 0cm 0cm 0cm}
\includegraphics[scale=0.4]{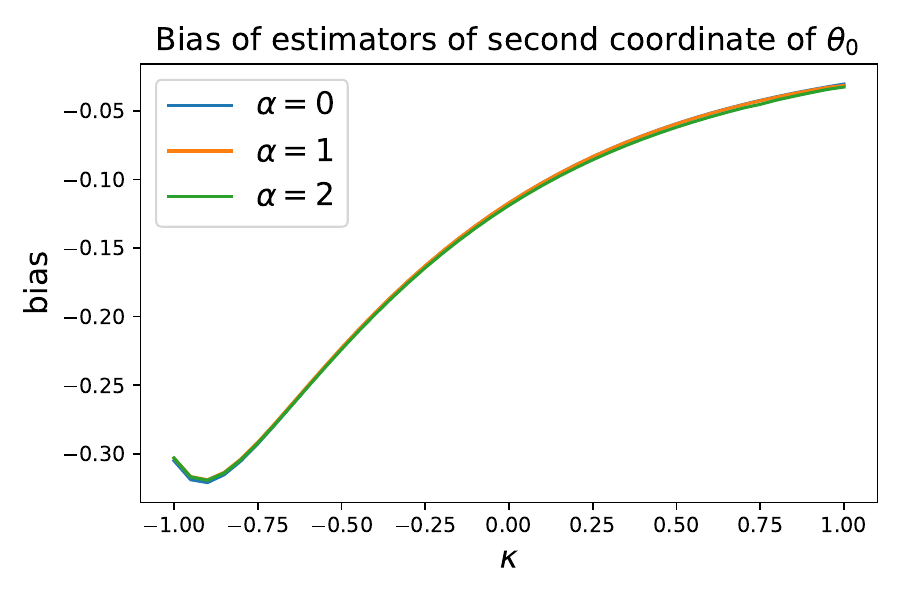}\includegraphics[scale=0.4]{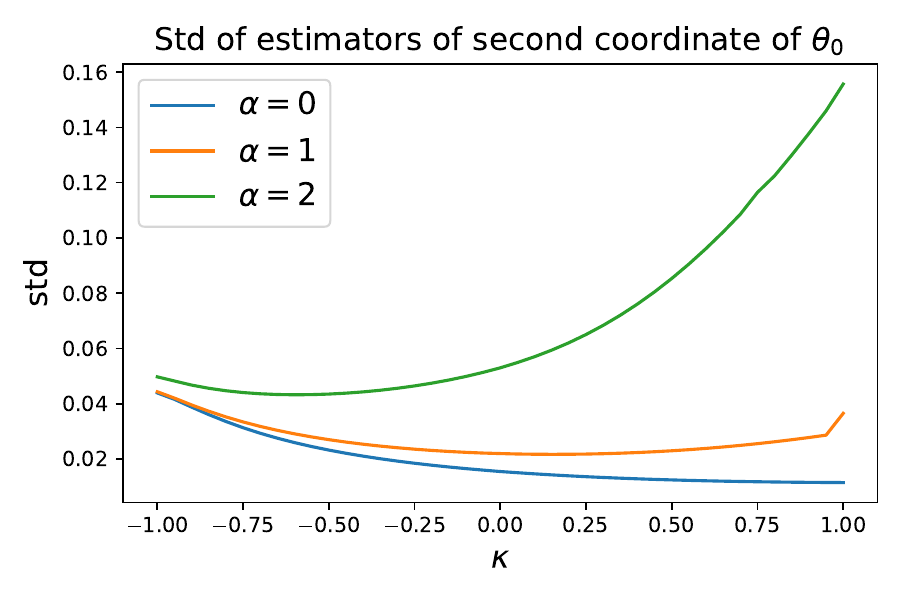}\includegraphics[scale=0.4]{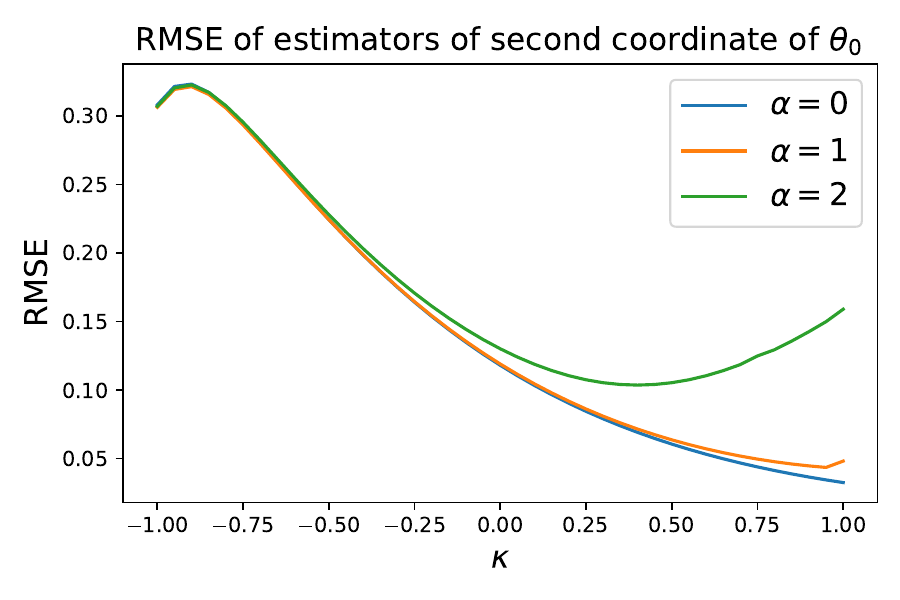}
\end{adjustbox}
\par\end{centering}
\caption{Bias, standard deviation, and RMSE of PML estimators under censoring probability $\frac{1}{1+(\tau\exp(\theta_{0}^{T}x))^\beta}$ and $\tau=1,\beta=2$. Results are calculated based on 5000 replications of data with sample size 3000.}
\label{fig:bias-std-rmse-kappa}
\end{figure}

Our next set of simulation results in \cref{fig:bias-std-rmse-kappa} reports the bias, standard deviation, and RMSE of generalized PML estimators for $\kappa \in [-1,1]$, under censoring probability $\frac{1}{1+(\tau\exp(\theta_{0}^{T}x))^\beta}$ with $\tau=1$, $\beta=2$, and  $\alpha\in [0,1,2]$. We find that the bias decreases as $\kappa$ increases from $-1$ to 1, as the corresponding generalized PML estimators place increasing weights on larger samples thus reducing bias. Similarly, the standard deviation of estimators decreases then increases with $\kappa$, resulting in a choice of $\kappa$ with the smallest variance.  However, this variance-minimizing $\kappa$ generally differs from the one in the uncensored case. For example, when $\alpha=2$, we know that the gamma PML ($\kappa=-1)$ is the most efficient estimator in the absence of censoring. However, in the presence of asymmetric censoring, the standard deviation is minimized at around $\kappa\approx-0.6$ for both coordinates of $\theta_0$. This observation is in line with our theoretical result on the asymptotic variance, which depends on the censoring parameters. Lastly, the RMSE plots confirm our intuition and quantitative results on the bias-variance trade-off. There is an optimal $\kappa$ with the smallest RMSE, which depends on the degree of heteroskedasticity $\alpha$ and sparsity $\tau$. For example, although Poisson PML ($\kappa=0$) is widely used practice especially when heteroskedasticity is not negligible (e.g., $\alpha=1$), in the presence of moderate censoring ($\tau=1$ and $\beta=2$), NLS ($\kappa=1$) in fact has lower RMSE and should be preferred.

We summarize the preceding three sets of simulation results in the phase transition plot in \cref{fig:phase-transition}. 
We see that, as heteroskedasticity decreases or sparsity increases, the optimal $\kappa$ decreases, and vice versa. Therefore, we advocate for a more systematic modeling of non-negative dependent variables based on assessing the relative importance of the two aspects of data. In order to determine $\alpha$ and $\tau$, one could consider performing tests along the lines of \citet{park1966estimation,mullahy1986specification} or \citet{cameron1990regression}. In this paper, we employ the simpler cross-validation procedure proposed in \cref{subsec:cross-validation} to select $\kappa$ directly, and apply it to a collection of datasets in finance to demonstrate its usefulness.

\begin{figure}[t]
\begin{centering}
\includegraphics[scale=0.5]{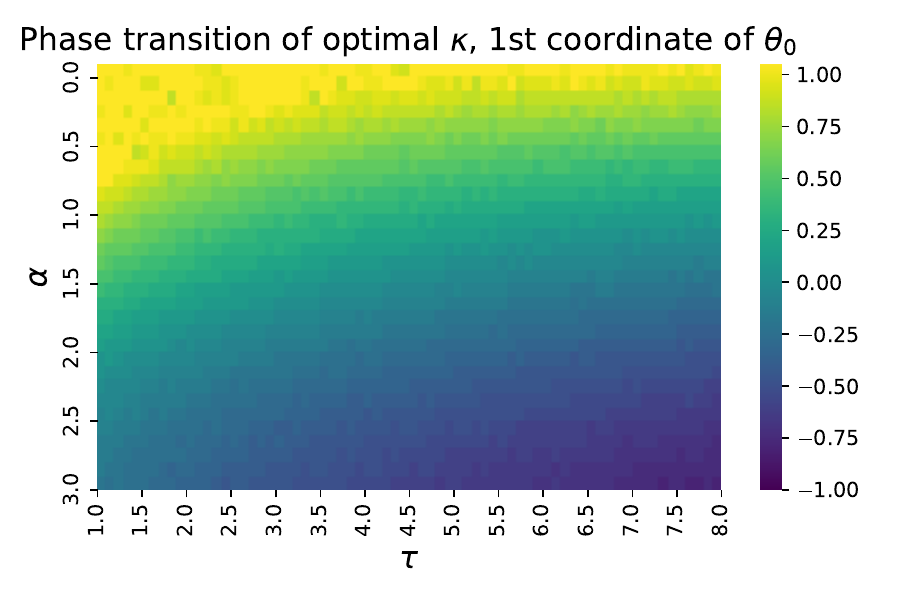}
\includegraphics[scale=0.5]{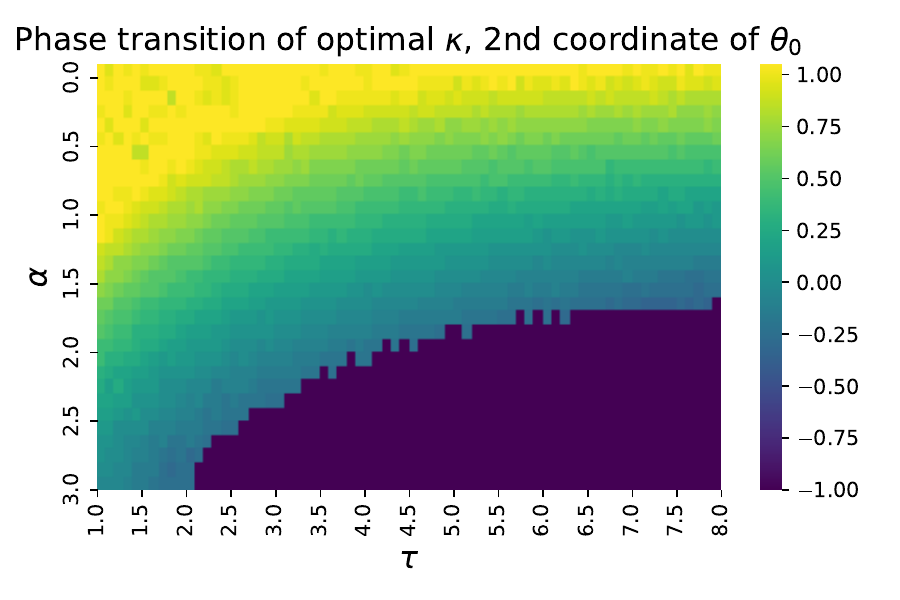}
\par\end{centering}
\caption{Phase transition of the optimal $\kappa$ for different levels of heteroskedasticity ($\alpha$) and sparsity ($\tau$). 
Results based on simulated data with true parameter $\theta_0=(1,1)$.}
\label{fig:phase-transition}
\end{figure}

\section{Empirical Validation on Datasets from Finance Applications}
\label{sec:empirical}
In this section, we illustrate the appeal of the proposed generalized PML approach using four distinct datasets: corporate defaults, credit ratings, corporate patents, and real estate investments. The first two applications focus on modeling default rates and credit ratings across industries defined by Moody's 35 Industry Categories\footnote{\url{https://www.moodys.com/sites/products/ProductAttachments/MCO_35\%20Industry\%20Categories.pdf}}.  In the third application, we analyze the number of corporate patents granted, using the dataset from \citet{cohn2022count}, who replicate the study by \citet{hirshleifer2012overconfident} on the impact of overconfident CEOs on corporate innovation. In the fourth application, we model the number of real estate permits using the dataset provided by \citet{bekkerman2023effect}. We demonstrate that, 
using the simple cross validation procedure based on out-of-sample MSE, one can determine the generalized PML estimator that best fits the data. Moreover, this optimal choice is highly data-dependent. In some cases, it is close to the Poisson regression estimator recommended in the literature; for other applications, it is very different from Poisson and significantly outperforms it. We provide some intuitive evidence on the different levels of heteroskedasticity in the datasets. 
Our results in this section therefore demonstrate that, instead of always using a particular method like the Poisson to model count-like data, one should employ a more principled approach to model such data that takes into account both heteroskedasticity and sparsity.

\subsection{Data Description}
We first describe the datasets used in this section. We construct our datasets on corporate default and credit rating  leveraging Moody's Default and Recovery Database (DRD), which provides a detailed history of firm defaults and ratings across a broad collection of industries. To construct the default outcome variable, we aggregate the total number of default events across all firms in each of the Moody's 35 Industry Categories for each month between January 2019 and December 2023. 
Each observation consists of the monthly default number for a particular industry category and particular month-year, which is a sparse non-negative variable. Default events as defined in the Moody's database include bankruptcy, missed payments, distressed exchange, and others. \cref{fig:default-hist} provides a histogram of the number of monthly default events. The majority of industry-year/month pairs have zero defaults, as expected. A large outlier of 35 events is recorded for February 2023 in the banking industry. These events correspond to the series of bank failures and bankruptcies in early 2023. The ratings outcome is constructed in a similar way: we compute the monthly number of firms in each of Moody's Industry Categories that are rated investment grade, i.e., Baa3 or higher. We then normalize the counts by their standard deviations in each month. Distribution of this non-negative outcome variable is also described in the histogram in \cref{fig:default-hist}.

\begin{figure}[t]
\begin{centering}
\includegraphics[scale=0.5]{figures/hist_defaults.pdf}
\includegraphics[scale=0.5]{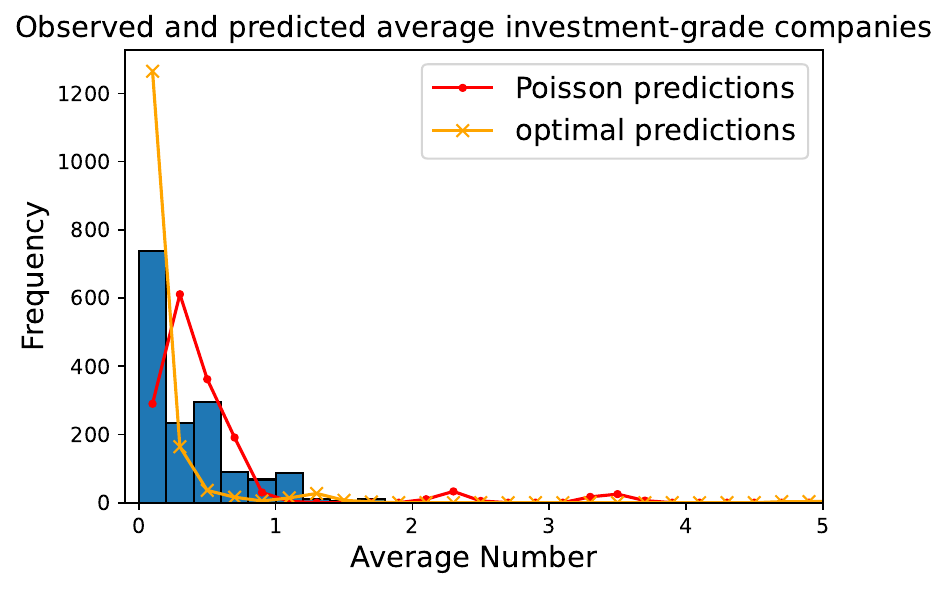}\\
\includegraphics[scale=0.5]{figures/hist_patents.pdf}
\includegraphics[scale=0.5]{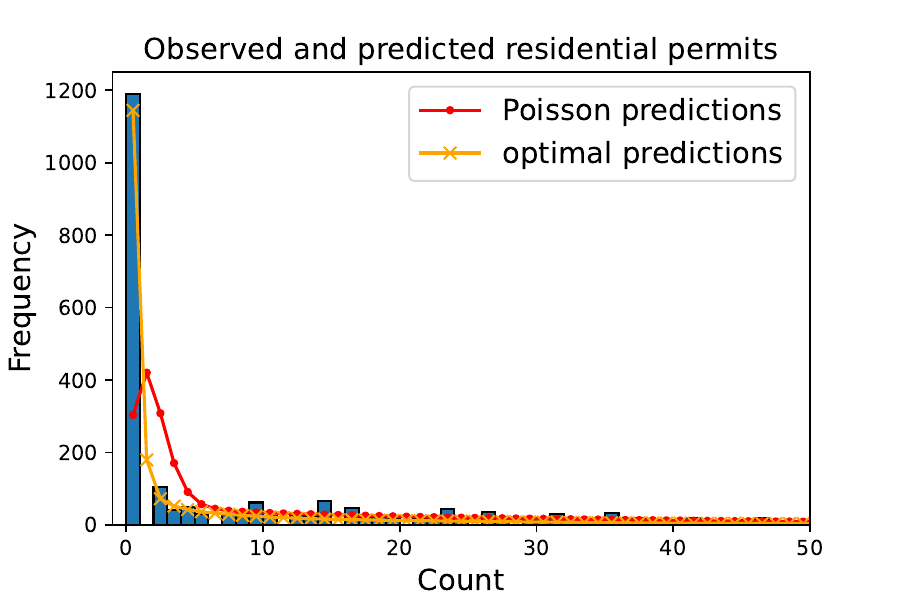}
\par\end{centering}
\caption{Histograms of outcome variables in finance datasets, overlayed with predictions of Poisson regression estimator and the optimal generalized PML estimator, selected using cross-validation.
Upper left: Monthly defaults in Moody's 35 Industrial Categories between January 2020 and December 2023. A single large outlier of 35 events is recorded for February 2023 in the banking industry. Upper right: Monthly number of investment-grade companies. Lower left: Replication dataset of \citet{hirshleifer2012overconfident} constructed by \citet{cohn2022count} of corporate patents, truncated to samples with $y\leq25$. Lower right: Residential permits based on data used in \citet{bekkerman2023effect}, truncated to samples with $y\leq50$.}
\label{fig:default-hist}
\end{figure}

To construct the predictors for default rates, we follow the setup in \citet{duffie2007multi} and consider the following macroeconomic and firm-specific factors:
\begin{itemize}

\item Average firm's distance to default, constructed based on \citet{vassalou2004default};

    \item Monthly industry average of trailing 1-year stock return;

    \item The 3-month Treasury bill rate;

    \item Trailing 1-year return on the S\&P 500 index.
\end{itemize}
The T-bill data is obtained from the Federal Reserve Bank of St. Louis, the trailing 1-year return on the S\&P 500 index is obtained from CRSP. Monthly industry average of trailing 1-year stock return is obtained from CRSP/Compustat. To calculate firm's distance to default, we combine several data sources available on WRDS and apply the iterative algorithm proposed by \citet{vassalou2004default}. In order to match stock returns data from CRSP with the Moody's 35 Industry Categories, we use the 3-digit NAICS code in the Moody's dataset and the first 3 digits of the 6-digit NAICS code in the CRSP data. Similarly, for the credit rating outcome, we use the following predictors constructed from CRSP/Compustat data that have been found to have predictive power \citep{hirk2022corporate}:
\begin{itemize}

\item Total debt to total asset ratio;

    \item Total debt to earnings before interest, taxes, depreciation, and
amortization ratio;

    \item Retained earnings divided by total assets;

    \item Long-term debt to assets ratio;

    \item Earnings before interest and taxes divided by total assets
\end{itemize}

The dataset on corporate patents is provided by \citet{cohn2022count} who replicate the work of \citet{hirshleifer2012overconfident}. It consists of the number of patents granted to a firm as the dependent variable, as well as firm characteristics such as stock returns, sales, institutional holdings, and CEO overconfidence. The dataset on real estate investments is constructed by \citet{bekkerman2023effect} using data from Airbnb, and consists of the number of residential permits in each zipcode area in a particular month as the outcome variable, along with a vector of time-varying zipcode characteristics from the American Community Survey, such as median household income, population size, percentage of people aged between 25 to 60 with a bachelor's degree, and employment rate.

\subsection{Optimal Generalized PML Estimators for Finance Datasets}
\begin{figure}[t]
\begin{centering}
\includegraphics[scale=0.5]{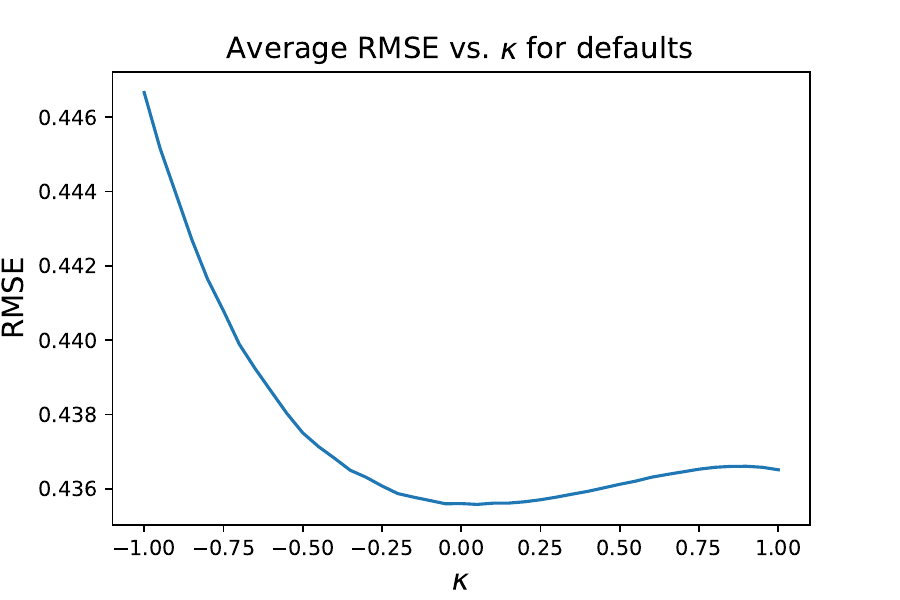}
\includegraphics[scale=0.5]{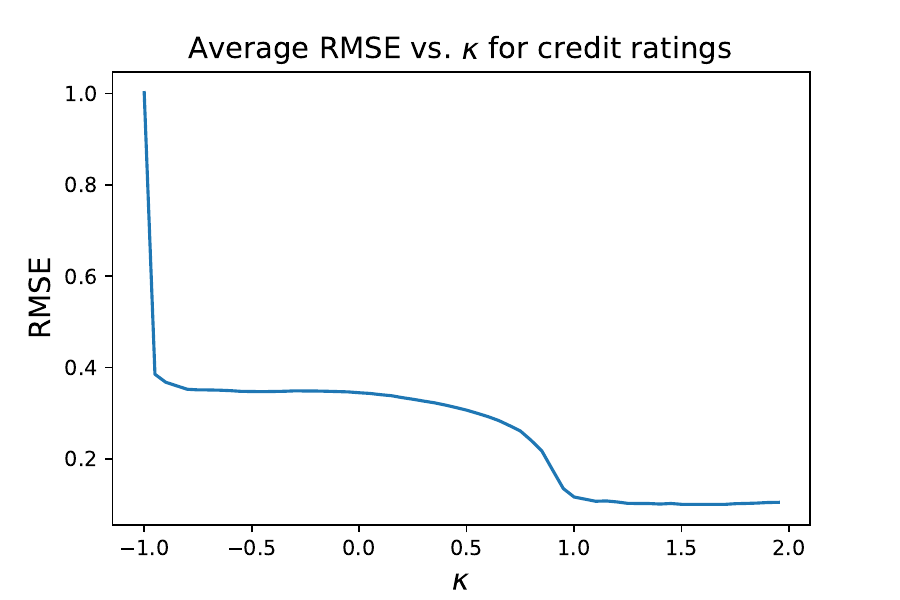}
\includegraphics[scale=0.5]{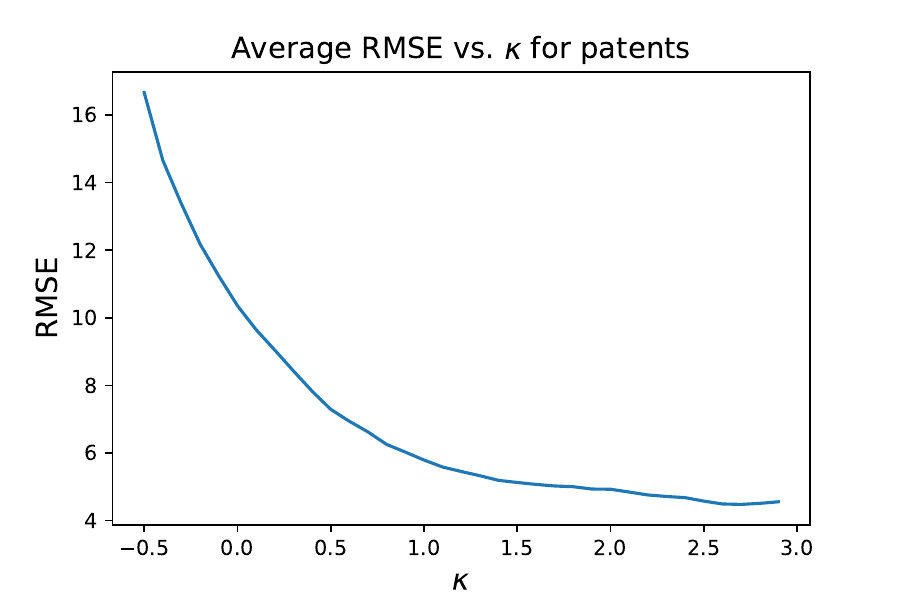}
\includegraphics[scale=0.5]{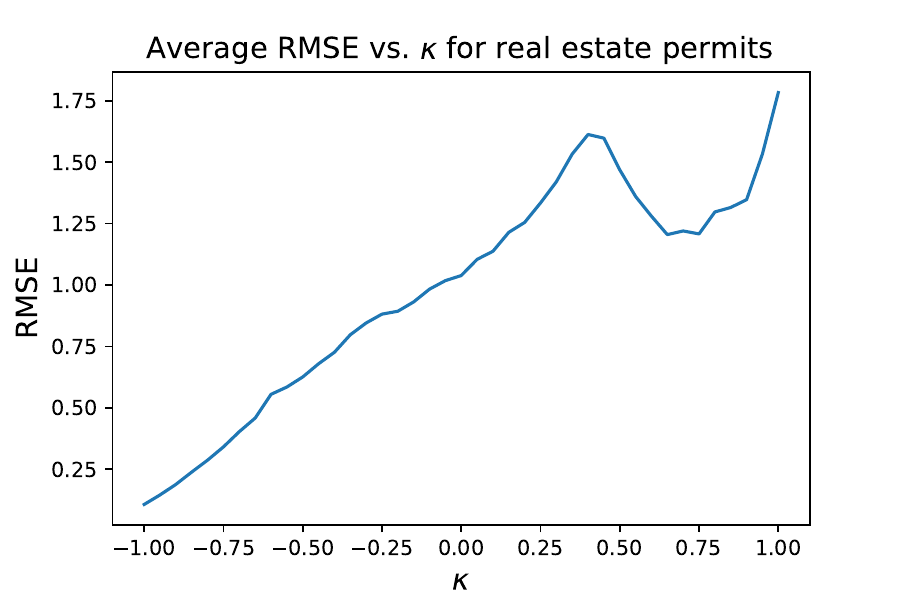}
\par\end{centering}
\caption{Average out of sample RMSEs of generalized PML estimators on corporate defaults (top left), credit ratings data (top right), corporate patents (bottom left), and residential permits (bottom right).}
\label{fig:RMSE}
\end{figure}

In this section, we investigate the performance of generalized PML estimators. To this end, we split the samples 80\%--20\% randomly into train and test sets. 
We estimate generalized PML estimators with varying $\kappa$ on the train set, and report the test set RMSEs averaged over 2000 random splittings of the data.  
When $\kappa\neq -1,0,1$, the corresponding generalized PML estimator is obtained by solving \eqref{eq:PML-optimization-problem} using the \texttt{scipy.optimize} module in Python. The average out of sample RMSEs for default rates are summarized in \cref{fig:RMSE}. We observe that for the corporate default data, the RMSE is optimized at $\kappa = 0.10$, with an average value of 0.435. Therefore, for the purpose of predicting corporate defaults in this setting, the conventional choice of Poisson PML estimator with $\kappa=0$ is close to the optimal choice. However, this is not always the case for other datasets and applications, as demonstrated by the results in \cref{fig:RMSE} for the three other datasets. For credit ratings and corporate patents, estimators with $\kappa>0$ are recommended with out-of-sample RMSEs about 50\% smaller than Poisson PML; for residential permits, the gamma PML is recommended with an out-of-sample RMSE that is about 90\% smaller than Poisson PML.

 It is therefore crucial to determine the optimal generalized PML estimator in a principled manner, e.g., using the cross validation procedure we propose. Moreover, for the default data, the RMSEs of generalized PML estimators with $\kappa \geq 0$ tend to outperform those with $\kappa < 0$. This provides evidence that the excess sparsity feature of the default data dominates efficiency considerations, leading us to favor estimators with smaller weights on samples with small conditional means. On the other hand, the results for credit rating prediction, summarized in \cref{fig:RMSE}, suggest that estimators with $\kappa > 0$ have a significantly better fit of the credit rating data compared to the Poisson or other generalized PML estimators. This demonstrates that, although estimators with $\kappa > 0$ like the NLS are generally considered inefficient when modeling non-negative outcome data, they could still be preferable in some cases. Lastly, in \cref{tab:estimates} we provide the point estimates and bootstrapped standard errors of some of the relevant variables from each dataset. We find that the optimal generalized PML estimate could differ significantly from the Poisson estimate, leading to qualitatively different economic conclusions. These results highlight the importance of appropriately choosing the regression model for non-negative sparse data instead of using a single model by default.
 
\begin{table}
\begin{centering}
\begin{tabular}{c|c|c|c|c}
 & Defaults & Ratings & Patents & Permits\tabularnewline
\hline 
Variable & S\&P500 return & debt to asset ratio& CEO overconfidence & log \#Airbnb \tabularnewline
\hline 
\hline 
\multirow{2}{*}{Poisson regression} & $-0.307^{\ast\ast\ast}$ & $-0.646$ & $0.617^{\ast\ast\ast}$ & 0.047\tabularnewline
 & (0.069) & (0.443) & (0.200) & (0.221)\tabularnewline
\hline 
\multirow{2}{*}{generalized PML} & $-0.302^{\ast\ast\ast}$ & $-2.565^{\ast\ast\ast}$ & 0.125 & $0.431^{\ast\ast}$\tabularnewline
 & (0.069) & (0.635) & (0.224) & (0.134)\tabularnewline
\end{tabular}
\par\end{centering}
\caption{Point estimates and bootstrapped standard errors of relevant standardized covariates from each of the four datasets: trailing 1-year return of the S\&P 500 index; long term debt to total asset ratio (\texttt{dltt\_at} in the data); indicator of whether a CEO is over-confident (\texttt{overconf\_ceo}), and log number of airbnbs in a residential area (\texttt{lnairbnb}).
In three datasets, the optimal generalized PML estimator yields qualitatively different economic conclusions than the Poisson model.
}
\label{tab:estimates}
\end{table}

Our results in this section demonstrate that the out-of-sample RMSE can be used as a general cross validation metric to determine the optimal $\kappa$ when modeling data with non-negative outcomes. For example, when estimating a gravity model on trade data, we can generate a plot similar to that in \cref{fig:RMSE}, and use it to assess the relative performance of the family estimators indexed by $\kappa \in \mathbb{R}$. In settings with strong heteroskedasticity, an estimator with $\kappa \leq 0$ should generally be preferred, while in settings with excess sparsity, an estimator with $\kappa \geq 0$ will be selected. We conclude by noting that other procedures to select $\kappa$ may be more appropriate depending on the modeling task. For example, we may determine the optimal $\kappa$ based on the bias-variance trade-off analysis in this paper. 
To determine the degree $\alpha$ of heteroskedasticity, we can employ hypothesis tests discussed in \citet{manning2001estimating}. To determine
the degree $\beta$ of asymmetric censoring, we can consider estimating a zero-inflated or hurdle model. We leave a systematic study of this proposal to future works.


\section{Conclusion}
\label{sec:conclusion}
In this paper, we show that existing approaches to modeling non-negative count and count-like data in finance can perform poorly on many datasets that jointly exhibit heteroskedasticity and sparsity.

To address this limitation, we develop a unified framework for quantifying the bias-variance trade-off arising from these two features. Within this framework, we introduce a novel class of estimators---generalized pseudo maximum likelihood (PML) estimators---that encompasses classical methods such as Poisson PML and nonlinear least squares (NLS). To guide empirical implementation, we propose a simple cross-validation procedure that allows  to select the estimator best suited to the structure of their specific dataset.

We validate the effectiveness and practical appeal of our approach through extensive numerical studies. Notably, while non-linear least squares (NLS) and other generalized PML estimators with $\kappa > 0$ are often viewed as inefficient under heteroskedasticity and over-dispersion---and therefore typically avoided in practice---we find that, in empirically relevant settings with pronounced sparsity, these estimators can outperform traditional alternatives such as Poisson regression, provided the associated optimization problems are properly addressed.
Our work therefore calls for a more systematic assessment of count-like data in finance and provides a principled approach to the modeling of such data in practice.


\bibliographystyle{elsarticle-harv} 
\bibliography{references}

\begin{thebibliography}{55}
\expandafter\ifx\csname natexlab\endcsname\relax\def\natexlab#1{#1}\fi
\providecommand{\url}[1]{\texttt{#1}}
\providecommand{\href}[2]{#2}
\providecommand{\path}[1]{#1}
\providecommand{\DOIprefix}{doi:}
\providecommand{\ArXivprefix}{arXiv:}
\providecommand{\URLprefix}{URL: }
\providecommand{\Pubmedprefix}{pmid:}
\providecommand{\doi}[1]{\href{http://dx.doi.org/#1}{\path{#1}}}
\providecommand{\Pubmed}[1]{\href{pmid:#1}{\path{#1}}}
\providecommand{\bibinfo}[2]{#2}
\ifx\xfnm\relax \def\xfnm[#1]{\unskip,\space#1}\fi
\bibitem[{Addoum et~al.(2023)Addoum, Ng and Ortiz-Bobea}]{addoum2023temperature}
\bibinfo{author}{Addoum, J.M.}, \bibinfo{author}{Ng, D.T.}, \bibinfo{author}{Ortiz-Bobea, A.}, \bibinfo{year}{2023}.
\newblock \bibinfo{title}{Temperature shocks and industry earnings news}.
\newblock \bibinfo{journal}{Journal of Financial Economics} \bibinfo{volume}{150}, \bibinfo{pages}{1--45}.
\bibitem[{Ai and Norton(2000)}]{ai2000standard}
\bibinfo{author}{Ai, C.}, \bibinfo{author}{Norton, E.C.}, \bibinfo{year}{2000}.
\newblock \bibinfo{title}{Standard errors for the retransformation problem with heteroscedasticity}.
\newblock \bibinfo{journal}{Journal of Health Economics} \bibinfo{volume}{19}, \bibinfo{pages}{697--718}.
\bibitem[{Akey and Appel(2021)}]{akey2021limits}
\bibinfo{author}{Akey, P.}, \bibinfo{author}{Appel, I.}, \bibinfo{year}{2021}.
\newblock \bibinfo{title}{The limits of limited liability: Evidence from industrial pollution}.
\newblock \bibinfo{journal}{The Journal of Finance} \bibinfo{volume}{76}, \bibinfo{pages}{5--55}.
\bibitem[{Amemiya(1985)}]{amemiya1985advanced}
\bibinfo{author}{Amemiya, T.}, \bibinfo{year}{1985}.
\newblock \bibinfo{title}{Advanced econometrics}.
\newblock \bibinfo{journal}{Harvard University Press} \bibinfo{volume}{2}, \bibinfo{pages}{153--161}.
\bibitem[{Anderson(1979)}]{anderson1979theoretical}
\bibinfo{author}{Anderson, J.E.}, \bibinfo{year}{1979}.
\newblock \bibinfo{title}{A theoretical foundation for the gravity equation}.
\newblock \bibinfo{journal}{American Economic Review} \bibinfo{volume}{69}, \bibinfo{pages}{106--116}.
\bibitem[{Anderson and Van~Wincoop(2003)}]{anderson2003gravity}
\bibinfo{author}{Anderson, J.E.}, \bibinfo{author}{Van~Wincoop, E.}, \bibinfo{year}{2003}.
\newblock \bibinfo{title}{Gravity with gravitas: A solution to the border puzzle}.
\newblock \bibinfo{journal}{American Economic Review} \bibinfo{volume}{93}, \bibinfo{pages}{170--192}.
\bibitem[{Bekkerman et~al.(2023)Bekkerman, Cohen, Kung, Maiden and Proserpio}]{bekkerman2023effect}
\bibinfo{author}{Bekkerman, R.}, \bibinfo{author}{Cohen, M.C.}, \bibinfo{author}{Kung, E.}, \bibinfo{author}{Maiden, J.}, \bibinfo{author}{Proserpio, D.}, \bibinfo{year}{2023}.
\newblock \bibinfo{title}{The effect of short-term rentals on residential investment}.
\newblock \bibinfo{journal}{Marketing Science} \bibinfo{volume}{42}, \bibinfo{pages}{819--834}.
\bibitem[{B{\"o}hning et~al.(1999)B{\"o}hning, Dietz, Schlattmann, Mendonca and Kirchner}]{bohning1999zero}
\bibinfo{author}{B{\"o}hning, D.}, \bibinfo{author}{Dietz, E.}, \bibinfo{author}{Schlattmann, P.}, \bibinfo{author}{Mendonca, L.}, \bibinfo{author}{Kirchner, U.}, \bibinfo{year}{1999}.
\newblock \bibinfo{title}{The zero-inflated poisson model and the decayed, missing and filled teeth index in dental epidemiology}.
\newblock \bibinfo{journal}{Journal of the Royal Statistical Society Series A: Statistics in Society} \bibinfo{volume}{162}, \bibinfo{pages}{195--209}.
\bibitem[{Boulton and Williford(2018)}]{boulton2018analyzing}
\bibinfo{author}{Boulton, A.J.}, \bibinfo{author}{Williford, A.}, \bibinfo{year}{2018}.
\newblock \bibinfo{title}{Analyzing skewed continuous outcomes with many zeros: A tutorial for social work and youth prevention science researchers}.
\newblock \bibinfo{journal}{Journal of the Society for Social Work and Research} \bibinfo{volume}{9}, \bibinfo{pages}{721--740}.
\bibitem[{Cameron and Trivedi(1990)}]{cameron1990regression}
\bibinfo{author}{Cameron, A.C.}, \bibinfo{author}{Trivedi, P.K.}, \bibinfo{year}{1990}.
\newblock \bibinfo{title}{Regression-based tests for overdispersion in the poisson model}.
\newblock \bibinfo{journal}{Journal of Econometrics} \bibinfo{volume}{46}, \bibinfo{pages}{347--364}.
\bibitem[{Cameron and Trivedi(2013)}]{cameron2013regression}
\bibinfo{author}{Cameron, A.C.}, \bibinfo{author}{Trivedi, P.K.}, \bibinfo{year}{2013}.
\newblock \bibinfo{title}{Regression Analysis of Count Data}. volume~\bibinfo{volume}{53}.
\newblock \bibinfo{publisher}{Cambridge University Press}.
\bibitem[{Card(2001)}]{card2001estimating}
\bibinfo{author}{Card, D.}, \bibinfo{year}{2001}.
\newblock \bibinfo{title}{Estimating the return to schooling: Progress on some persistent econometric problems}.
\newblock \bibinfo{journal}{Econometrica} \bibinfo{volume}{69}, \bibinfo{pages}{1127--1160}.
\bibitem[{Chang et~al.(2024)Chang, Koehler, Qu, Leskovec and Ugander}]{chang2024inferring}
\bibinfo{author}{Chang, S.}, \bibinfo{author}{Koehler, F.}, \bibinfo{author}{Qu, Z.}, \bibinfo{author}{Leskovec, J.}, \bibinfo{author}{Ugander, J.}, \bibinfo{year}{2024}.
\newblock \bibinfo{title}{Inferring dynamic networks from marginals with iterative proportional fitting}.
\newblock \bibinfo{journal}{arXiv preprint arXiv:2402.18697} .
\bibitem[{Chen and Roth(2024)}]{chen2024logs}
\bibinfo{author}{Chen, J.}, \bibinfo{author}{Roth, J.}, \bibinfo{year}{2024}.
\newblock \bibinfo{title}{Logs with zeros? some problems and solutions}.
\newblock \bibinfo{journal}{The Quarterly Journal of Economics} \bibinfo{volume}{139}, \bibinfo{pages}{891--936}.
\bibitem[{Cohn et~al.(2022)Cohn, Liu and Wardlaw}]{cohn2022count}
\bibinfo{author}{Cohn, J.B.}, \bibinfo{author}{Liu, Z.}, \bibinfo{author}{Wardlaw, M.I.}, \bibinfo{year}{2022}.
\newblock \bibinfo{title}{Count (and count-like) data in finance}.
\newblock \bibinfo{journal}{Journal of Financial Economics} \bibinfo{volume}{146}, \bibinfo{pages}{529--551}.
\bibitem[{Coles et~al.(2006)Coles, Daniel and Naveen}]{coles2006managerial}
\bibinfo{author}{Coles, J.L.}, \bibinfo{author}{Daniel, N.D.}, \bibinfo{author}{Naveen, L.}, \bibinfo{year}{2006}.
\newblock \bibinfo{title}{Managerial incentives and risk-taking}.
\newblock \bibinfo{journal}{Journal of Financial Economics} \bibinfo{volume}{79}, \bibinfo{pages}{431--468}.
\bibitem[{Correia et~al.(2020)Correia, Guimar{\~a}es and Zylkin}]{correia2020fast}
\bibinfo{author}{Correia, S.}, \bibinfo{author}{Guimar{\~a}es, P.}, \bibinfo{author}{Zylkin, T.}, \bibinfo{year}{2020}.
\newblock \bibinfo{title}{Fast poisson estimation with high-dimensional fixed effects}.
\newblock \bibinfo{journal}{The Stata Journal} \bibinfo{volume}{20}, \bibinfo{pages}{95--115}.
\bibitem[{Craig and Von~Peter(2014)}]{craig2014interbank}
\bibinfo{author}{Craig, B.}, \bibinfo{author}{Von~Peter, G.}, \bibinfo{year}{2014}.
\newblock \bibinfo{title}{Interbank tiering and money center banks}.
\newblock \bibinfo{journal}{Journal of Financial Intermediation} \bibinfo{volume}{23}, \bibinfo{pages}{322--347}.
\bibitem[{Dow and Norton(2003)}]{dow2003choosing}
\bibinfo{author}{Dow, W.H.}, \bibinfo{author}{Norton, E.C.}, \bibinfo{year}{2003}.
\newblock \bibinfo{title}{Choosing between and interpreting the heckit and two-part models for corner solutions}.
\newblock \bibinfo{journal}{Health Services and Outcomes Research Methodology} \bibinfo{volume}{4}, \bibinfo{pages}{5--18}.
\bibitem[{Duffie et~al.(2007)Duffie, Saita and Wang}]{duffie2007multi}
\bibinfo{author}{Duffie, D.}, \bibinfo{author}{Saita, L.}, \bibinfo{author}{Wang, K.}, \bibinfo{year}{2007}.
\newblock \bibinfo{title}{Multi-period corporate default prediction with stochastic covariates}.
\newblock \bibinfo{journal}{Journal of Financial Economics} \bibinfo{volume}{83}, \bibinfo{pages}{635--665}.
\bibitem[{Fang et~al.(2014)Fang, Tian and Tice}]{fang2014does}
\bibinfo{author}{Fang, V.W.}, \bibinfo{author}{Tian, X.}, \bibinfo{author}{Tice, S.}, \bibinfo{year}{2014}.
\newblock \bibinfo{title}{Does stock liquidity enhance or impede firm innovation?}
\newblock \bibinfo{journal}{The Journal of Finance} \bibinfo{volume}{69}, \bibinfo{pages}{2085--2125}.
\bibitem[{Frankel(1997)}]{frankel1997regional}
\bibinfo{author}{Frankel, J.A.}, \bibinfo{year}{1997}.
\newblock \bibinfo{title}{Regional trading blocs in the world economic system}.
\newblock \bibinfo{journal}{Institute for International Economics} .
\bibitem[{Frankel and Wei(1993)}]{frankel1993trade}
\bibinfo{author}{Frankel, J.A.}, \bibinfo{author}{Wei, S.J.}, \bibinfo{year}{1993}.
\newblock \bibinfo{title}{Trade blocs and currency blocs}.
\newblock \bibinfo{number}{11}, \bibinfo{publisher}{National Bureau of Economic Research}.
\bibitem[{Gourieroux et~al.(1984a)Gourieroux, Monfort and Trognon}]{gourieroux1984pseudoa}
\bibinfo{author}{Gourieroux, C.}, \bibinfo{author}{Monfort, A.}, \bibinfo{author}{Trognon, A.}, \bibinfo{year}{1984}a.
\newblock \bibinfo{title}{Pseudo maximum likelihood methods: Applications to poisson models}.
\newblock \bibinfo{journal}{Econometrica: Journal of the Econometric Society} , \bibinfo{pages}{701--720}.
\bibitem[{Gourieroux et~al.(1984b)Gourieroux, Monfort and Trognon}]{gourieroux1984pseudob}
\bibinfo{author}{Gourieroux, C.}, \bibinfo{author}{Monfort, A.}, \bibinfo{author}{Trognon, A.}, \bibinfo{year}{1984}b.
\newblock \bibinfo{title}{Pseudo maximum likelihood methods: Theory}.
\newblock \bibinfo{journal}{Econometrica: Journal of the Econometric Society} , \bibinfo{pages}{681--700}.
\bibitem[{Gratton et~al.(2007)Gratton, Lawless and Nichols}]{gratton2007approximate}
\bibinfo{author}{Gratton, S.}, \bibinfo{author}{Lawless, A.S.}, \bibinfo{author}{Nichols, N.K.}, \bibinfo{year}{2007}.
\newblock \bibinfo{title}{Approximate gauss--newton methods for nonlinear least squares problems}.
\newblock \bibinfo{journal}{SIAM Journal on Optimization} \bibinfo{volume}{18}, \bibinfo{pages}{106--132}.
\bibitem[{Greene(1994)}]{greene1994accounting}
\bibinfo{author}{Greene, W.H.}, \bibinfo{year}{1994}.
\newblock \bibinfo{title}{Accounting for excess zeros and sample selection in poisson and negative binomial regression models}.
\newblock \bibinfo{journal}{NYU working paper no. EC-94-10} .
\bibitem[{Hausman et~al.(1984)Hausman, Hall and Griliches}]{hausman1984econometric}
\bibinfo{author}{Hausman, J.}, \bibinfo{author}{Hall, B.H.}, \bibinfo{author}{Griliches, Z.}, \bibinfo{year}{1984}.
\newblock \bibinfo{title}{Econometric models for count data with an application to the patents-r \& d relationship}.
\newblock \bibinfo{journal}{Econometrica: Journal of the Econometric Society} , \bibinfo{pages}{909--938}.
\bibitem[{He and Tian(2013)}]{he2013dark}
\bibinfo{author}{He, J.J.}, \bibinfo{author}{Tian, X.}, \bibinfo{year}{2013}.
\newblock \bibinfo{title}{The dark side of analyst coverage: The case of innovation}.
\newblock \bibinfo{journal}{Journal of Financial Economics} \bibinfo{volume}{109}, \bibinfo{pages}{856--878}.
\bibitem[{Hirk et~al.(2022)Hirk, Vana and Hornik}]{hirk2022corporate}
\bibinfo{author}{Hirk, R.}, \bibinfo{author}{Vana, L.}, \bibinfo{author}{Hornik, K.}, \bibinfo{year}{2022}.
\newblock \bibinfo{title}{A corporate credit rating model with autoregressive errors}.
\newblock \bibinfo{journal}{Journal of Empirical Finance} \bibinfo{volume}{69}, \bibinfo{pages}{224--240}.
\bibitem[{Hirshleifer et~al.(2012)Hirshleifer, Low and Teoh}]{hirshleifer2012overconfident}
\bibinfo{author}{Hirshleifer, D.}, \bibinfo{author}{Low, A.}, \bibinfo{author}{Teoh, S.H.}, \bibinfo{year}{2012}.
\newblock \bibinfo{title}{Are overconfident ceos better innovators?}
\newblock \bibinfo{journal}{The Journal of Finance} \bibinfo{volume}{67}, \bibinfo{pages}{1457--1498}.
\bibitem[{Holford(1980)}]{holfold1980rates}
\bibinfo{author}{Holford, T.R.}, \bibinfo{year}{1980}.
\newblock \bibinfo{title}{The analysis of rates and of survivorship using log-linear models}.
\newblock \bibinfo{journal}{Biometrics} \bibinfo{volume}{36}, \bibinfo{pages}{299--305}.
\bibitem[{Hollingsworth et~al.(2024)Hollingsworth, Karbownik, Thomasson and Wray}]{hollingsworth2024gift}
\bibinfo{author}{Hollingsworth, A.}, \bibinfo{author}{Karbownik, K.}, \bibinfo{author}{Thomasson, M.A.}, \bibinfo{author}{Wray, A.}, \bibinfo{year}{2024}.
\newblock \bibinfo{title}{The gift of a lifetime: The hospital, modern medicine, and mortality}.
\newblock \bibinfo{journal}{American Economic Review} \bibinfo{volume}{114}, \bibinfo{pages}{2201--2238}.
\bibitem[{King(1988)}]{king1988statistical}
\bibinfo{author}{King, G.}, \bibinfo{year}{1988}.
\newblock \bibinfo{title}{Statistical models for political science event counts: Bias in conventional procedures and evidence for the exponential poisson regression model}.
\newblock \bibinfo{journal}{American Journal of Political Science} , \bibinfo{pages}{838--863}.
\bibitem[{King(1989)}]{king1989event}
\bibinfo{author}{King, G.}, \bibinfo{year}{1989}.
\newblock \bibinfo{title}{Event count models for international relations: Generalizations and applications}.
\newblock \bibinfo{journal}{International Studies Quarterly} \bibinfo{volume}{33}, \bibinfo{pages}{123--147}.
\bibitem[{Lambert(1992)}]{lambert1992zero}
\bibinfo{author}{Lambert, D.}, \bibinfo{year}{1992}.
\newblock \bibinfo{title}{Zero-inflated poisson regression, with an application to defects in manufacturing}.
\newblock \bibinfo{journal}{Technometrics} \bibinfo{volume}{34}, \bibinfo{pages}{1--14}.
\bibitem[{Manning(1998)}]{manning1998logged}
\bibinfo{author}{Manning, W.G.}, \bibinfo{year}{1998}.
\newblock \bibinfo{title}{The logged dependent variable, heteroscedasticity, and the retransformation problem}.
\newblock \bibinfo{journal}{Journal of Health Economics} \bibinfo{volume}{17}, \bibinfo{pages}{283--295}.
\bibitem[{Manning and Mullahy(2001)}]{manning2001estimating}
\bibinfo{author}{Manning, W.G.}, \bibinfo{author}{Mullahy, J.}, \bibinfo{year}{2001}.
\newblock \bibinfo{title}{Estimating log models: to transform or not to transform?}
\newblock \bibinfo{journal}{Journal of Health Economics} \bibinfo{volume}{20}, \bibinfo{pages}{461--494}.
\bibitem[{McCullagh(2019)}]{mccullagh2019generalized}
\bibinfo{author}{McCullagh, P.}, \bibinfo{year}{2019}.
\newblock \bibinfo{title}{Generalized linear models}.
\newblock \bibinfo{publisher}{Routledge}.
\bibitem[{Mullahy(1986)}]{mullahy1986specification}
\bibinfo{author}{Mullahy, J.}, \bibinfo{year}{1986}.
\newblock \bibinfo{title}{Specification and testing of some modified count data models}.
\newblock \bibinfo{journal}{Journal of Econometrics} \bibinfo{volume}{33}, \bibinfo{pages}{341--365}.
\bibitem[{Mullahy(1997)}]{mullahy1997instrumental}
\bibinfo{author}{Mullahy, J.}, \bibinfo{year}{1997}.
\newblock \bibinfo{title}{Instrumental-variable estimation of count data models: Applications to models of cigarette smoking behavior}.
\newblock \bibinfo{journal}{Review of Economics and Statistics} \bibinfo{volume}{79}, \bibinfo{pages}{586--593}.
\bibitem[{Mullahy(1998)}]{mullahy1998much}
\bibinfo{author}{Mullahy, J.}, \bibinfo{year}{1998}.
\newblock \bibinfo{title}{Much ado about two: reconsidering retransformation and the two-part model in health econometrics}.
\newblock \bibinfo{journal}{Journal of Health Economics} \bibinfo{volume}{17}, \bibinfo{pages}{247--281}.
\bibitem[{Mullahy and Norton(2022)}]{mullahy2022transform}
\bibinfo{author}{Mullahy, J.}, \bibinfo{author}{Norton, E.C.}, \bibinfo{year}{2022}.
\newblock \bibinfo{title}{Why transform Y? A critical assessment of dependent-variable transformations in regression models for skewed and sometimes-zero outcomes}.
\newblock \bibinfo{type}{Technical Report}. National Bureau of Economic Research.
\bibitem[{Newey and McFadden(1994)}]{newey1994large}
\bibinfo{author}{Newey, W.K.}, \bibinfo{author}{McFadden, D.}, \bibinfo{year}{1994}.
\newblock \bibinfo{title}{Large sample estimation and hypothesis testing}.
\newblock \bibinfo{journal}{Handbook of econometrics} \bibinfo{volume}{4}, \bibinfo{pages}{2111--2245}.
\bibitem[{Park(1966)}]{park1966estimation}
\bibinfo{author}{Park, R.E.}, \bibinfo{year}{1966}.
\newblock \bibinfo{title}{Estimation with heteroscedastic error terms.}
\newblock \bibinfo{journal}{Econometrica} \bibinfo{volume}{34}.
\bibitem[{Santos~Silva and Tenreyro(2006)}]{silva2006log}
\bibinfo{author}{Santos~Silva, J.}, \bibinfo{author}{Tenreyro, S.}, \bibinfo{year}{2006}.
\newblock \bibinfo{title}{The log of gravity}.
\newblock \bibinfo{journal}{Review of Economics and Statistics} \bibinfo{volume}{88}, \bibinfo{pages}{641--658}.
\bibitem[{Santos~Silva and Tenreyro(2010)}]{silva2010existence}
\bibinfo{author}{Santos~Silva, J.}, \bibinfo{author}{Tenreyro, S.}, \bibinfo{year}{2010}.
\newblock \bibinfo{title}{On the existence of the maximum likelihood estimates in poisson regression}.
\newblock \bibinfo{journal}{Economics Letters} \bibinfo{volume}{107}, \bibinfo{pages}{310--312}.
\bibitem[{Santos~Silva and Tenreyro(2011)}]{silva2011further}
\bibinfo{author}{Santos~Silva, J.}, \bibinfo{author}{Tenreyro, S.}, \bibinfo{year}{2011}.
\newblock \bibinfo{title}{Further simulation evidence on the performance of the poisson pseudo-maximum likelihood estimator}.
\newblock \bibinfo{journal}{Economics Letters} \bibinfo{volume}{112}, \bibinfo{pages}{220--222}.
\bibitem[{Sautner et~al.(2023)Sautner, Van~Lent, Vilkov and Zhang}]{sautner2023firm}
\bibinfo{author}{Sautner, Z.}, \bibinfo{author}{Van~Lent, L.}, \bibinfo{author}{Vilkov, G.}, \bibinfo{author}{Zhang, R.}, \bibinfo{year}{2023}.
\newblock \bibinfo{title}{Firm-level climate change exposure}.
\newblock \bibinfo{journal}{The Journal of Finance} \bibinfo{volume}{78}, \bibinfo{pages}{1449--1498}.
\bibitem[{Vassalou and Xing(2004)}]{vassalou2004default}
\bibinfo{author}{Vassalou, M.}, \bibinfo{author}{Xing, Y.}, \bibinfo{year}{2004}.
\newblock \bibinfo{title}{Default risk in equity returns}.
\newblock \bibinfo{journal}{The Journal of Finance} \bibinfo{volume}{59}, \bibinfo{pages}{831--868}.
\bibitem[{White(1982)}]{white1982maximum}
\bibinfo{author}{White, H.}, \bibinfo{year}{1982}.
\newblock \bibinfo{title}{Maximum likelihood estimation of misspecified models}.
\newblock \bibinfo{journal}{Econometrica: Journal of the Econometric Society} , \bibinfo{pages}{1--25}.
\bibitem[{Windmeijer and Santos~Silva(1997)}]{windmeijer1997endogeneity}
\bibinfo{author}{Windmeijer, F.A.}, \bibinfo{author}{Santos~Silva, J.M.}, \bibinfo{year}{1997}.
\newblock \bibinfo{title}{Endogeneity in count data models: an application to demand for health care}.
\newblock \bibinfo{journal}{Journal of Applied Econometrics} \bibinfo{volume}{12}, \bibinfo{pages}{281--294}.
\bibitem[{Wooldridge(1999)}]{wooldridge1999quasi}
\bibinfo{author}{Wooldridge, J.M.}, \bibinfo{year}{1999}.
\newblock \bibinfo{title}{Quasi-likelihood methods for count data}.
\newblock \bibinfo{journal}{Handbook of applied econometrics volume 2: Microeconomics} , \bibinfo{pages}{321--368}.
\bibitem[{Xu and Kim(2022)}]{xu2022financial}
\bibinfo{author}{Xu, Q.}, \bibinfo{author}{Kim, T.}, \bibinfo{year}{2022}.
\newblock \bibinfo{title}{Financial constraints and corporate environmental policies}.
\newblock \bibinfo{journal}{The Review of Financial Studies} \bibinfo{volume}{35}, \bibinfo{pages}{576--635}.
\bibitem[{Zorn(1998)}]{zorn1998analytic}
\bibinfo{author}{Zorn, C.J.}, \bibinfo{year}{1998}.
\newblock \bibinfo{title}{An analytic and empirical examination of zero-inflated and hurdle poisson specifications}.
\newblock \bibinfo{journal}{Sociological Methods \& Research} \bibinfo{volume}{26}, \bibinfo{pages}{368--400}.

\end{thebibliography}

\newpage 
\begin{appendix}
\section{Pseudo Maximum Likelihood}

\subsection{Gravity Models}
A popular instance of data with non-negative outcome variables is economic trade data between countries or geographical locations. We start our discussions based on this instance, although they apply more generally to non-negative data arising in other domains. In \citet{silva2006log}, the authors focus on a class of multiplicative trade models, often referred to as ``gravity models'' \citep{anderson1979theoretical,anderson2003gravity}:
\begin{align}
\label{eq:gravity-equation}
Y_{ij}=\alpha_{0}G_{i}^{\alpha_{1}}G_{j}^{\alpha_{2}}D_{ij}^{\alpha_{3}}\eta\equiv\exp(\theta^{T}X_{ij})\cdot\eta,
\end{align}
where $Y_{ij}$ is the trade volume between countries $i$ and $j$, $G_{i}$ is the GDP of country $i$,
and $D_{ij}$ is the distance between the two countries. $\eta$ is a multiplicative error that accounts for the randomness in data that results in the imperfect fit $Y_{ij}\neq \alpha_{0}G_{i}^{\alpha_{1}}G_{j}^{\alpha_{2}}D_{ij}^{\alpha_{3}}$, with $\mathbb{E}(\eta\mid X)=1$. Assuming for now that $Y_{ij}>0$, a common approach to estimating this model is to run an OLS regression of the following log-linearized model
\begin{align*}
\log(Y) & =\theta^{T}X+\log\eta.
\end{align*}
The validity of this approach relies on the standard requirement that $\log \eta$ is uncorrelated with $X$. However, $\mathbb{E}(\log\eta\mid X)\neq \log\mathbb{E}(\eta\mid X)=0$ in general due to Jensen's inequality. The quantity $\mathbb{E}(\log\eta\mid X)$ depends on $X$ because the non-linear logarithm transformation applied to $\eta$ induces dependence of $\mathbb{E}(\log\eta\mid X)$ on higher moments of the conditional distribution of $\eta$ on $X$, not just its mean, unless $\eta$ is of a particular form. Therefore, the log-linear OLS regression approach can
be severely biased, which has been confirmed in experiments and theoretical analysis in \cite{silva2006log,cohn2022count}.

In the presence of heteroskedasticity, a popular alternative to estimating the non-linear model $Y = f(X,\theta) + \varepsilon$ 
%
is {pseudo maximum likelihood} (PML). This approach does not assume knowledge of the distribution of $\varepsilon$, only the correct specification of $\mathbb{E}(Y\mid X)=f(X,\theta)$. Under this condition, PML then minimizes a ``pseudo log-likelihood'', which is constructed assuming a specified probability distribution for $\varepsilon$. For example, the non-linear least squares estimator, which solves the problem
\begin{align}
\label{eq:NLS}
\min_\theta\sum_{i}(y_{i}-f(x_{i},\theta))^{2},
\end{align}
can be understood as a PML estimator that assumes a normal distribution of $\varepsilon$. While the resulting estimator achieves consistent estimation of $\theta$ under regularity conditions \citep{gourieroux1984pseudoa}, it is only efficient when $\varepsilon$ is \emph{homoskedastic}, i.e., $\text{Var}(\varepsilon\mid X)$ does not depend on $X$. If in addition we had knowledge of the functional form of the conditional variance $\text{Var}(\varepsilon\mid X)$, we can obtain the efficient weighted NLS estimator. However, such information is often not available in practice. In this case, the Poisson pseudo maximum likelihood (PPML) estimator is usually preferable to the NLS. It constructs the pseudo log-likelihood by assuming that $Y\mid X$ has a Poisson distribution with mean parameter $f(X,\theta)$. Even when $Y$ is not necessarily integer-valued, consistency can be achieved given the  correct specification of $f(X,\theta)$. Recall that Poisson distributions have equal mean and variance, so that the Poisson PML approach can be understood as assuming $\text{Var}(\varepsilon\mid X)=\mathbb{E}(\varepsilon\mid X)$. In fact, it is efficient among all PML estimators when the condition
$\text{Var}(\varepsilon\mid X)=\mathbb{E}(\varepsilon\mid X)$ is satisfied. Given our discussion on the intrinsic heteroskedasticity of non-negative data, we see that Poisson PML is preferable to NLS partly because it implicitly accounts for the heteroskedasticity with a conditional variance that \emph{increases} with the conditional mean. Indeed, despite its ability to achieve consistency under correct specification of $f(X,\theta)$, the NLS is known to perform very poorly in practice for non-negative outcome data compared to PPML under strong heteroskedasticity \citep{silva2006log}, due to its efficiency loss in this setting.

An additional advantage of PPML is that its pseudo log-likelihood function
\begin{align}
\label{eq:PPML}
\sum_{i}y_{i}\cdot\log f(x_i,\theta)-f(x_i,\theta)
\end{align}
is concave in $\theta$ whenever $f(X,\theta)$ is log-concave. Examples include the exponential function $\exp(\theta^{T}X)$ and the Gaussian. In contrast, the NLS objective \eqref{eq:NLS} is usually non-convex for non-linear $f(X,\theta)$. As a result, it may have multiple local maxima, and standard non-convex optimization packages may only be able to find local maximizers instead of the true NLS estimator. Besides heteroskedasticity, this distinction between NLS and Poisson PML in terms of the difficulty of solving the corresponding optimization problems likely also contributes to the poor performance of NLS estimators observed in practice, although it is less studied than the effects of heteroskedasticity.


Moreover, an important feature of trade data is that there could be many zero observations of $Y_{ij}$,
i.e., no trade between pairs of countries or small trade volumes that
are often \emph{rounded} to zero. This can also be interpreted as $D_{ij}=+\infty$.
In the presence of many zeros, \citet{silva2011further} demonstrate with extensive simulations that PPML is more robust compared to log-linear models, as well as other pseudo maximum likelihood methods.

Given the specification $Y=\exp(\theta^{T}X)+\epsilon$ with $Y\geq0$
and $\mathbb{E}(\epsilon\mid X)=0$, here is the intuition why the
model is inherently heteroskedastic. When $\exp(\theta^{T}X)$ is
small, in order to guarantee $Y\geq0$, the dispersion of $\epsilon$
around 0 must be very small, whereas for larger $\exp(\theta^{T}X)$,
the variance of $\epsilon$ is allowed to be larger. In other words,
the heteroskedasticity of the model is inherent in the specification
$Y=\exp(\theta^{T}X)+\epsilon$. 

Assuming $\text{Var}(\epsilon\mid X)\propto\mathbb{E}(\epsilon\mid X)=\exp(\theta^{T}X)$,
we can consider the efficient \emph{weighted} nonlinear least squares estimator
\begin{align*}
\min_\theta\sum_{i}\exp(-\theta^{T}x_{i})(y_{i}-\exp(\theta^{T}x_{i}))^{2}
\end{align*}
 which has first order equations $\sum_{i}\exp(\theta^{T}x_{i})x_{i}-x_{i}\exp(-\theta^{T}x_{i})y_{i}^{2}=0$,
which is asymptotically equivalent to the first order conditions of
the PPML objective 
\begin{align*}
\sum_{i}(y_{i}-\exp(\theta^{T}x_{i}))x_{i} & =0.
\end{align*}

In practice, when we do not know the variance, the unweighted NLS
is used. \citet{silva2006log} observe that the estimator resulting from maximizing the PPML
objective is much better compared to the NLS estimator, even when
$\text{Var}(\epsilon\mid X)\propto\mathbb{E}(\epsilon\mid X)$ fails, e.g.,
$\text{Var}(\epsilon\mid X)\propto\mathbb{E}(\epsilon\mid X)^{2}$. To summarize, \citet{silva2006log} observe that
PPML has superior 
 performance across a variety of
patterns of heteroskedasticity as well as in settings with many zero responses, whereas
alternative estimators such as NLS and log-linear OLS only perform
well when the specification of variance is correct. 

\subsection{Negative Binomial PML}

We also remark on the behavior of the negative binomial (NB) compared to the class of generalized PML estimators in \eqref{eq:generalized-PML}. Recall the estimating equations of NB PML:
\begin{align*}
\sum_{i}(y_{i}-\exp(\theta^{T}x_{i}))(1+\exp(\theta^{T}x_{i}))^{-1}x_{i} & =0.
\end{align*}
Compared to the Gamma PML, whose estimating equations are
\begin{align*}
\sum_{i}(y_{i}-\exp(\theta^{T}x_{i}))(\exp(\theta^{T}x_{i}))^{-1}x_{i} & =0,
\end{align*}
NB PML is similarly robust to high degrees of heteroskedasticity and
large outliers. However, since it does not place overly large weights on samples with small $\exp(\theta^{T}x_{i})$, it has the additional advantage of being more robust to excess sparsity than the Gamma PML. In fact, for samples near zero, it places similar weights 
as the Poisson PML. Therefore, NB PML should be preferable when both heteroskedasticity and excess sparsity are prominent. We can similarly perform a trade-off analysis for generalized PML estimators of the form 
\begin{align*}
\sum_{i}(y_{i}-\exp(\theta^{T}x_{i}))(1+\exp(\theta^{T}x_{i}))^{-\kappa}x_{i}=0,
\end{align*}
but in this paper we focus on the estimators defined by \eqref{eq:generalized-PML}, as the NB PML estimator cannot accommodate separable group fixed effects \citep{cohn2022count}, and is therefore not applicable in some applications in finance and economics.
\subsection{Heteroskedasticity vs. Sparsity}
We note that the ability to distinguish between heteroskedasticity and excess sparsity is also important, as they could lead to similar dispersion behaviors in the data, but require \emph{opposite} types of generalized PML estimators. For example, 
\citet{greene1994accounting} notes that excess sparsity can ``masquerade'' as over-dispersion. \citet{zorn1998analytic} also discusses the effect of zero-inflation and hurdle models on the dispersion of data. The presence of excess zeros could result in over-dispersion of the data when it is generated by zero-inflation, while hurdle process can result in either over- or under-dispersion. This property is important since in the presence of zero-inflation type censoring, we should use generalized PMLs with larger $\kappa$, while if we mistakenly attribute excess sparsity to over-dispersion, we would be misled to use generalized PMLs with smaller $\kappa$.

\section{Computational Considerations for Generalized PML Estimators}
\label{subsec:optimization}

In this section, we discuss the computational considerations of generalized PML estimators for $\kappa\in[0,1]$. We formulate optimization problems that give rise to the generalized PML estimating equations as first order conditions, which provides a straightforward method for computing the generalized PML estimators given data. Recall the family of generalized PML estimators that we consider
\begin{align}
\label{eq:estimating-equation}
\sum_{i}\left\{ y_{i}-\exp(\theta^{T}x_{i})\right\} \cdot\exp(\kappa\theta^{T}x_{i})\cdot x_{i} & =0.
\end{align}
 An important practical question is how one can solve it for arbitrary $\kappa$. One option is to apply root-finding algorithms such as Newton's algorithm. 
 Here we discuss another approach, which finds local optimizers of optimization problems whose first order conditions correspond to the estimating equations \eqref{eq:estimating-equation}.

For $\kappa=-1$ and 0, the objectives are those of the
familiar gamma and Poisson PML:
\begin{align*}
\max_{\theta}\sum_{i}-y_{i}\exp(-\theta^{T}x_{i})-\theta^{T}x_{i}\\
\max_{\theta}\sum_{i}-\exp(\theta^{T}x_{i})+y_{i}\theta^{T}x_{i},
\end{align*}
which are both concave in $\theta$. More generally, we can verify that when $\kappa \in (-1,0) \cup (0,1)$,  \eqref{eq:estimating-equation} is the first order condition of the problem
\begin{align}
\label{eq:PML-optimization-problem}
\max_{\theta}\sum_i\frac{1}{\kappa}y_{i}\exp(\kappa\theta^{T}x_{i})-\frac{1}{\kappa+1}\exp((\kappa+1)\theta^{T}x_{i}).
\end{align}
For $\kappa\in(-1,0)$, \eqref{eq:PML-optimization-problem} is concave in $\theta$, which is easy to solve. On the other hand, for $\kappa\in(0,1)$, \eqref{eq:PML-optimization-problem} is no longer concave in $\theta$, since the term $\frac{1}{\kappa}y_{i}\exp(\kappa\theta^{T}x_{i})$
is \emph{convex} in $\theta$ and the term $-\frac{1}{\kappa+1}\exp((\kappa+1)\theta^{T}x_{i})$ is concave in $\theta$. In this case, we can still use non-convex optimization software to find local maxima, such as \texttt{scipy}'s \texttt{optimize} module
or the Gauss-Newton method \citep{gratton2007approximate}. In our simulations, we find that the \texttt{scipy.optimize} module 
has comparable accuracy as \texttt{CVX} when $\kappa<0$, and behaves well for $\kappa>0$. Note also that the global maximizer of \eqref{eq:PML-optimization-problem} is bounded, so that its first order condition is necessary for optimality. Moreover, as long as $x_i$ can take both positive and negative values, the solution to the first order condition is also bounded. To see this point more clearly, consider the first order condition
\[\sum_{i}\left\{ y_{i}-\exp(\theta^{T}x_{i})\right\} \cdot\exp(\kappa\theta^{T}x_{i})\cdot x_{i}=0.\]
When $-1<\kappa<0$, we can verify that with $\theta \rightarrow \pm \infty$ the left hand side cannot attain 0, regardless of the signs of $x_i$. However, when $\kappa > 0$, if $x_i\geq0$ for all $i$ or $x_i\leq0$ for all $i$, then the left hand side can attain a value of zero when $\theta \rightarrow -\infty$ or $\infty$. However, when some $x_i$ is positive and some $x_i$ is negative, the left hand side cannot attain 0 at infinity. We can therefore better understand the generalized PML estimators as solutions to \eqref{eq:PML-optimization-problem} when $\kappa>0$. Note also that when $\kappa=1$,
both optimization problems
\begin{align*}
\max_{\theta}y_{i}\exp(\theta^{T}x_{i})-\frac{1}{2}\exp(2\theta^{T}x_{i})\Longleftrightarrow
\min_{\theta}(y_{i}-\exp(\theta^{T}x_{i}))^{2}
\end{align*}
 yield the same NLS estimator, by expanding the square in the second
objective and removing constant terms. 
Intuitively, large values of $y_{i}$ will
make the NLS objective ``less'' concave, since the convex term $y_{i}\exp(\kappa\theta^{T}x_{i})$
will be large. On the other hand, many zero $y_i$'s will result in the concave term  $-\frac{1}{2}\exp(2\theta^{T}x_{i})$ dominating, leading to an objective that is closer to being concave. This argument provides some intuition on why NLS could actually be preferable to the Poisson PML.

Lastly, note that compared to $\kappa=-1$ (gamma PML), it is preferable
to use $\kappa\in(-1,0)$ since the objective $\max_{\theta}\frac{1}{\kappa}y_{i}\exp(\kappa\theta^{T}x_{i})-\frac{1}{\kappa+1}\exp((\kappa+1)\theta^{T}x_{i})$
will still retain strict concavity when there is a large number of
zeros. In comparison, the gamma PML has objective 
\begin{align*}
    \max_{\theta}\sum_{i}-y_{i}\exp(-\theta^{T}x_{i})-\theta^{T}x_{i},
\end{align*}
which could become close to linear (have a ``flat'' loss function) when many $y_i$ are zero, resulting in poor solution quality of the optimization problem. This feature may partially
explain the poor performance of gamma PML in the presence of many
zeros, even under strong over-dispersion.

\section{Discussion of Bias and Variance Results}
\label{sec:bias-variance-discussion}
In this section, we provide a more detailed discussion of our results on the bias and variance of the generalized PML estimators. Recall the bias formula derived in \cref{thm:bias}:
\begin{align*}
       \hat{\theta}-\theta_{0}	=(A^{T}A)^{-1}Ab + o_p(1),
   \end{align*}
   where, with $P(\theta_0,x)$ the censoring probability in \eqref{eq:censoring-model}, $A$ and $b$ are given by \begin{align*}
       \begin{split}
           A=\int(P(\theta_{0},x)\kappa-1)\cdot xx^{T}p(x)\exp((\kappa+1)\theta_{0}^{T}x)dx,\\
b	=\int P(\theta_{0},x)\cdot x^{T}p(x)\exp((\kappa+1)\theta_{0}^{T}x)dx.
       \end{split}   \end{align*}
We can use \cref{thm:bias} to estimate the bias of generalized PML estimators under censoring. For example, if $\kappa=-1$, which corresponds to Gamma PML, we have
\begin{align*}
A= & \int(-1-P(\theta_{0},x))\cdot xx^{T}p(x)dx,\\
b & =\int P(\theta_{0},x)\cdot xp(x)dx,
\end{align*}
 which can be used to explicitly quantify the bias given a particular
$P(\theta_{0},x)$, such as $P(\theta_{0},x)=\frac{1}{\exp(\tau \theta_{0}^{T}x)+1}$, as in the zero-inflation model of \citet{lambert1992zero}. In the absence of censoring, i.e., $ P(\theta_{0},x)\equiv 0$, \cref{thm:bias} implies that generalized PML estimators are all consistent.  

Under standard regularity conditions, the limit of $Z$-estimators converges to the solution $\tilde{\theta}_{0}$ to the population limit of their estimating equations, which for \eqref{eq:generalized-PML} is given by 
\begin{align*}
\mathbb{E}(\psi(y,x,\theta))=\mathbb{E}(y-\exp(\theta^{T}x))\cdot\exp(\kappa\theta^{T}x)\cdot x & =0,
\end{align*}
In the absence of censoring, we have
\begin{align*}
    \mathbb{E}(y-\exp(\theta^{T}x))\cdot\exp(\kappa\theta^{T}x)\cdot x = \mathbb{E}_x(\exp(\theta_0^{T}x)-\exp(\theta^{T}x))\cdot\exp(\kappa\theta^{T}x)\cdot x,
\end{align*}
which evaluates to 0 when $\theta=\theta_0$, ensuring the consistency of generalized PML estimators.
However, in the presence of censoring with probability $P(\theta_{0},x)$, the conditions that characterize the limit of the generalized PML estimators become
\begin{align*}
    \mathbb{E}_{x}\left\{ (1-P(\theta_{0},x))\cdot\exp(\theta_{0}^{T}x)-\exp(\theta^{T}x)\right\} \cdot\exp(\kappa\theta^{T}x)\cdot x=0,
\end{align*}
which is no longer solved by $\theta=\theta_0$. We use $\tilde{\theta}_{0}$ to differentiate from the true
parameter $\theta_{0}$ when censoring results in bias. 

Note that the terms $A$ and $b$ do not depend on the heteroskedasticity parameter $\alpha$ and only on $\kappa$ and the censoring probability $P(\theta_0,X)$. This is expected because the censoring probabilities are defined in terms of the conditional mean $\exp(\theta_{0}^{T}X)$ only, which does not depend on $\alpha$, which controls the variance of $\varepsilon$. If we assume an alternative censoring mechanism, where the censoring probability depends on the \emph{realized} outcome $\exp(\theta_{0}^{T}X) +\varepsilon$, then the bias will depend on $\alpha$ as well.  On the other hand, as shown in \cref{thm:asymptotic-variance}, the asymptotic variance $J^{-1}IJ^{-1}$ will depend on both $\alpha$ and $\kappa$, as well as $P(\theta_0,X)$:
\begin{align*}
\begin{split}
    I & =\mathbb{E}_{x}xx^{T}\exp(2\kappa\theta^{T}x)\cdot\left[P(\theta_{0},x)\left(\exp(\theta^{T}x)\right)^{2}+(1-P(\theta_{0},x))\left(\exp(\theta_{0}^{T}x)-\exp(\theta^{T}x)\right)^{2}\exp^{}(\alpha\theta_{0}^{T}x)\right]\\
J & =\mathbb{E}_{x}\left[-\left((1-P(\theta_{0},x))\cdot\exp(\theta_{0}^{T}x)-\exp(\theta^{T}x)\right)\cdot\exp(\kappa\theta^{T}x)\cdot\kappa xx^{T}+\exp(\kappa\theta^{T}x)\exp(\theta^{T}x)xx^{T}\right],
\end{split}
\end{align*}
evaluated at $\theta=\tilde{\theta}_{0}$, which solves the population limit of the estimating equations \eqref{eq:generalized-PML}:
\begin{align*}
\mathbb{E}_{x}\left\{ (1-P(\theta_{0},x))\cdot\exp(\theta_{0}^{T}x)-\exp(\theta^{T}x)\right\} \cdot\exp(\kappa\theta^{T}x)\cdot x=0.
\end{align*}
\cref{thm:asymptotic-variance} relies on the classic result on the asymptotic variance of $Z$-estimators of the form $1/n\sum_i \psi(y_i,x_i,\theta)$, which is given by $J^{-1}IJ^{-1}$ with
\begin{align*}
I & =\mathbb{E}_{x}\mathbb{E}\left[\psi(y,x,\tilde{\theta}_{0})\psi(y,x,\tilde{\theta}_{0})^{T}\mid x\right],\\
J & =\mathbb{E}_{x}\mathbb{E}\left[-\frac{\partial\psi(y,x,\tilde{\theta}_{0})}{\partial\theta}\mid x\right].
\end{align*}
Note that we need to evaluate $I,J$ at the limit $\tilde{\theta}_{0}$ of the generalized PML estimators, which under the censoring model solves \eqref{eq:population-moments}. In practice, we can form the plug-in estimator of the variance by replacing $\theta$ with the generalized PML estimator $\hat{\theta}_0$ solving \eqref{eq:generalized-PML} and replacing the integral with the sample average in \eqref{eq:variance-component}. Note that unlike the bias, the asymptotic variance depends on all three components of our framework: $\alpha$, which governs heteroskedasticity, $P(\theta_0,x)$, which describes the censoring model, and $\kappa$, which determines the generalized PML estimator. The bias and asymptotic variance formulae can be leveraged to analyze the bias-variance trade-off under heteroskedasticity and censoring. To be concrete, here we focus on the NLS ($\kappa=1$) and Poisson PML ($\kappa=0$) estimators in the one-dimensional setting, but our analysis can be extended to the general case. We have the following formula.

\begin{cor}
\label{cor:NLS-Poisson-bias}
Assume that the parameter $\theta$ is one-dimensional. Near $\theta_{0}$, the biases of NLS and Poisson PML estimators are approximated by \zq{compute bias and variance of these estimators using the same data setting used to produce the previous table, to show how bad these estimators are. This would highlight why they are not performing well (because of either variance or bias) caused by heteroskedasticity and variance.}
\begin{align*}
\tilde \theta_{PPML}-\theta_{0} & =-\frac{\int p(X)X\cdot P(\theta_{0},X)\cdot\exp(\theta_{0}^{T}X)dX}{\int p(X)X^{2}\cdot\exp(\theta_{0}^{T}X)dX},\\
\tilde \theta_{NLS}-\theta_{0} & =-\frac{\int p(X)X\cdot P(\theta_{0},X)\cdot\exp(2\theta_{0}^{T}X)dX}{\int p(X)X^{2}\left(1+P(\theta_{0},X)\right)\exp(2\theta_{0}^{T}X)dX}.
\end{align*}
\end{cor}
Comparing the bias of Poisson PML and NLS, we can see that the NLS is more robust to censoring toward zero. Consider for example the censoring model where 
\begin{align*}
P(\theta_{0},X) = & \begin{cases}
1 & \exp(\theta_{0}^{T}X)\leq c\\
0 & \text{otherwise.}
\end{cases}
\end{align*}
This threshold censoring model rounds down samples with bounded conditional mean to zero, and can be viewed as the limit $\beta \rightarrow \infty$ of the censoring models in \eqref{eq:heteroskedasticity-general}, with $c=1$. When $c\leq1$, we have 
\begin{align*}
    \int p(X)X\cdot P(\theta_{0},X)\cdot\exp(\theta_{0}^{T}X)dX &= \int_{\exp(\theta_{0}^{T}X)\leq c} p(X)X\cdot\exp(\theta_{0}^{T}X)dX\\
    \geq \int_{\exp(2\theta_{0}^{T}X)\leq c} p(X)X\cdot\exp(\theta_{0}^{T}X)dX
     &=\int p(X)X\cdot P(\theta_{0},X)\cdot\exp(2\theta_{0}^{T}X)dX,
\end{align*}
assuming that $X\in \mathbb{R}_+$, while 
\begin{align*}
    \int p(X)X^{2}\cdot\exp(\theta_{0}^{T}X)dX & \leq \int p(X)X^{2}\exp(2\theta_{0}^{T}X)dX \leq \int p(X)X^{2}\left(1+P(\theta_{0},X)\right)\exp(2\theta_{0}^{T}X)dX
\end{align*}
as long as the tail of $p(X)$ decays fast enough. \cref{cor:NLS-Poisson-bias} then implies that $|\tilde \theta_{NLS}-\theta_{0}|\leq|\tilde \theta_{PPML}-\theta_{0}|$.

In the presence of censoring, the asymptotic variance of PPML and NLS can also be computed following \cref{thm:asymptotic-variance}. For PPML, we can modify $I,J$ as follows:
\begin{align*}
I_{PPML} =\mathbb{E}_{x}\mathbb{E}\left[(y-\exp(\theta_{0}^{T}x))^{2}\cdot xx^{T}\mid x\right]&=\mathbb{E}_{x}xx^{T}\left[\exp(\alpha\theta_{0}^{T}x)+(\exp(2\theta^{T}x)-\exp(\alpha\theta_{0}^{T}x))P(\theta_{0},x)\right],
\end{align*}
where we note that the first term is the usual variance term for the DGP without censoring. Interestingly, if $\alpha=2$, the variance of the censoring model coincides with that of the uncensored model, regardless of the actual censoring mechanism. 
When $\alpha=0$, the variance of NLS is smaller than that of PPML. As $\alpha$ increases from 0 to 2, the variance of NLS starts to increase, while the variance of PPML initially decreases and then increases, but always staying below that of the NLS. These behaviors lead to a particular value of $\overline{\alpha}$ where the MSE of NLS starts to dominate that of PPML. Therefore, below this level $\overline{\alpha}$ of heteroskedasticity, NLS should in fact be preferred than PPML, even though it may have larger variance.

\section{Generalized PML Estimators for General Conditional Mean Models}
Instead of assuming $f(X,\theta)=\exp(\theta^T X)$, we can consider the more general specification $f(X,\theta)=f(\theta^T X)$.  Then the general first order conditions of the NLS, Poisson, Gamma, and negative
binomial PMLs, are given by \citet{gourieroux1984pseudob}:
\begin{align*}
\text{NLS:}\quad \sum_{i}(f'(\theta^{T}x_{i})-y_{i}f'(\theta^{T}x_{i})/f(\theta^{T}x_{i}))\cdot f(\theta^{T}x_{i})\cdot x_{i} & =0\\
\text{Poisson:}\quad \sum_{i}(f'(\theta^{T}x_{i})-y_{i}f'(\theta^{T}x_{i})/f(\theta^{T}x_{i}))\cdot x_{i} & =0\\
\text{Gamma:}\quad \sum_{i}(f'(\theta^{T}x_{i})-y_{i}f'(\theta^{T}x_{i})/f(\theta^{T}x_{i}))\cdot\frac{1}{f(\theta^{T}x_{i})}\cdot x_{i} & =0\\
\text{Negative Bionomial:}\quad \sum_{i}(f'(\theta^{T}x_{i})-y_{i}f'(\theta^{T}x_{i})/f(\theta^{T}x_{i}))\cdot\frac{1}{1+f(\theta^{T}x_{i})}\cdot x_{i} & =0.
\end{align*}
More generally, we can consider a continuous family of estimators indexed by $\kappa \in \mathbb{R}$ that solve the estimating
equations 
\begin{align}
\sum_{i}(f'(\theta^{T}x_{i})-y_{i}f'(\theta^{T}x_{i})/f(\theta^{T}x_{i}))\cdot f^{\kappa}(\theta^{T}x_{i})\cdot x_{i} & =0,
\end{align}
as well as 
\begin{align}
\sum_{i}(f'(\theta^{T}x_{i})-y_{i}f'(\theta^{T}x_{i})/f(\theta^{T}x_{i}))\cdot (1+f(\theta^{T}x_{i}))^{\kappa}\cdot x_{i} & =0.
\end{align}
When $\kappa\geq0$, the resulting estimator is robust
to measurement errors that bias towards 0, while when $\kappa\leq0$,
the resulting estimator is robust to measurement errors that bias
towards $\infty$. As a result, the Poisson estimating equation $\sum_{i}(f'(\theta^{T}x_{i})-y_{i}f'(\theta^{T}x_{i})/f(\theta^{T}x_{i}))\cdot x_{i}=0$
corresponding to $\kappa=0$ is a reasonable compromise that is robust
to both types of measurement errors. However, depending on how the conditional variance varies with the conditional
mean. If the variance increases, then NLS maybe be robust to bias
towards 0, but not to heteroskedasticity. However, based on observations
if 
\begin{align*}
V(Y\mid X)=\mathbb{E}[Y_{i}\mid X_{i}],
\end{align*}
 NLS may still be ok, but if 
\begin{align*}
V(Y\mid X)=\mathbb{E}^{2}[Y_{i}\mid X_{i}],
\end{align*}
NLS performs very poorly. So it could be possible that when $V$ does
not grow too fast, the benefit of NLS that are robust to small outliers
can outweigh the disadvantage of not being robust to heteroskedastic
errors. For example, in the simulation results in Table 1 of \citet{silva2006log}, the NLS performs the best even when $V(Y\mid X)=\mathbb{E}[Y_{i}\mid X_{i}]$
and there are rounding errors to 0. But its performance significantly
deteriorates when $V(Y\mid X)=\mathbb{E}^{2}[Y_{i}\mid X_{i}]$. 

If we consider a graph where the x axis is how variance increases
with conditional mean, and the y axis is how much measurement bias
towards 0 is, we will observe a phase transition where in some regions
each of the estimators will be preferable. 

Can also consider the Negative Binomial 
\begin{align*}
\sum_{i}(y_{i}-\exp(\theta^{T}x_{i}))(1+\exp(\theta^{T}x_{i}))^{-1}x_{i} & =0.
\end{align*}
 Compared with Gamma, 
\begin{align*}
\sum_{i}(y_{i}-\exp(\theta^{T}x_{i}))(\exp(\theta^{T}x_{i}))^{-1}x_{i} & =0,
\end{align*}
 NB is similarly robust to high degrees of heteroskedasticity and
large outliers. However it has the additional advantage of being robust
to small outliers, just like the Poisson PML. Similarly,
if we do not want to downweight large samples as much, an estimator
satisfying 
\begin{align*}
\sum_{i}(y_{i}-\exp(\theta^{T}x_{i}))(1+\exp(\theta^{T}x_{i}))^{-\alpha}x_{i}=0
\end{align*}
 for some $\alpha\in(0,1)$ will be closer in behavior to Poisson
for larger samples. 

Therefore, we may understand the pseudo ML estimators in terms of how they weight observations using the value of the conditional mean $f(\theta^TX)$. There are at least three sources of mis-specification/measurement errors
that can impact the performance of PML estimators: 

1. Heteroskedasticity, most commonly variance increasing with mean.
Estimators that \textbf{downweight }large samples generally have better
performance;

2. Measurement bias towards 0, i.e., small responses being recorded
as 0. Estimators that \textbf{upweight} large samples generally have
better performance;

3. Measurement bias towards $\infty$, i.e., large responses being
distorted more. For example, if a measurement error of the form $(1+\alpha)Y_{i}$
results in large samples being distorted more. Estimators that \textbf{downweight
}large samples generally have better performance. This type of measurement bias could be behind data where there are some large outliers.

We therefore see that depending on the severity of each issue, it
is reasonable to expect different PML estimators can perform differently.
We therefore propose to jointly analyze the trade-offs of these issues.
The key components of our analysis will be the degree of heteroskedasticity
and measurement bias towards 0 (and $\infty$). Given these parameters
of the problem, there will be a particular PML that is optimal. 

\section{Additional Numerical Results}

\subsection{Accuracy of Bias Formula}
In this section, we investigate the accuracy of the bias formula derived in \cref{sec:trade-off}. Since they are based on the population estimating equations 
\begin{align*}
\mathbb{E}_{x}\left\{ (1-P(\theta_{0},x))\cdot\exp(\theta_{0}^{T}x)-\exp(\theta^{T}x)\right\} \cdot\exp(\kappa\theta^{T}x)\cdot x & =0,
\end{align*}
we first verify through simulations that they accurately characterize the expectation of the PML estimators when data is subject to censoring. For each replication, we compute
PML estimators $\hat{\theta}_{\kappa}$ based on 1000 i.i.d. samples,
and average over 1000 replications to simulate the expectation $\mathbb{E}\hat{\theta}_{\kappa}$
of the PML estimators. Then we plug $\mathbb{E}\hat{\theta}_{\kappa}$
into $\theta$ in the above equation, and evaluate the equations by
replacing $\mathbb{E}_{x}$ with average over 10000 samples of $x$. We find that  the moment equations evaluate to similar values near zero compared to the reference case without censoring.

We report in \cref{fig:check-moment} the values of each coordinate of the left hand side for $\kappa=1$ (NLS), $\kappa=0$ (Poisson), and $\kappa=-0.95$, for DGPs with varying values of $\beta$, with $\tau=1$ and $\alpha=1$. We choose $\kappa=-0.95$ instead of $\kappa=-1$ (gamma PML) because in the presence of asymmetric censoring, optimization solvers may fail to find an accurate solution of the gamma PML problem $\max_\theta\sum_{i}-\theta^{T}x_{i}-y_{i}\exp(-\theta^{T}x_{i})$, which becomes flat when many $y_i$'s are zero. We still use gamma PML to refer to the estimator with $\kappa=-0.95$ in the plots for convenience.

\begin{figure}[ht]
\begin{centering}
\includegraphics[scale=0.4]{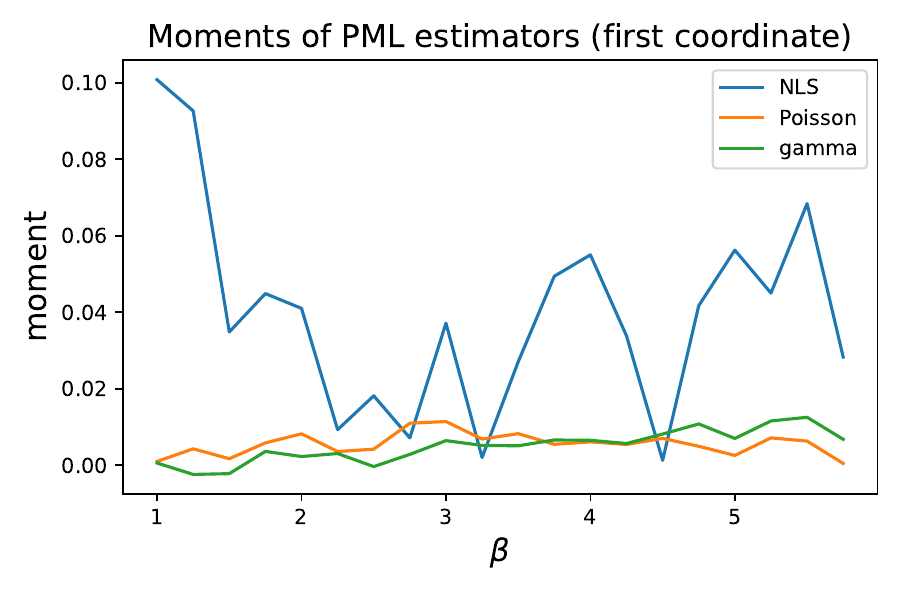} \includegraphics[scale=0.4]{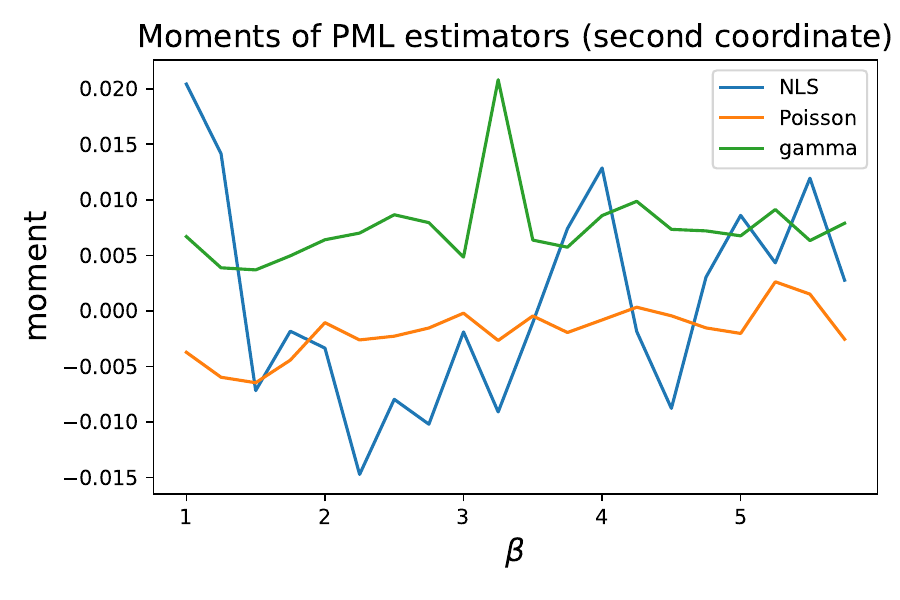}
\par\end{centering}
\caption{Values of $\mathbb{E}_{x}\left\{ (1-P(\theta_{0},x))\cdot\exp(\theta_{0}^{T}x)-\exp(\theta^{T}x)\right\} \cdot\exp(\kappa\theta^{T}x)\cdot x$ for three PML estimators. DGP has $P(X_i,\theta_0) = {1}/{(1+\exp(\beta\theta_{0}^{T}x))}$ and $\alpha=1$. As a reference, when there is no censoring, the population estimating equations have numerical values equal to $(0.010, 0.004)$ for NLS, $(0.0008,0.0006)$ for Poisson, and $(-0.022, 0.026)$ for gamma. \zq{omit plots}}
\label{fig:check-moment}
\end{figure}
In \cref{fig:check-moment}, we find that the moments are close to 0 and comparable to the case when there is no censoring, which confirms the validity of the population moment conditions.

\begin{table}[ht]
\begin{centering}
\begin{tabular}{c|c|c|c}
estimator & censoring & moment 1 & moment 2\tabularnewline
\hline 
\hline 
NLS & no  & 0.010 & 0.004\tabularnewline
\hline 
Poisson & no  & 0.0008 & 0.0006\tabularnewline
\hline 
gamma & no  & -0.022 & 0.026\tabularnewline
\hline 
NLS & yes $(\beta=0.25)$ & 1.111 & 0.298\tabularnewline
\hline 
Poisson & yes $(\beta=0.25)$ & 0.022 & 0.004\tabularnewline
\hline 
gamma & yes $(\beta=0.25)$ & -0.004 & 0.0004\tabularnewline
\hline 
NLS & yes $(\beta=1)$ & 0.461 & 0.142\tabularnewline
\hline 
Poisson & yes $(\beta=1)$ & 0.005 & 0.002\tabularnewline
\hline 
gamma & yes $(\beta=1)$ & -0.022 & 0.014\tabularnewline
\hline 
NLS & yes $(\beta=5)$ & 0.049 & 0.023\tabularnewline
\hline 
Poisson & yes $(\beta=5)$ & -0.003 & 0.0007\tabularnewline
\hline 
gamma & yes $(\beta=5)$ & -0.007 & 0.002\tabularnewline
\hline 
NLS & yes $(\beta=9)$ & -0.030 & 0.006\tabularnewline
\hline 
Poisson & yes $(\beta=9)$ & -0.006 & -0.001\tabularnewline
\hline 
gamma & yes $(\beta=9)$ & -0.006 & 0.001\tabularnewline
\end{tabular}
\par\end{centering}
\caption{Values of $\mathbb{E}_{x}\left\{ (1-P(\theta_{0},x))\cdot\exp(\theta_{0}^{T}x)-\exp(\theta^{T}x)\right\} \cdot\exp(\kappa\theta^{T}x)\cdot x$
for three PML estimators. DGP has $\alpha=2$ and $\beta=1$. Note
that when $\beta\rightarrow\infty$, the censoring mechanism becomes
a threshold model where observations with $\exp(\theta_{0}^{T}x)<1$
are set to 0 and $\exp(\theta_{0}^{T}x)$ otherwise.}
\label{tab:check-moment}
\end{table}

\begin{figure}[ht]
\begin{centering}
\includegraphics[scale=0.35]{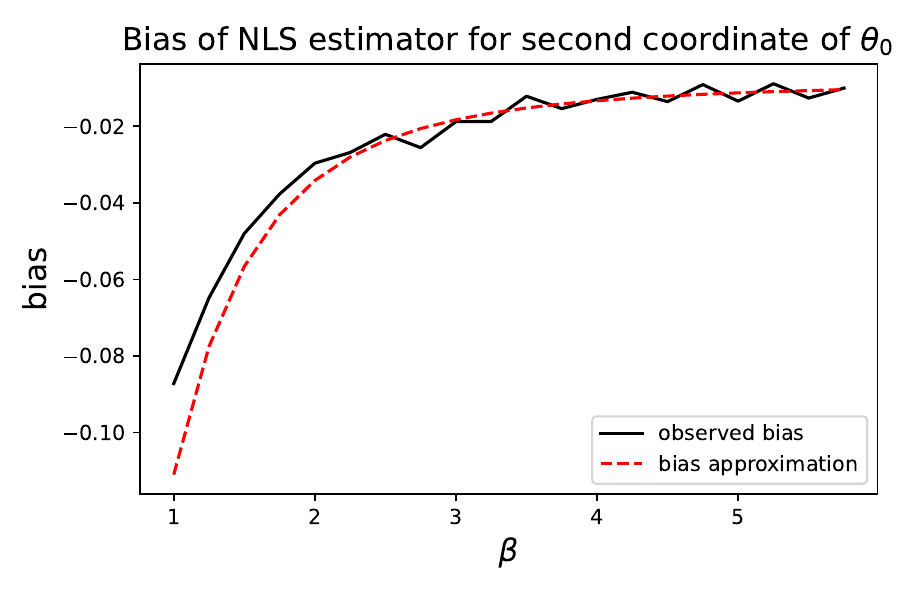}\includegraphics[scale=0.35]{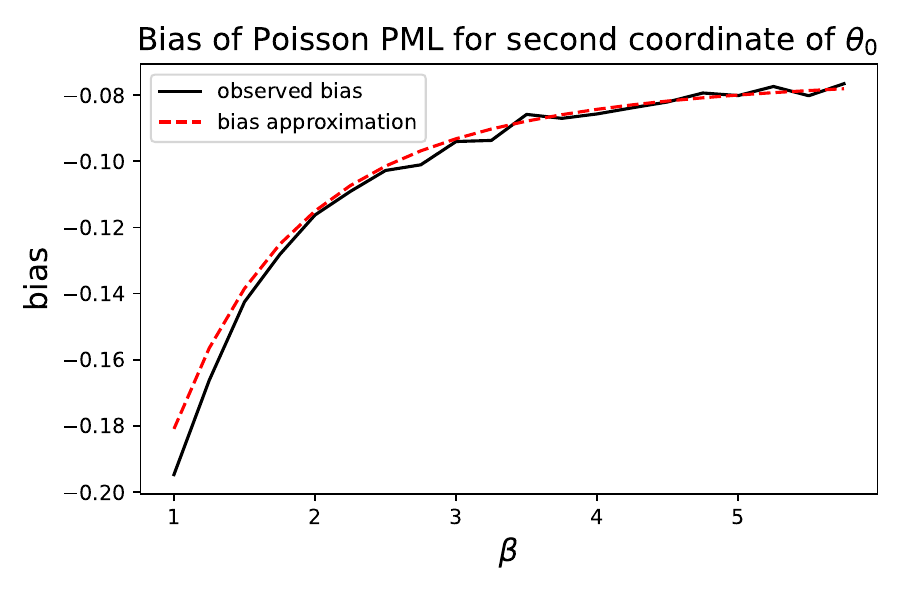}
\includegraphics[scale=0.35]{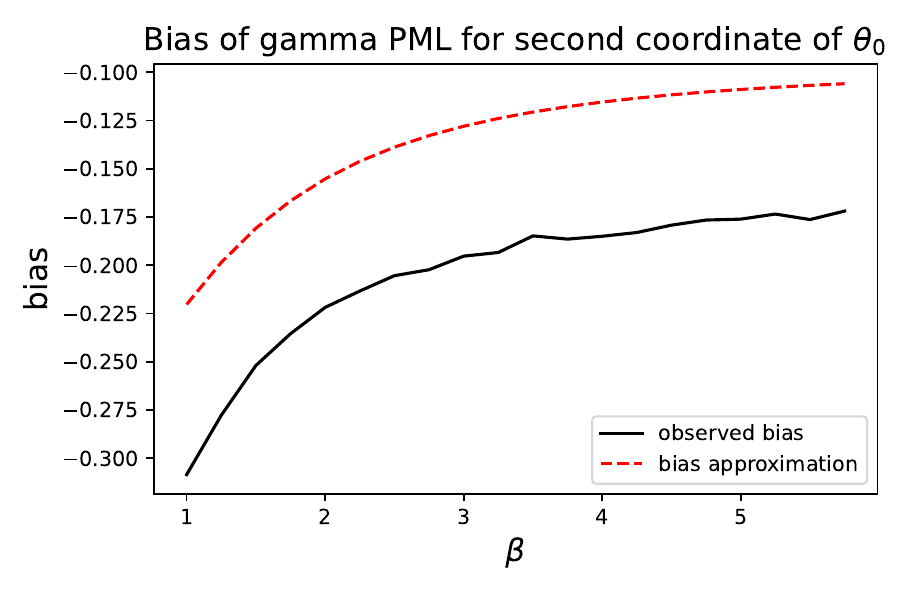}
\par\end{centering}
\caption{Comparison of simulated and analytic bias with censoring probability 
$\frac{1}{1+\exp(\beta\theta_{0}^{T}x)}$.}
\label{fig:bias-formula-check}
\end{figure}
Next, we investigate the accuracy of our approximation formula \eqref{eq:bias-formula} for
bias in \cref{fig:bias-formula-check}. We find that \eqref{eq:bias-formula} approximates the bias of NLS and Poisson PML well, particularly for $\beta>1$, but is less accurate for the gamma PML, which remains to be understood further.

\subsection{Assessing the Degree of Heteroskedasticity}

 \begin{figure}[t]
\begin{centering}
\includegraphics[scale=0.4]{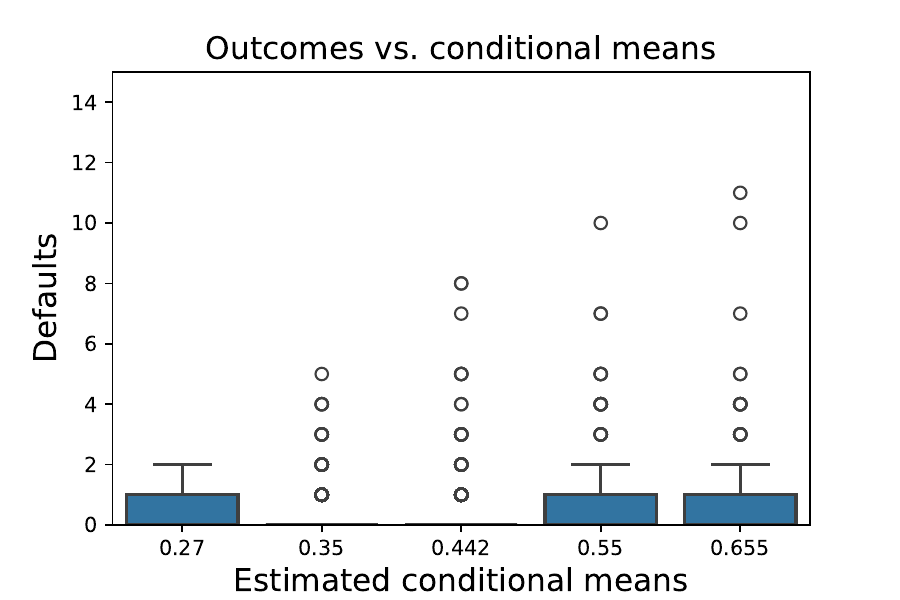}
\includegraphics[scale=0.4]{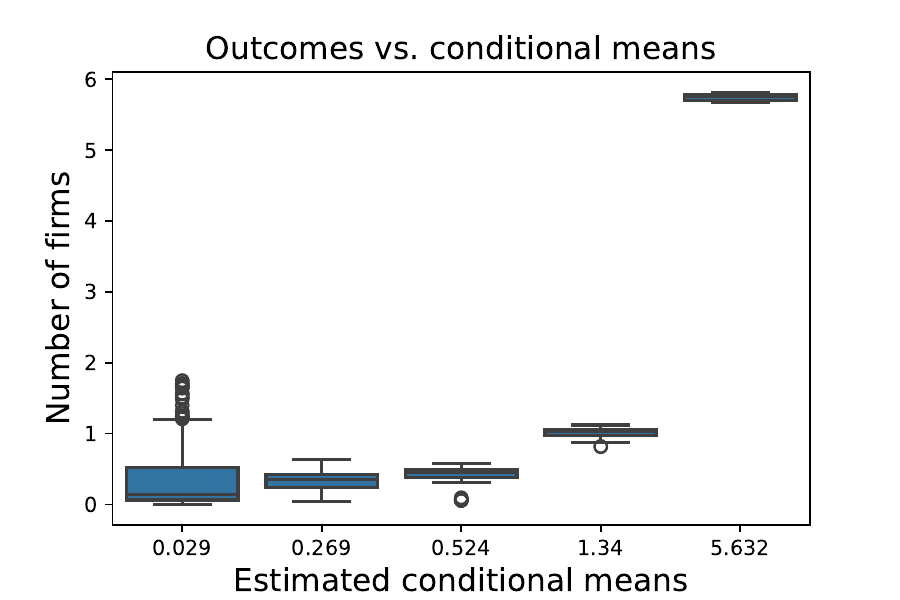}\\
\includegraphics[scale=0.4]{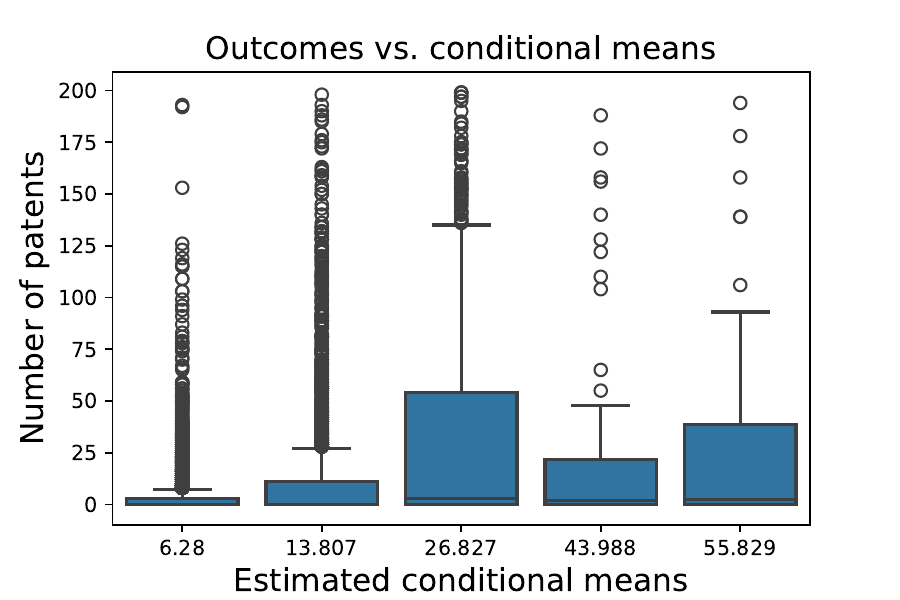}
\includegraphics[scale=0.4]{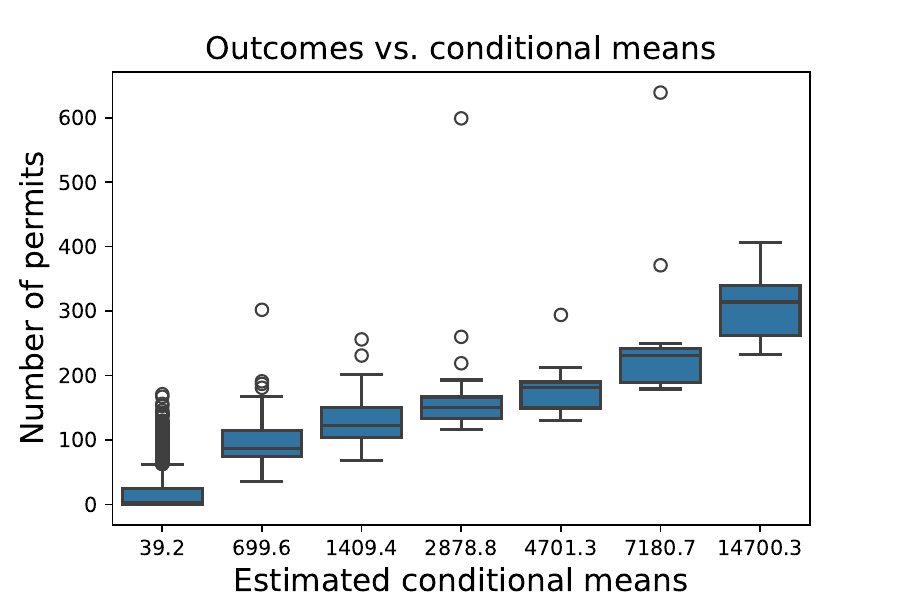}
\par\end{centering}
\caption{Box plots of outcomes vs. estimated conditional means for the default, credit rating, patent, and real estate data. For ratings and patents data, heteroskedasticity is relatively mild, which justifies the use of NLS ($\kappa=1$). On the other hand, heteroskedasticity is significant for real estate permits data, and moderate for defaults data, which explains the optimal choices of $\kappa=-1$ and $\kappa=0.1$, respectively.}
\label{fig:heteroskedasticity}
\end{figure}
 To assess the degree of heteroskedasticity, in \cref{fig:heteroskedasticity} we provide box plots of the outcome variable $y_i$ as a function of the conditional mean $\exp(\hat\theta^T x_i)$ predicted by the estimated models. We see that although the defaults data exhibits slightly higher sparsity, as illustrated by the histograms in \cref{fig:default-hist}, its degree of heteroskedasticity is considerably higher than that of the ratings data, which justifies the superior RMSE of estimators with $\kappa>0$ over that of the Poisson PML on the ratings data.

 
\section{Proofs}
\subsection{Proof of \cref{thm:bias}}
\begin{proof}
    First note that without censoring,
\begin{align*}
\mathbb{E}(y-\exp(\theta^{T}x))\cdot\exp(\kappa\theta^{T}x)\cdot x & =\mathbb{E}(\exp(\theta_{0}^{T}x+\epsilon)-\exp(\theta^{T}x))\cdot\exp(\kappa\theta^{T}x)\cdot x\\
 & =\mathbb{E}(\exp(\theta_{0}^{T}x)-\exp(\theta^{T}x))\cdot\exp(\kappa\theta^{T}x)\cdot x
\end{align*}
using $\mathbb{E}[\exp(\epsilon)\mid x]=1$. Therefore when $\theta=\theta_{0}$,
the above expectation is zero. 

On the other hand, under censoring with probability $P(\theta_{0},x)$:
\begin{align*}
Y_{i}^{true} & =\exp(\theta_{0}^{T}X_{i}+\epsilon_{i})\\
Y_{i} & =\begin{cases}
Y_{i}^{true} & 1-P(\theta_{0},X_{i})\\
0 & P(\theta_{0},X_{i}),
\end{cases}
\end{align*}
 we have 
\begin{align*}
\mathbb{E}(y-\exp(\theta^{T}x))\cdot\exp(\kappa\theta^{T}x)\cdot x & =\mathbb{E}_{x}\mathbb{E}\left[(y-\exp(\theta^{T}x))\mid x\right]\cdot\exp(\kappa\theta^{T}x)\cdot x
\end{align*}
 and 
\begin{align*}
\mathbb{E}\left[(y-\exp(\theta^{T}x))\mid x\right] & =(1-P(\theta_{0},x))\cdot(\exp(\theta_{0}^{T}x)-\exp(\theta^{T}x))+P(\theta_{0},x)(-\exp(\theta^{T}x))\\
 & =(1-P(\theta_{0},x))\cdot\exp(\theta_{0}^{T}x)-\exp(\theta^{T}x)
\end{align*}
 so that the moment conditions become 
\begin{align}
\mathbb{E}_{x}\left\{ (1-P(\theta_{0},x))\cdot\exp(\theta_{0}^{T}x)-\exp(\theta^{T}x)\right\} \cdot\exp(\kappa\theta^{T}x)\cdot x & =0,
\end{align}
and we denote the solution to these moment conditions $\tilde{\theta}_{0}$. By assumptions in the theorem, we can invoke Theorem 3.1 of \citet{newey1994large} to guarantee that the solution $\hat \theta$ to \eqref{eq:generalized-PML} converges in probability to $\tilde{\theta}_0$.

Barring explicit solutions, we can use Taylor expansion to estimate
bias as follows: expanding $\exp(\theta^{T}X)-\exp(\theta_{0}^{T}X)$
around $\theta_{0}$ as 
\begin{align*}
\exp(\theta^{T}X) & \approx\exp(\theta_{0}^{T}X)+\exp(\theta_{0}^{T}X)(\theta-\theta_{0})^{T}X\\
\exp(\kappa\theta^{T}X) & \approx\exp(\kappa\theta_{0}^{T}X)+\exp(\kappa\theta_{0}^{T}X)(\theta-\theta_{0})^{T}\kappa X
\end{align*}
we have 
\begin{align*}
\mathbb{E}_{x}\left\{ (1-P(\theta_{0},x))\cdot\exp(\theta_{0}^{T}x)-\exp(\theta^{T}x)\right\} \cdot\exp(\kappa\theta^{T}x)\cdot x & =\\
\int_{x}\left((1-P(\theta_{0},x))\cdot\exp(\theta_{0}^{T}x)-\exp(\theta^{T}x)\right)\cdot\exp(\kappa\theta^{T}x)\cdot xp(x)dx & \approx\\
\int_{x}\left((1-P(\theta_{0},x))\cdot\exp(\theta_{0}^{T}x)-\exp(\theta_{0}^{T}x)-\exp(\theta_{0}^{T}x)\cdot(\theta-\theta_{0})^{T}x\right) & \cdot\\
\left(\exp(\kappa\theta_{0}^{T}x)+\exp(\kappa\theta_{0}^{T}x)(\theta-\theta_{0})^{T}\kappa x\right)\cdot xp(x)dx & =\\
\int\left((1-P(\theta_{0},x))-1-(\theta-\theta_{0})^{T}x\right)\cdot\left(1+(\theta-\theta_{0})^{T}\kappa x\right)\cdot xp(x)\exp((\kappa+1)\theta_{0}^{T}x) & \approx\\
\int\left(-P(\theta_{0},x)-(\theta-\theta_{0})^{T}x+P(\theta_{0},x)\cdot(\theta-\theta_{0})^{T}\kappa x\right)\cdot xp(x)\exp((\kappa+1)\theta_{0}^{T}x)
\end{align*}
 which implies 
\begin{align*}
\int(P(\theta_{0},x)\kappa-1)(\theta-\theta_{0})^{T}x\cdot xp(x)\exp((\kappa+1)\theta_{0}^{T}x) & =\int P(\theta_{0},x)\cdot xp(x)\exp((\kappa+1)\theta_{0}^{T}x)
\end{align*}
 or 
\begin{align*}
(\theta-\theta_{0})^{T}\cdot\int(P(\theta_{0},x)\kappa-1)xx^{T}p(x)\exp((\kappa+1)\theta_{0}^{T}x) & =\int P(\theta_{0},x)\cdot x^{T}p(x)\exp((\kappa+1)\theta_{0}^{T}x)
\end{align*}
 so that 
\begin{align*}
\tilde \theta_0-\theta_{0} & =(A^{T}A)^{-1}Ab
\end{align*}
 where 
\begin{align*}
A= & \int(P(\theta_{0},x)\kappa-1)xx^{T}p(x)\exp((\kappa+1)\theta_{0}^{T}x)\\
b & =\int P(\theta_{0},x)\cdot xp(x)\exp((\kappa+1)\theta_{0}^{T}x).
\end{align*}
\end{proof}

\subsection{Proof of \cref{thm:asymptotic-variance}}
\begin{proof}
Recall that the generalized PML estimators $\hat \theta$ (omitting dependence on $\kappa$) defined in \eqref{eq:generalized-PML} converge to the limit $\tilde{\theta}_{0}$, which solves 
\begin{align*}
\mathbb{E}_{x}\left\{ (1-P(\theta_{0},x))\cdot\exp(\theta_{0}^{T}x)-\exp(\theta^{T}x)\right\} \cdot\exp(\kappa\theta^{T}x)\cdot x & =0.
\end{align*}
 Under the assumptions in \cref{thm:asymptotic-variance}, the asymptotic variance of the $Z$-estimators that solve 
 \begin{align*}
     \frac{1}{n}\sum_i \psi(y_i,x_i,\theta)=0
 \end{align*}
is given by $J^{-1}IJ^{-1}$, where
\begin{align*}
I & =\mathbb{E}_{x}\mathbb{E}\left[\psi(y,x,\tilde{\theta}_{0})\psi(y,x,\tilde{\theta}_{0})^{T}\mid x\right]\\
J & =\mathbb{E}_{x}\mathbb{E}\left[-\frac{\partial\psi(y,x,\tilde{\theta}_{0})}{\partial\theta}\mid x\right].
\end{align*}

We have, using $V[\exp(\varepsilon_{i})\mid X_{i}]=\exp^{\alpha}(\theta_{0}^{T}X_{i})=\exp(\alpha\theta_{0}^{T}X_{i})$,
\begin{align}
I & =\mathbb{E}_{x}\mathbb{E}\left[\left((y-\exp(\theta^{T}x))\cdot\exp(\kappa\theta^{T}x)\cdot x\right)\left((y-\exp(\theta^{T}x))\cdot\exp(\kappa\theta^{T}x)\cdot x\right)^{T}\mid x\right]\nonumber \\
 & =\mathbb{E}_{x}xx^{T}\exp(2\kappa\theta^{T}x)\cdot\mathbb{E}\left[\left(y-\exp(\theta^{T}x)\right)^{2}\mid x\right]\nonumber \\
 & =\mathbb{E}_{x}xx^{T}\exp(2\kappa\theta^{T}x)\cdot\mathbb{E}\left[P(\theta_{0},x)\left(\exp(\theta^{T}x)\right)^{2}+(1-P(\theta_{0},x))\left(\exp(\theta_{0}^{T}x+\varepsilon)-\exp(\theta^{T}x)\right)^{2}\mid x\right]\nonumber \\
 & =\mathbb{E}_{x}xx^{T}\exp(2\kappa\theta^{T}x)\cdot\mathbb{E}\left[P(\theta_{0},x)(\exp(\theta^{T}x))^{2}+(1-P(\theta_{0},x))((\exp(\theta_{0}^{T}x)-\exp(\theta^{T}x))^{2}+\exp(\alpha\theta_{0}^{T}x))\mid x\right],
\end{align}
 and 
\begin{align*}
J & =\mathbb{E}_{x}\mathbb{E}\left[-(y-\exp(\theta^{T}x))\cdot\exp(\kappa\theta^{T}x)\cdot\kappa xx^{T}+\exp(\kappa\theta^{T}x)\exp(\theta^{T}x)xx^{T}\mid x\right]\\
 & =\mathbb{E}_{x}\mathbb{E}\left[-\left((1-P(\theta_{0},x))\cdot\exp(\theta_{0}^{T}x)-\exp(\theta^{T}x)\right)\cdot\exp(\kappa\theta^{T}x)\cdot\kappa xx^{T}+\exp(\kappa\theta^{T}x)\exp(\theta^{T}x)xx^{T}\mid x\right].
\end{align*}
Combining these calculations, we have 
\begin{align*}
    \sqrt{n}(\hat{\theta}-\tilde{\theta}_{0}) \rightarrow_d \mathcal{N}(0,J^{-1}IJ^{-1}),
\end{align*}
where
\begin{align*}
I & =\mathbb{E}_{x}xx^{T}\exp(2\kappa\theta^{T}x)\cdot\mathbb{E}\left[P(\theta_{0},x)(\exp(\theta^{T}x))^{2}+(1-P(\theta_{0},x))((\exp(\theta_{0}^{T}x)-\exp(\theta^{T}x))^{2}+\exp(\alpha\theta_{0}^{T}x))\right]\\
J & =\mathbb{E}_{x}\left[-\left((1-P(\theta_{0},x))\cdot\exp(\theta_{0}^{T}x)-\exp(\theta^{T}x)\right)\cdot\exp(\kappa\theta^{T}x)\cdot\kappa xx^{T}+\exp(\kappa\theta^{T}x)\exp(\theta^{T}x)xx^{T}\right],
\end{align*}
evaluated at $\theta=\tilde{\theta}_{0}$.
\end{proof}
\zq{
\subsection{Proof of \cref{thm:general-duality-PML}}
\begin{proof}
We provide the proof for the case when $\phi(y)\equiv y$. The Lagrangian
form of the constrained problem is equivalent to 
\begin{align*}
\max_{\theta}\min_{\mu_{i}}\gamma^{\ast}(\mu_{i})+\theta^{T}(\sum_{i}Y_{i}X_{i}-\sum_{i}\mu_{i}X_{i}).
\end{align*}
 Now using the fact that $\gamma$ is convex and lsc, 
\begin{align*}
\min_{\mu_{i}}\gamma^{\ast}(\mu_{i})-\mu_{i}\theta^{T}X_{i} & =-\gamma(\theta^{T}X_{i}).
\end{align*}
 It then follows that the problem is equivalent to 
\begin{align*}
\max_{\theta}\sum_{i}Y_{i}\theta^{T}X_{i}-\gamma(\theta^{T}X_{i})=\max_{\theta}g(f(\theta^{T}X_{i}))\phi(Y_{i})-\gamma(f(\theta^{T}X_{i}))
\end{align*}
 if the conditional mean model $f$ is the inverse of the canonical
link function, i.e., $f=g^{-1}$. 
\end{proof}}

The variance of various pseudo ML estimators without censoring has been characterized by \citet{gourieroux1984pseudob}. Here we leverage their result to compute the variance of PPML and NLS.
\begin{lem}
\label{lem:variance-PPML-NLS}
The asymptotic variance of PPML and NLS are given by 
\begin{align*}
V_{NLS} & =\mathbb{E}_{x}\left[\exp(2\theta_{0}^{T}x)\right]^{-1}\cdot\left(\mathbb{E}_{x}(x\exp(\theta_{0}^{T}x))(x\exp(\theta_{0}^{T}x))^{T}\cdot V\left[y\mid x\right]\right)\cdot\mathbb{E}_{x}\left[\exp(2\theta_{0}^{T}x)\right]^{-1}\\
V_{PPML} & =\mathbb{E}_{x}\left[xx^{T}\exp(\theta_{0}^{T}x)\right]^{-1}\cdot\left(\mathbb{E}_{x}xx^{T}V\left[y\mid x\right]\right)\cdot\mathbb{E}_{x}\left[xx^{T}\exp(\theta_{0}^{T}x)\right]^{-1}
\end{align*}
\end{lem}
\begin{proof}
The general formula for the asymptotic variance of PML estimators
is given by $J^{-1}IJ^{-1}$, where
\begin{align*}
J & =\mathbb{E}_{x}\mathbb{E}\left[-\frac{\partial^{2}\varphi(y,x,\theta_{0})}{\partial\theta\partial\theta^{T}}\mid x\right]\\
I & =\mathbb{E}_{x}\mathbb{E}\left[\frac{\partial\varphi(y,x,\theta_{0})}{\partial\theta}\frac{\partial\varphi(y,x,\theta_{0})}{\partial\theta}^{T}\mid x\right].
\end{align*}
 Recall that for NLS, 
\begin{align*}
\varphi_{NLS}(y,x,\theta_{0}) & =(\exp(\theta_{0}^{T}x)-y)^{2},
\end{align*}
 so that 
\begin{align*}
J_{NLS} & =2\mathbb{E}_{x}\mathbb{E}\left[-x^{2}\exp(2\theta_{0}^{T}x)-(\exp(\theta_{0}^{T}x)+x^{2}\exp(\theta_{0}^{T}x))(\exp(\theta_{0}^{T}x)-y)\mid x\right]\\
 & =2\mathbb{E}_{x}\mathbb{E}\left[-\exp(2\theta_{0}^{T}x)+\exp(\theta_{0}^{T}x)y-x^{2}y\exp(\theta_{0}^{T}x)\mid x\right]\\
 & =2\mathbb{E}_{x}\left[-\exp(2\theta_{0}^{T}x)\right]\\
I_{NLS} & =4\mathbb{E}_{x}\mathbb{E}\left[\left(x\exp(\theta_{0}^{T}x)(\exp(\theta_{0}^{T}x)-y)\right)\cdot\left(x\exp(\theta_{0}^{T}x)(\exp(\theta_{0}^{T}x)-y)\right)^{T}\mid x\right]\\
 & =4\mathbb{E}_{x}(x\exp(\theta_{0}^{T}x))(x\exp(\theta_{0}^{T}x))^{T}\mathbb{E}\left[\left((\exp(\theta_{0}^{T}x)-y)\right)^{2}\mid x\right]\\
 & =4\mathbb{E}_{x}(x\exp(\theta_{0}^{T}x))(x\exp(\theta_{0}^{T}x))^{T}\mathbb{E}\left[\epsilon^{2}\mid x\right]\\
 & =4\mathbb{E}_{x}(x\exp(\theta_{0}^{T}x))(x\exp(\theta_{0}^{T}x))^{T}\cdot V\left[y\mid x\right]
\end{align*}
 and the asymptotic variance is given by 
\begin{align*}
\mathbb{E}_{x}\left[\exp(2\theta_{0}^{T}x)\right]^{-1}\cdot\left(\mathbb{E}_{x}(x\exp(\theta_{0}^{T}x))(x\exp(\theta_{0}^{T}x))^{T}\cdot V\left[y\mid x\right]\right)\cdot\mathbb{E}_{x}\left[\exp(2\theta_{0}^{T}x)\right]^{-1}.
\end{align*}

Similarly, for PPML, we have 
\begin{align*}
\varphi_{PPML}(y,x,\theta_{0}) & =y\cdot\theta_{0}^{T}x-\exp(\theta_{0}^{T}x),
\end{align*}
 so that 
\begin{align*}
J_{PPML} & =\mathbb{E}_{x}\left[xx^{T}\exp(\theta_{0}^{T}x)\right]\\
I_{PPML} & =\mathbb{E}_{x}\mathbb{E}\left[(y-\exp(\theta_{0}^{T}x))^{2}\cdot xx^{T}\mid x\right]\\
 & =\mathbb{E}_{x}xx^{T}\mathbb{E}\left[(\epsilon)^{2}\cdot\mid x\right]\\
 & =\mathbb{E}_{x}xx^{T}V\left[y\mid x\right]
\end{align*}
 and the asymptotic variance is given by 
\begin{align*}
J^{-1}IJ^{-1} & =\mathbb{E}_{x}\left[xx^{T}\exp(\theta_{0}^{T}x)\right]^{-1}\cdot\left(\mathbb{E}_{x}xx^{T}V\left[y\mid x\right]\right)\cdot\mathbb{E}_{x}\left[xx^{T}\exp(\theta_{0}^{T}x)\right]^{-1}.
\end{align*}
\end{proof}
\end{appendix}

\end{document}